%% file: main.tex
\documentclass[11pt]{article}

\input{./preamble}

\usepackage{tocloft}

\title{The Existence of Maximin Share Allocations is $\Delta_2^P$-complete}
\date{}
\author{
    Sushmita Gupta 
    \thanks{The Institute of Mathematical Sciences, a CI of Homi Bhabha National Institute, Chennai 600113, Tamil Nadu, India. {\tt sushmitagupta@imsc.res.in}} \and 
    Sanjay Seetharaman
    \thanks{The Institute of Mathematical Sciences, a CI of Homi Bhabha National Institute, Chennai 600113, Tamil Nadu, India. {\tt sanjays@imsc.res.in}}
}

\usepackage[natbib,style=alphabetic,maxnames=50]{biblatex}

\begin{document}
\maketitle

\begin{abstract}
The maximin share (MMS) guarantee has become one of the central fairness benchmarks for allocating indivisible items.
Since Kurokawa, Procaccia, and Wang [EC'14, JACM'18] showed that exact MMS allocations need not exist, a substantial literature has developed around the existence and computation of approximate MMS allocations, including the recent work of Heidari, Kaviani, Seddighin, and Shahrezaei [SODA'26].
In contrast, a fundamental complexity question posed more than a decade ago by Bouveret and Lemaître [JAAMAS'16] has remained unresolved: how hard is it to decide whether an exact MMS allocation exists?

For additive valuations, Lonc and Truszczynski [JAIR'20] established membership in the complexity class $\DeltaP$ (also known as $\P^\NP$), but no hardness result was known.
For the more general class of 2-additive valuations, Bouveret and Lemaître established NP-hardness, leaving a substantial gap to the $\DeltaP$ upper bound.
Moreover, the complexity of MMS existence for additive and $k$-additive valuations was posed as open questions.

We resolve this longstanding question by proving that deciding the existence of an MMS allocation for additive goods is $\Delta_2^P$-complete, even for five agents.

Additionally, we prove weak \coNP-hardness for three agents, thereby establishing a precise dichotomy with the known existence guarantee for two agents, and strong \coNP-hardness when the number of agents is unrestricted.
The latter implies that the hardness is not caused by large binary-encoded numbers; in fact, the problem remains hard even when all numerical values are given in unary.
Furthermore, our strong hardness construction yields an inverse-polynomial gap in the optimal MMS approximation ratio, ruling out an FPTAS for approximating this ratio unless $\P=\NP$.
Finally, we show that all these results for goods extend to the chores setting through a polynomial-time transformation that preserves MMS existence.
\end{abstract}

\thispagestyle{empty}
\newpage
\setcounter{tocdepth}{2}
\tableofcontents
\thispagestyle{empty}
\newpage
\setcounter{page}{1}

\newpage


\input{./1_frontmatter/introduction}





\input{./1_frontmatter/technical-overview}
\input{./1_frontmatter/related-work}
\input{./1_frontmatter/prelim}
\input{./2_no_instances/no-instances}

\input{./3_np/np-hardness-goods}
\input{./5_delta2p/delta2p-hardness-additive-goods}
\input{./7_conp/conp-hardness-goods}
\input{./8_goods_to_chores/goods_to_chores}

\input{./9_backmatter/new_conclusion}

\section*{Remarks and Acknowledgments}

An earlier version of this work established \DP-hardness for additive valuations and
$\Delta_2^P$-completeness for $2$-additive valuations.
We would like to thank an anonymous reviewer whose question about pursuing higher-level
hardness within the same framework motivated us to investigate this direction further.

\section*{Declaration of Generative AI Use}
We used ChatGPT 5.5/5.6/6 Plus to assist with language editing, typesetting/formatting, and checking/simplifying calculations.
We additionally used ChatGPT 6 Plus for technical assistance in some parts of \Cref{sec:delta2p-additive} and in preparing a preliminary draft of the section. 
More details can be found at the end of the section.

All statements, proofs, and technical results in the paper have been checked by the authors, who take full responsibility for their correctness and for any errors.

\printbibliography

\clearpage
\appendix
\input{./9_backmatter/appendix}

\end{document}

%% file: preamble.tex
\def\ShowComments{false} 

\def\AlgoStyle{algorithm2e} 

\usepackage[margin=1in]{geometry}
\usepackage[T1]{fontenc}
\usepackage[utf8]{inputenc}
\usepackage{charter}
\usepackage[english]{babel}
\usepackage{setspace}
\usepackage{multicol}
\usepackage[inline]{enumitem}
\usepackage{csquotes}
\usepackage{tabularx}
\usepackage{array}
\usepackage{booktabs}

\usepackage{amsmath, amssymb, amstext, amsthm}
\usepackage{mathtools}
\usepackage{mathrsfs}
\usepackage{nicefrac}
\usepackage{mleftright}
\usepackage{thmtools}
\usepackage{thm-restate}
\DeclareMathAlphabet{\mathbbold}{U}{bbold}{m}{n}

\usepackage{graphicx}
\usepackage[dvipsnames, svgnames, x11names]{xcolor}
\usepackage{tcolorbox}
\usepackage{soul} 

\usepackage{ifthen} 
\usepackage{hyperref}
\usepackage{cleveref} 
\usepackage{refcount}

\ifthenelse{\equal{\AlgoStyle}{algorithm2e}}{
    \usepackage[noend,ruled,linesnumbered]{algorithm2e}
    
    \SetAlFnt{\small} \SetAlCapFnt{\small} \SetAlCapNameFnt{\small}
    \SetAlCapHSkip{0pt} \IncMargin{-\parindent}
}{}
\ifthenelse{\equal{\AlgoStyle}{algorithmic}}{
    \usepackage{algorithm}
    \usepackage{algpseudocode}

}{}

\usepackage{todonotes}
\newcommand{\hide}[1]{} 
\newcommand{\shortv}[1]{} 

\ifthenelse{\equal{\ShowComments}{true}}{
    \newcommand{\il}[1]{\todo[inline, color=yellow!40]{#1}}

    \newcommand{\Ma}[1]{\textcolor{magenta}{#1}}

}{
    \presetkeys{todonotes}{disable}{}

    \newcommand{\il}[1]{\hide{#1}}

    \newcommand{\Ma}[1]{#1}

}

\declaretheorem[numberwithin=section]{theorem}

\declaretheorem[numberlike=theorem,style=definition]{definition}
\declaretheorem[numberlike=theorem]{claim,observation,corollary,proposition,lemma}

\declaretheorem[numberwithin=theorem,name=Claim]{claim-inside-theorem}
\declaretheorem[numberlike=claim-inside-theorem,name=Lemma]{lemma-inside-theorem}
\declaretheorem[numberlike=claim-inside-theorem,name=Observation]{observation-inside-theorem}

\declaretheorem[numberwithin=lemma,name=Claim]{claim-inside-lemma}
\declaretheorem[numberlike=claim-inside-lemma,name=Corollary]{corollary-inside-lemma}
\declaretheorem[numberwithin=lemma,name=Lemma]{lemma-inside-lemma}
\declaretheorem[numberlike=claim-inside-lemma,name=Observation]{observation-inside-lemma}

\declaretheorem[numberwithin=lemma-inside-theorem,name=Claim]{claim-inside-lemma-inside-theorem}
\declaretheorem[numberwithin=lemma-inside-theorem,name=Lemma]{lemma-inside-lemma-inside-theorem}
\declaretheorem[numberwithin=lemma-inside-theorem,name=Observation]{observation-inside-lemma-inside-theorem}

\newcommand{\bigoh}{\mathcal{O}}

\newcommand{\poly}{\mathrm{poly}}

\newcommand{\app}{$\dagger$}

\renewcommand{\part}{part\xspace}

\newcommand{\NP}{\textsf{NP}\xspace}
\newcommand{\NPH}{\textsf{NP}-{\sf hard}\xspace}

\newcommand{\DP}{\ensuremath{D^P}\xspace}
\newcommand{\DeltaP}{\ensuremath{\Delta^P_2}\xspace}

\newcommand{\Class}[1]{\ensuremath{{#1}_2^{P}}}
\newcommand{\Sanjay}[2]{\ensuremath{{#1}_{#2}^{P}}}

\newcommand{\threewaypartition}{\textsc{3-way Partition}\xspace}
\newcommand{\threepartition}{\textsc{3-Partition}\xspace}
\newcommand{\partition}{\textsc{Partition}\xspace}

\DeclareMathSymbol{\qm}{\mathalpha}{operators}{"3F}

\newcommand{\mms}{\mathrm{MMS}\xspace}
\newcommand{\coNP}{\ensuremath{\text{co-}\textsf{NP}}\xspace}
\newcommand{\coNPH}{\ensuremath{\text{co-}\textsf{NP}\text{-hard}}\xspace}

\addto\extrasenglish{}

\newcommand{\KPW}{KPW\xspace}

\newcommand{\Q}{\textcolor{black}{row-forcing}\xspace}

\newcommand{\rowsum}{\ensuremath{\alpha}\xspace}

\newcommand{\fourpartition}{\textsc{4-Partition}\xspace}

\renewcommand{\P}{\textsf{P}\xspace}

\newcommand{\lexmaxsat}{\textsc{LexMaxSAT}\xspace}
\newcommand{\lexmaxsatpromise}{\textsc{LexMaxSAT-Promise}\xspace}

\newcommand{\lexmaxthreesatpromise}{\textsc{LexMax3SAT-Promise}\xspace}

\newcommand{\lexmaxfourpartitionpromise}{\textsc{LexMax4Partition-Promise}\xspace}

\newcommand{\incmat}{\operatorname{inc-mat}\xspace}
\newcommand{\inc}{\operatorname{inc}\xspace}
\newcommand{\solS}{\mathcal{S}\xspace}

\newcommand{\decision}{decision problem\xspace}
\newcommand{\search}{search problem\xspace}

\usepackage{chngcntr}
\counterwithin{equation}{section}

\newtcolorbox{notebox}[1]{
  colback=teal!4,
  colframe=teal!55!black,
  title=\textbf{#1}
}

\newcommand{\score}{\operatorname{score}\xspace}

\Crefname{observation}{Observation}{Observations}
\Crefname{claim-inside-theorem}{Claim}{Claims}
\Crefname{lemma-inside-theorem}{Lemma}{Lemmas}
\Crefname{claim-inside-lemma}{Claim}{Claims}
\Crefname{corollary}{Corollary}{Corollaries}

%% file: 1_frontmatter/introduction.tex
\section{Introduction}

Fair division concerns the allocation of resources among agents with possibly different preferences.
A central difficulty in the allocation of indivisible goods is that classical fairness notions such as proportionality and envy-freeness may be unattainable.
The \emph{maximin share} (MMS) guarantee, introduced by \citet{DBLP:conf/bqgt/Budish10}, provides a natural relaxation inspired by the familiar ``cut-and-choose'' principle. 

Consider an instance with a set $N$ of $n$ agents and a set $M$ of indivisible items.
Each agent $i\in N$ has a valuation function
$v_i\colon 2^M\to \mathbb{R}$.
The {\it maximin share} of agent $i$ is
\[
\mms_i
=
\max_{(P_1, \ldots, P_n) \in \Pi_n(M)}
\min_{j \in [n]} v_i(P_j),
\]
where $\Pi_n(M)$ denotes the set of partitions of $M$ into $n$ bundles.
Thus, $\mms_i$ is the largest value that agent $i$ can guarantee by partitioning the items into $n$ bundles and receiving the least valuable bundle according to her own valuation.

We consider both the \emph{goods} setting, where valuation functions are nonnegative, and the \emph{chores} setting, where valuation functions are nonpositive.\footnote{Throughout this paper, the term additive valuations encompasses both additive goods and additive chores, unless otherwise specified.}
We point the reader to \cite{DBLP:journals/ai/AmanatidisABFLMVW23,DBLP:journals/jair/LiuLSW24} for an expansive discussion of these and other related concepts.

An allocation $(A_1,\ldots,A_n)$ is an \emph{MMS allocation} if
\[
v_i(A_i) \geq \mms_i
\qquad\text{for every } i\in N.
\]

The computational complexity of deciding the existence of an MMS allocation for goods was first raised over a decade ago by \citet{DBLP:conf/atal/BouveretL14}, and is now an established question with a large body of work.
They showed that the {\it \decision}, i.e., deciding whether an MMS allocation exists, is NP-hard for 2-additive valuation functions (defined formally in \Cref{sec:prelim}), a strict superclass of additive valuations.
They also showed that computing the maximin share of an agent is itself \NP-hard, even for additive valuations with two (identical) agents.
Tantalizingly, however, the computational complexity of the \decision for additive valuations was left open, as was the question of identifying {\it any} valuation class for which a complete complexity classification could be obtained.

In this paper, we resolve this question by showing that the \decision is $\Sanjay{\Delta}{2}$-complete even for additive valuations. 
Furthermore, we show a transference result from goods to chores setting, that enables us to transfer a hardness result about additive goods to additive chores. 



Contemporaneously with \cite{DBLP:conf/atal/BouveretL14}, \citet*{DBLP:conf/sigecom/ProcacciaW14}, via an explicit no-instance, showed that an MMS allocation need not always exist, even when all valuation functions are additive.
Subsequently, \citet{DBLP:conf/wine/FeigeST21} showed a ``stronger'' no-instance with an exponentially larger gap between the value an agent receives in a partition and her maximin share.
These results naturally motivate the search for approximation algorithms where each agent receives items of value at least an $\alpha$-fraction of her respective maximin share, for some $\alpha\in(0,1)$.
Unsurprisingly, a substantial body of work has emerged around improving the approximation guarantees for MMS allocations, both for additive valuations~\cite{DBLP:journals/jacm/KurokawaPW18,DBLP:conf/soda/AkramiG24,DBLP:journals/corr/abs-2511-13056,DBLP:conf/soda/HeidariKSS26} and more generally for  submodular and subadditive valuation classes~\cite{DBLP:conf/sigecom/GhodsiHSSY18,DBLP:conf/sigecom/SeddighinS25,DBLP:journals/corr/abs-2605-08859,DBLP:journals/corr/abs-2303-12444}.

\hide{\il{Suppressing: May include in Related Work}
At the time of writing, the best approximation guarantee attainable in polynomial time is $7/9$ shown by \citet*{DBLP:journals/corr/abs-2511-13056}, an improvement over $10/13$, given by \citet*{DBLP:conf/soda/HeidariKSS26} last year. For valuation classes beyond additive, such as subadditive, the best approx guarantee is $\nicefrac{1}{\bigoh((\log \log n)^2)}$ given by \citet*{DBLP:conf/sigecom/SeddighinS25}; for submodular it is 10/27 given by \citet{DBLP:journals/corr/abs-2303-12444}. }

In contrast, the computational complexity of the decision question of MMS allocation is less understood.
While it was posed by \citet{DBLP:conf/atal/BouveretL14} over a decade ago, it has eluded precise characterization.
While that paper showed NP-hardness for the 2-additive, it only established that the existence problem belongs to the class $\Sanjay{\Sigma}{2}$, a superclass of $\Sanjay{\Delta}{2}$.
Thus, a gap remained between the hardness and the membership of the \decision, even for the class of 2-additive valuations.
The open status of the additive case was also explicitly highlighted by \citet*{DBLP:journals/aamas/HeinenNNR18}, who stated that the complexity of deciding whether a maximin share allocation exists was the only missing part in their study of maximin, proportional, and minimax share allocations.
Later, \citet{DBLP:journals/jair/TruszczynskiL20} showed that the additive version belongs to $\Class{\Delta}$; their argument extending to 2-additive valuations as well.

\citet*{DBLP:conf/aaai/AzizRSW17} studied the {\it \search}, i.e., computing an MMS allocation if one exists, under additive valuations, and showed strong NP-hardness for an arbitrary number of agents and weak NP-hardness when there are only two agents.
We note, however, that this reduction does not immediately imply anything about the complexity of the \decision.
In fact, their reduction uses agents with identical valuations, and thus an MMS allocation always exists.

One might hope to relate the \search and the \decision through a standard search-to-decision reduction, where an MMS-decision oracle answers whether the current instance admits an allocation satisfying the maximin shares of that instance.
We note that if the maximin share values are supplied explicitly as part of the input to a decision oracle, which asks whether there exists an allocation meeting each of these thresholds, the usual search-to-decision self-reduction holds by fixing item assignments one at a time.
This argument, however, does not apply directly to the MMS \decision, because modifying the instance changes the maximin shares themselves.
Hence, the standard self-reducibility argument for recovering an MMS allocation from an MMS-existence oracle does not go through.
In general, such a self-reduction is not currently known.

Conversely, computational hardness of the \decision would imply hardness of the \search.
For general additive valuations, however, the computational hardness of the \decision had remained completely unresolved; prior to our work, even polynomial-time solvability had not been ruled out.

Both the decision and search problems are equally natural in the chores setting.
The study of MMS allocations for chores \hide{ where items have nonpositive values, The study of MMS allocations for chores} was initiated by \citet*{DBLP:conf/aaai/AzizRSW17}, who showed that an MMS allocation may not exist even for additive valuations and studied approximation guarantees.
Later, as in the goods setting, \citet*{DBLP:conf/wine/FeigeST21} strengthened the non-existence result by presenting an example exhibiting a larger gap.
As in the goods setting, these non-existence results have motivated a growing line of work on approximate MMS allocations for chores~\cite{DBLP:conf/aaai/AzizRSW17,DBLP:journals/teco/BarmanK20,DBLP:conf/sigecom/HuangL21,DBLP:conf/sigecom/HuangS23}.
These works further emphasize the role of MMS as a fairness benchmark across valuation models, while also highlighting the need to understand the computational complexity of exact MMS.



We note that the decision (or even the search) problem is computationally unusual.
Even {\it verifying} that a given allocation\hide{ of goods or chores} is an MMS allocation can be shown to be \coNP-complete (\Cref{prop:verifying-mms-conp-complete}). 
This rules out the straightforward NP-certificate approach of guessing an allocation and verifying it in polynomial time.
However, this computational hardness does not directly imply hardness of the \decision, since it may still admit a polynomial-time algorithm. 



Thus, the difficulty is not only in searching over allocations, but also in evaluating the fairness benchmark against which an allocation must be tested.
These features suggest that the computational complexity of the \decision (and thus the \search) for the additive setting may extend beyond ordinary NP-hardness.
In this paper, we resolve this matter.


We show that deciding whether an MMS allocation exists under additive valuations is $\Class{\Delta}$-complete, even for five agents.
The membership follows from the result of \citet{DBLP:journals/jair/TruszczynskiL20}, while our main contribution is a matching $\Delta_2^P$-hardness result.
Thus, we settle the computational complexity of the \decision for additive valuations, one of the stated objectives of \citet*{DBLP:journals/aamas/BouveretL16}.
More generally, our hardness result extends beyond additive valuations to compactly representable valuation classes that contain additive valuations.
By \Cref{prop:circuit-membership-delta2p}, membership in $\Class{\Delta}$ also holds for such classes when bundle values can be computed in polynomial time and have polynomial encoding length, yielding $\Class{\Delta}$-completeness.


The reductions presented here, especially the one showing \Class{\Delta}-hardness is quite technical, requiring us to first define a suitable source problem: \lexmaxfourpartitionpromise, a lexicographic variant of \fourpartition. This combined with precise and intricate gadgets work together to yield the desired reduction that establishes the hardness for mere five agents. 
To aid comprehension and enhance exposition, we first present the \NP-hardness result which introduces the basic structural properties that are further enhanced for the \Class{\Delta}-hardness result. 
In addition, we also present weak and strong \coNP-hardness results that work for an even fewer number of agents and agent types. These results exhibit sharp dichotomies with known results. Moreover, they imply hardness in the unary setting, as well as towards standard approximation goals.

Next, we discuss our results and their implications in further details; \Cref{tab:main-results} summarizes the main results at a glance.



\input{./1_frontmatter/contributions}


\paragraph{Structure of the paper.}
In \Cref{sec:prelim}, we cover the preliminaries, and in \Cref{sec:technical-overview}, we give a technical overview of the results.
In \Cref{sec:no-instances}, we present our analysis of the two no-instance constructions used throughout the paper that culminate in \Cref{prop:KPW-core-construction,prop:FST21-3-agent,prop:FST-core-no-instance}.
We first present the \NP-hardness result for additive goods in \Cref{sec:np}, which introduces some of the structural ideas used later.
We then prove our main hardness result, the $\Class{\Delta}$-hardness for additive goods, in \Cref{sec:delta2p-additive}.
The \NP-hardness construction serves as a simpler version of the later construction and helps isolate some of the main ideas underlying the final reduction.
In \Cref{sec:conp}, we prove the \coNP-hardness results for additive goods.
Finally, in \Cref{sec:goods-to-chores-transfer}, we show how to transfer the hardness results from goods to chores.


Missing details and proofs, marked with (\app), are in the appendix.

%% file: 1_frontmatter/contributions.tex
\subsection{Our contributions}\label{contributions}

\begin{table}[t]
\centering
\small
\begin{tabular}{@{}llll@{}}
\toprule
\textbf{Setting} & \textbf{Complexity} & \textbf{Restriction} & \textbf{Reference}\\
\midrule
Additive
    & $\Delta_2^P$-complete
    & Five agents & \Cref{thm:delta2p-additive-completeness}\\
Additive
    & weakly \coNP-hard
    & Three agents & \Cref{thm:weak-coNP-hardness}\\
Additive
    & strongly \coNP-hard
    & Two agent types & \Cref{thm:strong-coNP-hardness}\\
\bottomrule
\end{tabular}
\caption{Our results, for both goods and chores setting, at a glance}
\label{tab:main-results}
\end{table}



We defer all formal definitions to \Cref{sec:prelim}.
We start our discussion with the main result, which settles the computational complexity of the MMS \decision. 
\begin{restatable}{theorem}{deltatwocompletegoods}
\label{thm:delta2p-additive-completeness}
Deciding whether an additive goods instance admits an MMS allocation is $\Class{\Delta}$-complete, even for five agents.
\end{restatable}

As stated earlier in the introduction, the membership in $\Class{\Delta}$ follows from the result of \citet{DBLP:journals/jair/TruszczynskiL20}.
Our contribution is a matching hardness result.
The essential ingredients of this reduction are discussed in \Cref{sec:technical-overview}, followed by a fuller analysis in \Cref{sec:delta2p-additive}.

It is known from \cite{DBLP:journals/aamas/BouveretL16} that when there are only two agents, MMS allocations always exist for additive valuations.
Thus, the following result demonstrates a sharp dichotomy.
\begin{restatable}{theorem}{weakconphardness}
    \label{thm:weak-coNP-hardness}
    Deciding whether an additive goods instance admits an MMS allocation is weakly \coNPH, even for three agents.
\end{restatable}
The reduction is highly restricted: the agents agree on the values of all but nine goods.

We next establish strong \coNP-hardness even when there are only two agent types.
In particular, the hardness persists even when all numerical values are given in unary.
We say that two agents have the same {\it type} if they have identical valuation functions; thus, the number of agent types is the number of distinct valuation functions in the input.
With a single agent type, an MMS allocation always exists: we can simply allocate the bundles of an MMS-witnessing partition, which is common to all agents.
Thus, the following theorem establishes a second sharp dichotomy, this time with respect to the number of agent types.

\begin{restatable}{theorem}{strongconphardness}
    \label{thm:strong-coNP-hardness}
    Deciding whether an additive goods instance admits an MMS allocation is strongly \coNPH, even when there are only two agent types.
    \end{restatable}


The strong \coNP-hardness result also has an approximation consequence.
Approximation algorithms for MMS allocations usually search for allocations that achieve some guaranteed factor, called the \emph{approximation ratio}, of every agent's maximin share value, and therefore do not necessarily distinguish instances admitting an exact MMS allocation from those that do not.
The {\it optimal MMS ratio} of an instance is defined to be the maximum approximation ratio attainable for all agents \cite{DBLP:conf/aaai/AzizRSW17}.


The following result rules out an algorithm that runs in time polynomial in both the input size and $1/\varepsilon$, and returns a $(1-\varepsilon)$-approximation of the optimal MMS ratio.

\begin{restatable}{corollary}{nofptas}
    \label{cor:no-FPTAS-optimal-MMS-ratio}
    Unless \P = \NP, there is no FPTAS for computing the optimal MMS (approximation) ratio of an additive goods instance.
\end{restatable}

Finally, we show that the above results can be transferred, in polynomial time, from goods to chores.
The transfer preserves the existence of MMS allocations, and therefore gives the corresponding hardness results for chore instances.

\begin{restatable}[Informal]{theorem}{goodstochorestransfer}
\label{thm:goods-to-chores-transfer}
The hardness results above extend to chores.
In particular, deciding whether an additive chores instance admits an MMS allocation is $\Delta_2^P$-complete.
The restricted \coNP-hardness results for additive goods yield the corresponding hardness results for additive chores.
\end{restatable}


We end our discussion by noting that our hardness reductions depend on a very careful analysis of the no-instances of MMS allocations given by \citet{DBLP:journals/jacm/KurokawaPW18} and \citet*{DBLP:conf/wine/FeigeST21}, that includes making certain parameters explicit in order to work with our gadgets.
This culminates in certain structural properties of the KPW core listed in \Cref{prop:KPW-core-construction}.
This may be of independent interest beyond this work, enabling further analysis of the computational complexity of MMS.

%% file: 1_frontmatter/technical-overview.tex
\section{Technical overview}
\label{sec:technical-overview}

In this section we sketch the main ideas underlying our hardness proofs. 
At the heart of our reductions are rigid MMS no-instance cores given by \cite{DBLP:journals/jacm/KurokawaPW18} and \cite{DBLP:conf/wine/FeigeST21}, formally defined in \Cref{sec:no-instances}, along with gadgets encoding the respective source problems. 
Hence, we start our overview by discussing these cores.
Following that, we discuss the \NP-hardness and $\Delta_2^P$-hardness results for additive goods, and finally the \coNP-hardness constructions.
We conclude with a discussion of how the above hardness results for goods transfer to chores.


\subsection{No-instance cores}
A common ingredient in our reductions is an integer matrix, called a \emph{core}, from which the valuations are defined to construct an instance that does not admit MMS allocations.
The positive matrix entries are interpreted as goods. 
The core is rigid in the sense that any partition of goods that is MMS-fair corresponds to very structured partitions of the matrix.
We use two such cores.

The first is the construction of \citet{DBLP:journals/jacm/KurokawaPW18}; we call it the {\it KPW core}, and define it formally in \Cref{subsec:KPW-core}.
Their presentation involves parameters that are not specified explicitly.
For our purposes, it is important to make the construction parameters explicit and to clear denominators, so that the construction gives us a nonnegative integer matrix $A$ with polynomial bit complexity.
For both our \NP-hardness and $\Delta_2^P$-hardness reductions, we use the KPW core with $n=5$.
Thus, $A$ is a fixed $5\times 5$ matrix with $19$ positive entries.
The positive entries of $A$ are the \emph{core goods}.
Every row and every column of $A$ has the same sum, denoted by $\rowsum$.

The key property here is row-or-column rigidity: 
If the positive entries of $A$ are partitioned into five bundles and every bundle has value at least $\rowsum$, then the partition must be either the row or the column partition, up to relabeling.
Thus, we define the following special partitions, used throughout the article.

\begin{definition}[Row and Column Partitions]
A \emph{row partition} is an allocation in which each bundle consists of the core goods in one row of $A$; a \emph{column partition} is defined analogously. 
\end{definition}

They define two types of agents on top of this core by perturbing the positive entries of the core matrix in two different ways.
The maximin share of every agent is exactly $\rowsum$.
By the row-or-column rigidity, any allocation giving every agent value at least $\rowsum$ must be either the row partition or the column partition.
The valuations are chosen so that one type of agent makes the row partition infeasible, while the other type makes the column partition infeasible.
Since every MMS allocation would have to induce one of these two canonical partitions, and both are made infeasible by the two agent types, the instance does not admit an MMS allocation.

Our reductions use this core in a slightly different way, so as to capture the reduction's source problem.
We add additional goods to the instance, and create appropriate valuation functions for the agents.
We scale the KPW core sufficiently so that its row-or-column rigidity is preserved even after the additional goods and lower-order perturbations are introduced.
The row-forcing agents then rule out the column partition, ensuring that every MMS allocation must allocate the core goods by rows.
The remaining agents and additional goods encode the source problem.

For our strong \coNP-hardness result, we use a different core due to \citet{DBLP:conf/wine/FeigeST21}.
We will explain the corresponding structural facts later in \Cref{subsec:overview-conp}, in the overview of the \coNP-hardness reduction.

\subsection{\NP-hardness}
We reduce from \textsc{Partition}: we are given positive numbers $a_1,\dots,a_m$ with total sum $2T$, and the goal is to decide whether they can be partitioned into two subsets, each of sum exactly $T$.
Our construction uses only five agents, consisting of three puzzle-solving agents and two row-forcing agents.
There are also two kinds of goods: the KPW core goods, which are valued by all agents, and the puzzle goods corresponding to the \textsc{Partition} instance, which are valued positively only by the puzzle-solving agents.
All agents share the same scaled $5\times 5$ KPW core, but their valuations on the core goods differ by small additive perturbations, as shown in \Cref{fig:np-hardness-core-perturbations}.
The scaling parameter $\Gamma$ is chosen large enough so that the row-or-column rigidity of the KPW core is preserved even after the lower-order perturbations and values of the puzzle goods are added.
Consequently, any allocation that gives every agent value at least $\tau$ must allocate the core goods either as rows or as columns.

The role of the row-forcing agents is to rule out the column partition.
For these agents, every row has value exactly $\tau$, but every non-last column has value only $\tau - 1$, while only the last column has value above $\tau$.
Since there are two row-forcing agents and they do not value puzzle goods, a column allocation cannot satisfy both of them.
Thus, any MMS allocation must allocate the KPW core by rows.

Once the row structure is forced, the puzzle-solving agents encode the \textsc{Partition} instance.
For a puzzle-solving agent, every non-last row has value $\tau - T$, while the last row has value $\tau + 2T$.
Moreover, every puzzle-solving agent has MMS exactly $\tau$, as witnessed by the column partition with all puzzle goods added to the last column.
Thus, the puzzle-solving agent receiving the last row is automatically satisfied, but each puzzle-solving agent receiving a non-last row needs puzzle goods of total value at least $T$.
There are three puzzle-solving agents and only one last row, so at least two puzzle-solving agents must receive non-last rows.
This creates at least two deficits, each of value $T$.
Since the total value of all puzzle goods is exactly $2T$, an MMS allocation can exist only if the puzzle goods can be partitioned into two bundles of value exactly $T$.
Conversely, if such a partition exists, we give the last row to one puzzle-solving agent, two non-last rows together with the corresponding puzzle bundles to the remaining puzzle-solving agents, and the other two rows to the row-forcing agents.
This gives an MMS allocation.
Therefore, the constructed instance admits an MMS allocation if and only if the original \textsc{Partition} instance is a yes-instance.

\begin{figure}
    \centering
    \includegraphics[width=0.7\linewidth]{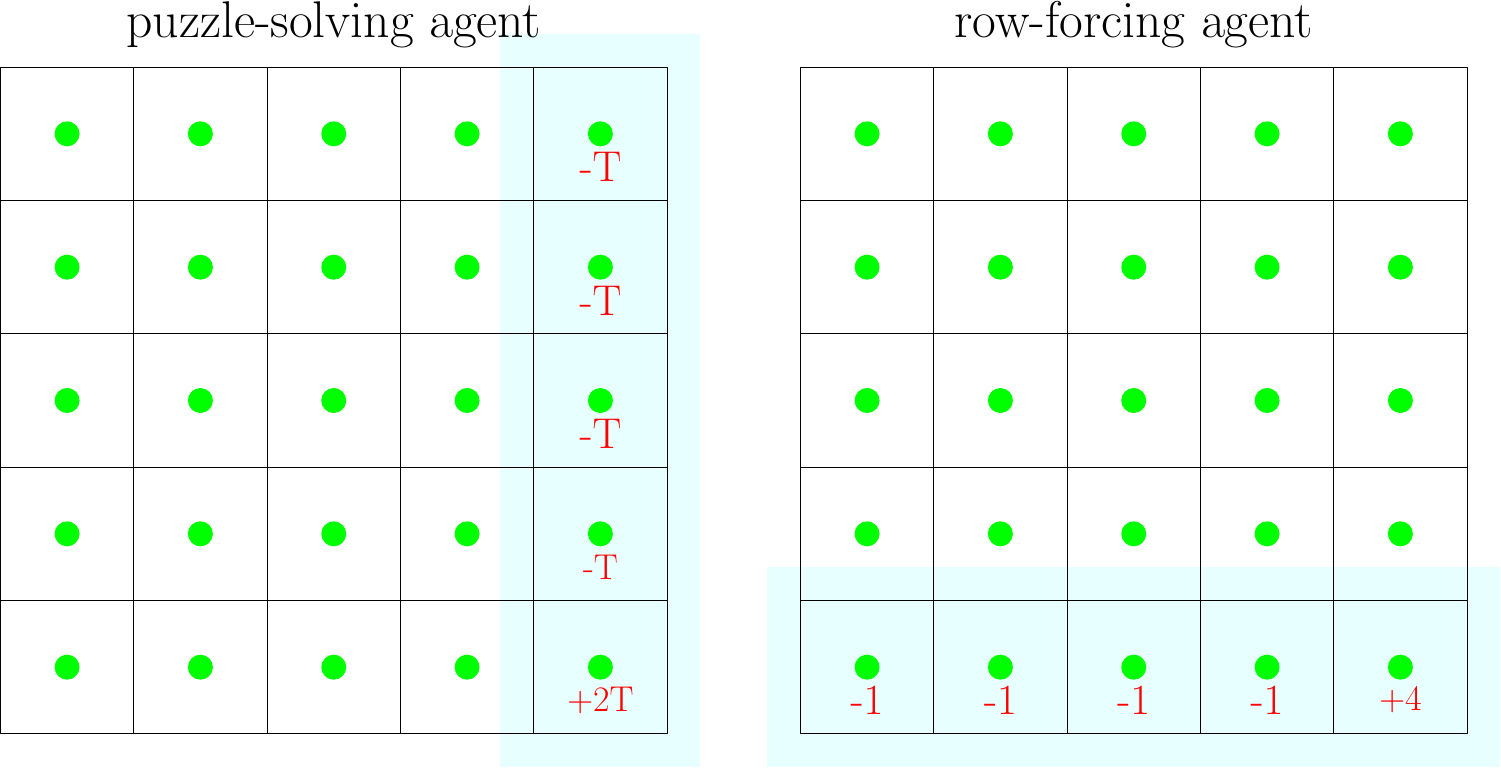}
    \caption{
    The figure shows a schematic diagram of the $5 \times 5$ KPW core.
    Each \textcolor{green}{green} dot represents the value inherited from the scaled KPW core, and each \textcolor{red}{red} value denotes the perturbation applied to that entry.
    Cells without a red label are unperturbed.
    For puzzle-solving agents, the first four rows are deficient by $T$ each, while the last row has a surplus of $2T$.
    Consequently, the last column has a total deficit of $2T$.
    For row-forcing agents, the first four columns are deficient by $1$ each, while the last column has a surplus of $4$.
    }
    \label{fig:np-hardness-core-perturbations}
\end{figure}

\subsection{$\Delta_2^P$-hardness}\label{subsec:technical-overview-delta2p}

Our main hardness result is obtained by reducing from a lexicographic version of \fourpartition.
The reduction has two steps.
First, we show that \lexmaxfourpartitionpromise is $\Delta_2^P$-hard by reducing from \lexmaxsat.
We then reduce \lexmaxfourpartitionpromise to MMS allocation existence.

Consider a source instance with $m = 4\ell$ numbers $a_1,\ldots,a_m$ and an ordered solution $\solS=(S_1,\ldots,S_\ell)$.
We associate with every possible element-part incidence pair $(i,j)$ a distinct power-of-two weight $\omega_{i,j}$.
The score of a solution is the sum of the weights corresponding to its incidence vector.
Thus, comparing solution scores is equivalent to comparing the corresponding incidence strings lexicographically.
Consequently, the lexicographically maximum solution is precisely a maximum-score solution.
Let $\mathrm{OPT}$ denote this maximum score.
The source problem asks whether the last incidence bit of a maximum-score solution is one.

The MMS construction again uses a scaled KPW core, but now requires only five agents.
There are one puzzle-solving agent $p$, one lex-opt indicator $a^*$, two row-forcing agents $q_1,q_2$, and one anchor agent $f$.
Besides the core goods, we introduce an assignment good $h_{i,j}$ for every pair $(i,j)$, representing the assignment of source number $a_i$ to part $j$.
The key idea is to encode feasibility and lexicographic optimality through the valuations of these assignment goods.

For this purpose, we associate with every assignment good $h_{i,j}$ a consistency number $c_{i,j}$.
These numbers are constructed using powers of a sufficiently large base $D$.
Each consistency number contains two nonzero contributions: one corresponding to the source index $i$ and another corresponding to the source number $a_i$ and its assigned part $j$.
We also define a target $B$ whose digits specify that every source index must be selected exactly once and that the numbers assigned to every part must sum to $T$.

More precisely, for a set $S$ of assignment goods, let $c(S)$ denote the sum of their consistency numbers.
The base $D$ is chosen large enough to prevent carries between digits.
Consequently, $c(S)=B$ holds if and only if the goods in $S$ encode a feasible solution of the source \fourpartition instance.
The cardinality requirement is automatic, since every source number lies strictly between $T/5$ and $T/3$, so every subset summing to $T$ contains exactly four numbers.
Furthermore, if $\omega(S)$ denotes the sum of the score weights of the selected assignment goods, then $\omega(S)$ is precisely the score of the corresponding source solution.

We first describe how the construction forces the intended allocation of the core goods.
As in the \NP-hardness reduction, the KPW core is scaled sufficiently to preserve row-or-column rigidity in the presence of the additional goods and valuation perturbations.
The two row-forcing agents rule out the column partition, so every MMS allocation must allocate the core goods by rows.
The anchor agent is then forced to receive the last row $R_5$.
The lex-opt indicator is forced to receive the first row $R_1$, since the other non-last rows have insufficient value to her even after adding all assignment goods.
Thus, the puzzle-solving agent receives one of the three remaining rows.

The puzzle-solving agent enforces the feasibility constraints and the final incidence bit.
She values each assignment good $h_{i,j}$ at $c_{i,j}-\varepsilon_{i,j}$, where $\varepsilon_{i,j}$ is a unit penalty when $i=m$ and $j\ne\ell$, and is zero otherwise.
Her non-last core rows have value $\tau-B$, while her MMS is exactly $\tau$.
Hence, if $S$ denotes the assignment goods she receives in an MMS allocation, then
$
    c(S)-\varepsilon(S)\ge B.
$
In particular, $c(S)\ge B$.
Moreover, once $c(S)=B$ is established, the same inequality will force $\varepsilon(S)=0$, meaning that the last source number is assigned to the last part.

The lex-opt indicator enforces lexicographic optimality.
For every assignment good $h_{i,j}$, her value is $Kc_{i,j}-\omega_{i,j}$, where $K$ is chosen sufficiently large that a unit discrepancy in the consistency number dominates every possible score contribution.
Her valuations on the core goods are carefully calibrated so that
$
    \mms_{a^*}=\tau+\mathrm{OPT}.
$

The lower bound follows by considering a maximum-score source solution.
The corresponding assignment goods are added to the first column of the core, and all remaining assignment goods are added to the second column.
The valuations are chosen so that the second bundle has value exactly $\tau+\mathrm{OPT}$, while every other bundle has at least this value.
Thus, the column partition witnesses $\mms_{a^*}\ge\tau+\mathrm{OPT}$.

For the upper bound, consider an arbitrary MMS-witnessing partition for the lex-opt indicator.
By the KPW rigidity, its core goods must be partitioned either by rows or by columns.
The row partition is ruled out by making certain core rows too deficient, even if they receive all assignment goods.
Hence, the core goods must be partitioned by columns.
If $S$ denotes the assignment goods in the first column bundle, the values of the first two column bundles force $c(S)=B$.
Therefore, $S$ encodes a feasible source solution.
The second column bundle then has value at most $\tau+\omega(S)\le\tau+\mathrm{OPT}$, proving the matching upper bound.

Finally, consider an actual MMS allocation.
By the forced row structure, the puzzle-solving agent receives the assignment goods in some set $S$, while the lex-opt indicator receives the first core row together with a subset of the remaining assignment goods.
Since her MMS is $\tau+\mathrm{OPT}$, the valuation of her bundle gives
$
    K(B-c(S))+\omega(S)\ge\mathrm{OPT}.
$
On the other hand, the puzzle-solving agent's MMS requirement gives $c(S)-\varepsilon(S)\ge B$.
Since $K$ dominates every possible score contribution, these two inequalities force $c(S)=B$.
They then imply $\varepsilon(S)=0$ and $\omega(S)=\mathrm{OPT}$.

Thus, the assignment goods selected by the puzzle-solving agent encode a feasible maximum-score source solution, and the zero penalty ensures that its final incidence bit is one.
Conversely, if the lexicographically maximum source solution has final bit one, its assignment goods can be distributed between the puzzle-solving agent and the lex-opt indicator so that both receive their MMS, while the remaining agents are satisfied by the intended row allocation.

Consequently, the constructed additive instance admits an MMS allocation if and only if the last bit of the lexicographically maximum solution of the \lexmaxfourpartitionpromise instance is one.
This yields $\Delta_2^P$-hardness for additive valuations, even for five agents.
Together with the known $\Delta_2^P$ membership, this establishes $\Delta_2^P$-completeness.

\subsection{\coNP-hardness}
\label{subsec:overview-conp}
The \coNP-hardness reductions use FST no-instances along with identical puzzle goods for all agents.
The weak \coNP-hardness reduction starts from the three-agent nine-item no-instance of \citet{DBLP:conf/wine/FeigeST21}.
In this instance, each agent has total value $120$ and maximin share $40$, but no allocation gives all three agents value at least $40$.
Given a \textsc{Three-Way-Partition} instance of total weight $3T$, we add one puzzle good for each number and give it value $121w_j$ for every agent.

If the puzzle items can be split into three parts of weight exactly $T$, then every agent has MMS exactly $121T + 40$.
However, in any allocation, either some agent receives puzzle weight at most $T - 1$, in which case the loss of one unit of puzzle weight cannot be compensated even by all core goods, or every agent receives puzzle weight exactly $T$, in which case the FST no-instance prevents all agents from receiving core value at least $40$.
Thus a yes-instance of \textsc{Three-Way-Partition} maps to a no-instance of MMS existence.

Conversely, if there is no perfect three-way partition, we choose a puzzle partition maximizing the minimum load.
The bundles attaining the minimum load are the bottleneck bundles.
If there is one bottleneck bundle, we give it all core goods.
If there are two bottleneck bundles, we use the two-agent maximin share guarantee on the core goods to support the two corresponding agents.
The remaining agent receives a strictly larger puzzle load, which by the choice of the scaling factor already covers any possible core deficit.
This gives an MMS allocation.

For strong \coNP-hardness, we use the $n$-agent FST no-instance on the core goods and generalize the previous idea.
Here there are two types of agents: $n - 2$ \emph{row} agents and $2$ \emph{column} agents.
Every agent has MMS $\tau$, but no allocation of the core goods gives every agent value at least $\tau$.
We reduce from \threepartition and scale each puzzle item by $L = n\tau + 1$.
If the \threepartition instance is a yes-instance, then every agent has MMS exactly $LT + \tau$.
But any allocation either gives some agent puzzle load at most $T - 1$, which cannot be repaired even with all core goods, or gives every agent puzzle load exactly $T$, in which case the FST no-instance prevents all agents from receiving core value at least $\tau$.

If the \threepartition instance is a no-instance, we again maximize the minimum puzzle load.
Suppose exactly $r$ puzzle bundles attain this minimum value.
The proof then only needs to support these $r$ bottleneck bundles with core goods.
For this, we use a simple feasibility property of the FST core: for every $r \in [n - 1]$, there is a set of $r$ agents among whom the core goods can be allocated so that each receives her $r$-agent maximin share on the core goods.
The non-bottleneck agents receive puzzle load at least one larger, and the scaling factor $L = n\tau + 1$ makes this one additional unit dominate the entire core value.
Thus, a no-instance of \threepartition maps to a yes-instance of MMS existence.

The strong \coNP-hardness reduction also has an approximation consequence.
In the no-case of \threepartition, the constructed instance admits an exact MMS allocation, whereas in the yes-case every allocation leaves some agent short by one unit from her MMS value.
Since all numbers in the strong reduction are polynomially bounded, this gives an inverse-polynomial gap for the optimal MMS approximation ratio: the largest $\alpha$ such that the instance admits an allocation where every agent $i$ receives a bundle of value at least $\alpha \mms_i$.
Thus, an algorithm that computes the optimal ratio could distinguish the two cases of \threepartition.
Consequently, unless $\P = \NP$, there is no FPTAS for computing the optimal MMS approximation ratio.

\subsection{From goods to chores}
\label{subsec:from-goods-to-chores}
Finally, we transfer our hardness results from goods to chores.
The reduction is based on the following trick.
Given a goods instance with $n$ agents and $m$ goods, we first \emph{pad} the instance with zero-valued dummy items so that the total number of items becomes exactly $nq$, where $q = m$.
The dummy items do not change the goods MMS values, but they ensure that size-wise balanced allocations, with exactly $q$ items per agent, are possible.

We then transform the goods valuations into chore utilities by subtracting a large cardinality-dependent term.
In the additive case, the constructed utility has the form
\[
u_i^C(S) = v_i(S) - \alpha_1 |S|.
\]
The parameter $\alpha_1$ is chosen large enough so that a bundle having more than $q$ items is too costly to be part of an MMS allocation.
Thus, every MMS allocation in the chores instance must give exactly $q$ items to every agent.

The main consequence is that, on balanced bundles of size exactly $q$, the chore utility differs from the original goods value by the same agent-independent constant.
Therefore, each agent's chore MMS is exactly her goods MMS \emph{shifted} by this constant.
It follows that the original goods instance admits an MMS allocation if and only if the constructed chores instance admits an MMS allocation.
Consequently, all hardness results for additive goods transfer to additive chores. 

%% file: 1_frontmatter/related-work.tex
\section{Other Related Work}
MMS allocations have also been studied beyond the standard goods-only and chores-only settings.
\citet{DBLP:conf/sigecom/KulkarniMT21} initiated the study of MMS together with Pareto optimality for indivisible mixed manna, where an item may be a good for some agents and a chore for others.
They show that, unlike the goods-only setting, for every fixed $\alpha > 0$ an $\alpha$-MMS allocation may fail to exist, and they give a PTAS for finding an approximately optimal MMS+PO allocation under suitable assumptions, together with hardness results when these assumptions are dropped.

Another line of work studies connected allocations, where the goods are vertices of a graph and each allocated bundle is required to induce a connected subgraph.
\citet{DBLP:conf/ijcai/BouveretCEIP17} showed that in the special case of trees, MMS allocations always exist and can be found in polynomial time, while \citet{DBLP:journals/jair/TruszczynskiL20} studied cycles and unicyclic graphs and showed that while the general decision problem is in $\Delta_2^P$, it is in $\NP$ for unicyclic graphs.
\citet{DBLP:conf/aaai/GrecoS20} studied the complexity of computing MMS allocations under connectivity constraints.
They showed that finding an MMS allocation is $F\Delta_2^P$-complete, even for instances with a single agent type, where an MMS allocation is guaranteed to exist.
Thus, their $F\Delta_2^P$-hardness concerns computing an MMS allocation when existence is guaranteed, whereas our work concerns the existence decision itself in the standard unconstrained setting.
In a related direction, \citet{DBLP:conf/aaai/IgarashiP19} studied Pareto optimality under the same connectivity constraints.
They showed that even on paths, finding a Pareto-optimal connected allocation satisfying MMS is NP-hard via Turing reductions.

Finally, \citet{DBLP:conf/aaai/ChristodoulouM26} studied MMS under graphical valuations, where agents are vertices and items are edges, and an edge can have nonzero value only for its endpoints.
For additive valuations on multigraphs, they show that an exact MMS allocation always exists; they also obtain approximation results for more general graphical valuations.

Another closely related line of work concerns the complexity of deciding whether an allocation satisfying both fairness and efficiency requirements exists.
\citet{DBLP:journals/jair/BouveretL08} showed that deciding whether there exists an envy-free allocation maximizing utilitarian social welfare is $\Delta_2^P$-complete, even for two agents with monotone valuations represented in a compact logical language.
For additive valuations, \citet{DBLP:conf/aldt/KeijzerBKZ09} established $\Sigma_2^P$-completeness for deciding whether a Pareto-optimal and envy-free allocation exists.
More recently, \citet{DBLP:journals/corr/abs-2609-31849} provided a detailed complexity classification for the existence of envy-free and efficient allocations, considering both Pareto optimality and social welfare optimality, different valuation classes, and restrictions on the number of agents.
Their results include $\Theta_2^P$- and $\Delta_2^P$-completeness results, as well as $\Sigma_2^P$-completeness for the existence of a Pareto-optimal and envy-free allocation under additive valuations even for a constant number of agents.
These results simultaneously satisfy envy-freeness and efficiency, whereas our work studies the existence of an exact MMS allocation without any additional efficiency requirement. 


\paragraph{Concurrent works.}Recently
\citet{ezra2026improved} present new additive no-instances with improved impossibility bounds for both goods and chores.
Our reductions use the earlier no-instance constructions given by \cite{DBLP:journals/jacm/KurokawaPW18} and \cite{DBLP:conf/wine/FeigeST21}.

\citet{deHaan2026} independently investigated the computational complexity of exact MMS allocations under additive valuations.
The preprint reports $\Delta_2^P$-completeness of MMS existence even for four agents with three distinct valuation functions, as well as $\Theta_2^P$-completeness for unary-encoded valuations when the number of agents is unrestricted.
Based on personal communication with the author, the manuscript is in the preliminary stages and the author intends to complete the verification of the stated results and further polish the draft. 
Moreover, we note that their reduction-its overall design and specifications-appear to be substantially different from the one presented here, in part because it uses a different no-instance.

%% file: 1_frontmatter/prelim.tex
\section{Preliminaries}
\label{sec:prelim}

\paragraph{Valuation functions.}
A valuation function $v_i : 2^M \to \mathbb{Q}$ is $k$-additive \citep{DBLP:journals/anor/ChevaleyreEEM08} if there is a coefficient function $\sigma_i : \{T \subseteq M : |T| \le k\} \to \mathbb{Q}$ such that for every $S \subseteq M$,
\[
    v_i(S)
    =
    \sum_{\substack{T \subseteq S \\ |T| \le k}}
    \sigma_i(T).
\]
Thus, the value of a bundle is obtained by adding all coefficients corresponding to subsets contained in it.
For an item $g \in M$, we call $\sigma_i(\{g\})$ the \emph{singleton coefficient} of $g$.
For two distinct items $g,h \in M$, we call $\sigma_i(\{g,h\})$ the \emph{pairwise coefficient} of the unordered pair $\{g,h\}$.
The case $k = 1$ is the usual additive case, up to the constant term $\sigma_i(\emptyset)$.
Throughout the paper, valuations are assumed to be normalized, that is, $\sigma_i(\emptyset) = 0$ for every agent $i$.
Note that there are no restrictions on the coefficients.
Thus, a valuation need not be monotone and may take negative values as well. 

In our inputs, the valuations are represented explicitly.
In particular, we do not work in a value-oracle model. 
This is the standard Turing-machine viewpoint for computational hardness involving non-additive valuations: as emphasized by \citet{DBLP:journals/tcs/GoldbergHH25}, the question is how hard the problem is when the valuation functions are given by a succinct representation that permits efficient evaluation, as opposed to value-oracle lower bounds that hold for algorithms accessing valuations only through value queries.
The $k$-additive valuations are represented by the coefficients $\sigma_i(T)$ for all $T \subseteq M$ with $|T| \le k$.
Thus, a $k$-additive function can be represented using $\bigoh(m^k)$ many numbers.

The following proposition gives the upper bound used in our completeness results.
Its proof follows the arguments used by \citet{DBLP:journals/jair/TruszczynskiL20} for additive valuations.


\begin{restatable}[\app]{proposition}{circuitmembership}
\label{prop:circuit-membership-delta2p}
For valuations represented by circuits that, given the characteristic vector of a bundle, output its value as a signed integer in binary, deciding whether an MMS allocation exists is in $\Delta_2^P$.
\end{restatable}

The proposition also yields the corresponding membership result for the explicitly
represented rational-valued $k$-additive valuations considered in this article.
For each agent, we may clear the denominators of all their coefficients by multiplying them by a common positive integer.
This common multiplier has polynomial bit complexity in the explicit representation,  and the resulting integer-valued $k$-additive function can be evaluated by a polynomial-size circuit.
Moreover, multiplying all values of an agent by a positive constant scales their maximin share by the same constant and therefore preserves MMS feasibility.

\subparagraph{Submodular valuations.}
A valuation function $v_i : 2^M \to \mathbb{Q}$ is submodular if it satisfies the following property: for every $S \subseteq T \subseteq M$ and every $x \in M \setminus T$,
\[
    v_i(S \cup \{x\}) - v_i(S)
    \ge
    v_i(T \cup \{x\}) - v_i(T).
\]

In the following lemma, we characterize $2$-additive functions that are submodular.

\begin{restatable}[\app]{lemma}{nonpositivepairssubmodular}
\label{lem:two-additive-submodular}
\label{lem:nonpositive-pairs-submodular}
Let $v_i$ be a $2$-additive valuation of the form
$
    v_i(S)
    =
    \sigma_i(\emptyset)
    +
    \sum_{g \in S} \sigma_i(\{g\})
    +
    \sum_{\{g,h\} \subseteq S} \sigma_i(\{g,h\}).
$
Then, $\sigma_i(\{g,h\}) \le 0$ for all distinct items $g,h \in M$ if and only if $v_i$ is submodular.
\end{restatable}

\subparagraph{Goods and chores.}
In a goods instance, the valuations are nonnegative, and larger values are preferred.
In a chores instance, the valuations are nonpositive, and larger values are still preferred.
When discussing chores instances, it is often convenient to use costs instead of valuations.
The \emph{cost} of a bundle is defined as the negation of its value:
\[
    c_i(S) = -v_i(S)
    \qquad
    \text{for every agent }i \text{ and } S \subseteq M.
\]
Thus, costs are nonnegative, and smaller costs are preferred.

For a chores instance, the maximin share of an agent $i$ can equivalently be written in cost notation as
\[
    \mms_i
    =
    \max_{(P_1, \dots, P_n)\in \Pi_n(M)}
    \min_{j \in [n]} v_i(P_j)
    =
    -
    \min_{(P_1, \dots, P_n)\in \Pi_n(M)}
    \max_{j \in [n]} c_i(P_j).
\]
Consequently, an allocation $A = (A_1, \dots, A_n)$ is an MMS allocation for a chores instance iff
\[
    c_i(A_i)
    \le
    \min_{(P_1, \dots, P_n)\in \Pi_n(M)}
    \max_{j \in [n]} c_i(P_j)
    \qquad
    \text{for every agent }i.
\]

\paragraph{Complexity classes and source problems.}
We establish our hardness results using standard polynomial-time many-one reductions.
For our \NP- and \coNP-hardness results, we use \threewaypartition and \threepartition.

The problem \threepartition is strongly \NP-hard \citep{DBLP:books/fm/GareyJ79}: given a multiset $W = \{w_1, \ldots, w_{3n}\}$ of positive integers and an integer $T$ such that $\sum_{j = 1}^{3n} w_j = nT$ and $\frac{T}{4} < w_j < \frac{T}{2}$ for every $j \in [3n]$, decide whether $W$ can be partitioned into $n$ triples, each of sum exactly $T$.

The problem \threewaypartition is weakly \NP-hard: given a multiset $W = \{w_1, \dots, w_m\}$ of positive integers and an integer $T$ such that $\sum_{j = 1}^m w_j = 3T$, decide whether $W$ can be partitioned into three parts, each of sum exactly $T$.

The class $D^P$, introduced by \citet{DBLP:journals/jcss/PapadimitriouY84}, is the second level of the Boolean hierarchy: a language $L$ is in $D^P$ if there exist languages $L_1 \in \NP$ and $L_2 \in \coNP$ such that $L = L_1 \cap L_2$.
It contains the classes \NP and \coNP.
A canonical complete problem is \textsc{(SAT,UNSAT)} \citep{DBLP:journals/jcss/PapadimitriouY84}: given two Boolean formulae $\varphi$ and $\psi$, decide whether $\varphi$ is satisfiable and $\psi$ is unsatisfiable.



The class $\Delta_2^P$ (also denoted by $P^\NP$) is the class of problems decidable in polynomial time with polynomially many adaptive queries to an $\NP$ oracle.
It contains the class $\Theta_2^P$.
Deciding whether the last variable is set to true in the lexicographically maximum satisfying assignment is a canonical $\Delta_2^P$-complete problem \citep{DBLP:journals/jcss/Krentel88}.

We also use the following lexicographic variant of \fourpartition.
In \lexmaxfourpartitionpromise, the input is a \fourpartition instance
$(W=\{w_1,\ldots,w_{4q}\},T)$ that is promised to admit a solution.
A solution $\solS=(S_1,\ldots,S_q)$ induces an incidence matrix $\incmat(\solS)\in\{0,1\}^{q\times 4q}$, where $\incmat(\solS)_{j,i}=1$ if and only if $w_i\in S_j$.
The corresponding incidence string is obtained by reading this matrix column by column.
The question is whether the last bit of the lexicographically maximum incidence string is equal to $1$.
We show in \Cref{lem:lexmaxfourpartitionpromise-delta2p-hardness} that this problem is $\Delta_2^P$-hard.


\paragraph{Verification problem.}
We present a simple hardness result for the problem of verifying whether a given allocation is an MMS allocation or not.
We include the proof in \Cref{sec:missingproofs} for completeness, since we are not aware of an explicit proof in the literature.

\begin{restatable}[\app]{proposition}{verificationmms}
\label{prop:verifying-mms-conp-complete}
Given an additive goods/chores instance and an allocation, deciding whether the allocation is an MMS allocation or not is \coNP-complete. 
\end{restatable}

%% file: 2_no_instances/no-instances.tex
\section{No instances for MMS allocations}
\label{sec:no-instances}

In this section, we discuss the details of the MMS no-instance constructions of \citet{DBLP:journals/jacm/KurokawaPW18} and \citet{DBLP:conf/wine/FeigeST21}, which our hardness reductions build upon.
These constructions are based on a small set of goods arranged as selected entries of a matrix, which we call the \emph{core}.
The entries of the matrix correspond to the values of the goods, and the valuation functions of the agents are obtained by perturbing this core in carefully chosen ways.

The constructed cores have strong structural properties: any allocation that gives every agent a certain target value must have a very restricted form.
Our reductions exploit this rigidity by adding further goods and using them to encode computationally hard problems.

The main takeaways from this section are summarized in \Cref{prop:KPW-core-construction,prop:FST21-3-agent,prop:FST-core-no-instance}.
A reader interested primarily in the reductions may skip the details of the constructions and proceed directly to these propositions, which contain all the structural properties used later.

We first revisit the construction of \citet{DBLP:journals/jacm/KurokawaPW18} focusing on the five-row core needed for our hardness reductions.
We then revisit two no-instances due to \citet{DBLP:conf/wine/FeigeST21}.
The first is a $3$-agent no-instance, which will be used in our proof of weak \coNP-hardness for $n = 3$.
The second is an $n$-agent no-instance, which will be used in our proof of strong \coNP-hardness.

\input{./2_no_instances/kpw-core}
\input{./2_no_instances/fst-core}

%% file: 2_no_instances/kpw-core.tex
\subsection{The KPW core for $n=5$}
\label{sec:KPW-core}
\label{subsec:KPW-core}

In this subsection, we discuss the core matrix underlying the no-instance of \citet{DBLP:journals/jacm/KurokawaPW18}.
We specialize their construction to five rows, which is sufficient for our hardness reductions.
We choose the perturbation parameters explicitly and then clear denominators, since our later reductions require an integer-valued core.

Let $n=5$, $K=125$, and $\varepsilon=2^{-K}$.
We first define two $5 \times 5$ matrices $S$ and $T$.
The matrix $S$ gives the unadjusted row-and-column structure, while $T$ is a small adjustment used to enforce the rigidity.

Let $S$ be the matrix
\[
S=
\begin{pmatrix}
\frac{15}{16} & 0 & 0 & 0 & \frac1{16}\\[6pt]
0 & \frac78 & 0 & 0 & \frac18\\[6pt]
0 & 0 & \frac34 & 0 & \frac14\\[6pt]
0 & 0 & 0 & \frac12 & \frac12\\[6pt]
\frac1{16} & \frac18 & \frac14 & \frac12 & \frac1{16}
\end{pmatrix}.
\]
The nonzero entries of $S$ consist of the diagonal entries, the last row, and the last column.
Observe that every row and every column of $S$ sums to $1$.

We now define the adjustment matrix $T$.
For every $i\in[3]$, let $r_i$ and $c_i$ be defined as follows:
\[
    r_i=\varepsilon^{8-2i}
    \qquad\text{and}\qquad
    c_i=\varepsilon^{7-2i}.
\]
Thus, we have
\[
0<r_1\ll c_1\ll r_2\ll c_2\ll r_3\ll c_3=\varepsilon.
\]
Next, let
\[
    u_i
    =
    \sum_{\substack{j\le i\\j\equiv i\pmod2}}c_j
    -
    \sum_{\substack{j\le i\\j\not\equiv i\pmod2}}r_j,
\]
and
\[
    v_i
    =
    \sum_{\substack{j\le i\\j\equiv i\pmod2}}r_j
    -
    \sum_{\substack{j\le i\\j\not\equiv i\pmod2}}c_j.
\]
For convenience, we write these quantities explicitly:
\[
\begin{aligned}
u_1&=c_1,
&
v_1&=r_1,\\
u_2&=c_2-r_1,
&
v_2&=r_2-c_1,\\
u_3&=c_1+c_3-r_2,
&
v_3&=r_1+r_3-c_2.
\end{aligned}
\]
Finally, we set
\[
    x=v_3,\qquad
    y=u_3,\qquad
    z=r_1+c_1+r_3+c_3.
\]

Note that each of the quantities $u_i,v_i,x,y,z$ is positive. 
For $u_i$ and $v_i$, this is true because the leading positive term is larger than the sum of all lower-order terms.

Let $T$ be the matrix
\[
T=
\begin{pmatrix}
0 & v_1 & 0 & 0 & -r_1\\
u_1 & 0 & v_2 & 0 & -r_2\\
0 & u_2 & 0 & v_3 & -r_3\\
0 & 0 & u_3 & 0 & -y\\
-c_1 & -c_2 & -c_3 & -x & z
\end{pmatrix}.
\]

Let $M=S+T$, and let $M^+=\{(i,j) \in [5]\times[5] : M_{i,j}>0\}$ be the positive entries of $M$.
We interpret each positive entry of $M$ as a good whose value is the corresponding entry of $M$.

For a set $X\subseteq M^+$, let
\[
    M(X)=\sum_{(i,j)\in X}M_{i,j},
    \qquad
    S(X)=\sum_{(i,j)\in X}S_{i,j},
    \qquad
    T(X)=\sum_{(i,j)\in X}T_{i,j}.
\]

We first establish the structural properties of the matrix $M$.
\begin{lemma}
    \label{lem:KPW-eight-properties}
    The matrix $M=S+T$ satisfies the following properties.
    \begin{enumerate}
        \item[\textnormal{[P1]}] 
        For every $i,j \in [5]$, we have $M_{i,j}\ge 0$. 
        Moreover, if $S_{i,j}\ne 0$ or $T_{i,j}\ne 0$, then $M_{i,j}>0$.

        \item[\textnormal{[P2]}]
        For every $i,j \in [5]$,
        \[
            |M_{i,j}-S_{i,j}|<\frac{1}{640}.
        \]

        \item[\textnormal{[P3]}] 
        Every row and every column of $M$ sums to $1$.

        \item[\textnormal{[P4]}] 
        Let $i\in[4]$. 
        If $X\subseteq M^+$ satisfies $(i,i)\in X$ and $M(X)=1$, then exactly one of the following five alternatives holds:
        \begin{enumerate}
            \item $(i,5)\in X$;
            \item $(5,i)\in X$;
            \item $(1,5),(2,5),\dots,(i-1,5),(5,5)\in X$;
            \item $(5,1),(5,2),\dots,(5,i-1),(5,5)\in X$;
            \item there exist $j,k<i$ such that $(j,5),(5,k)\in X$.
        \end{enumerate}

        \item[\textnormal{[P5]}] 
        For every $i\in[3]$, if $X\subseteq M^+$ satisfies $M(X)=r_i$, then
        \[
            X=
            \begin{cases}
                \{(1,2)\}, & \text{if } i=1,\\
                \{(i,i-1),(i,i+1)\}, & \text{if }2\le i\le 3.
            \end{cases}
        \]

        \item[\textnormal{[P6]}] 
        For every $i\in[3]$, if $X\subseteq M^+$ satisfies $M(X)=c_i$, then
        \[
            X=
            \begin{cases}
                \{(2,1)\}, & \text{if } i=1,\\
                \{(i-1,i),(i+1,i)\}, & \text{if }2\le i\le 3.
            \end{cases}
        \]

        \item[\textnormal{[P7]}] 
        If $X\subseteq M^+$ satisfies $M(X)=x$, then $X=\{(3,4)\}$.

        \item[\textnormal{[P8]}] 
        If $X\subseteq M^+$ satisfies $M(X)=y$, then $X=\{(4,3)\}$.
    \end{enumerate}
\end{lemma}

\begin{proof}
    We first study the sequence $r_1,c_1,r_2,c_2,r_{3},c_{3}$, which will help us in the subsequent arguments.

    \begin{claim-inside-lemma}  
        \label{clm:dominance}
        Every term in the sequence $r_1,c_1,r_2,c_2,r_3,c_3$ is larger than the sum of all preceding terms.
    \end{claim-inside-lemma}
    \begin{proof}
        Let $\lambda_t=\varepsilon^{7-t}$ for each $t \in [6]$.
        Then, we have 
        \[
            (\lambda_1,\dots,\lambda_6)
            =
            (r_1,c_1,r_2,c_2,r_3,c_3).
        \]
        Therefore, for every $2 \le t\le 6$, we have
        \[
            \sum_{s<t}\lambda_s
            =
            \varepsilon^{7-(t-1)}
            +
            \varepsilon^{7-(t-2)}
            +\cdots+
            \varepsilon^{6}
            \le
            \lambda_t(\varepsilon+\varepsilon^2+\cdots)
            =
            \lambda_t\frac{\varepsilon}{1-\varepsilon}.
        \]
        Since $\varepsilon=2^{-125}$, we have
        \[
            \frac{\varepsilon}{1-\varepsilon}<\frac{1}{80}.
        \]
        Hence,
        \[
            \lambda_t>80\sum_{s<t}\lambda_s
            \qquad\text{for every }t\ge 2.
        \]
    \end{proof}
    We now prove the first two properties [P1] and [P2]. 
    Observe that every quantity $u_i,v_i,x,y,z$ is positive. 
    The leading term in the expression of $u_i$ is $c_i$, and all negative terms appearing in $u_i$ are among $r_1,c_1,\dots,r_{i-1},c_{i-1},r_i$.
    By \Cref{clm:dominance}, their total cannot cancel $c_i$. 
    Similarly, the leading term in the expression of $v_i$ is $r_i$, and the sum of negative terms appearing in it is smaller than $r_i$. 
    Thus, $v_i>0$. 
    Since $x=v_{3}$ and $y=u_{3}$, we also have $x,y>0$. 
    Finally, $z$ is a sum of positive terms and thus, $z>0$.

    By the definitions of $x$ and $y$, we have
    $x\le r_1+r_3$ and $y\le c_1+c_3$.
    Thus, each negative entry of $T$ has absolute value at most
    \[
        \sum_{t=1}^{6}\lambda_t
        \le
        \frac{\varepsilon}{1-\varepsilon}
        <
        \frac{1}{640},
    \]
    where the last inequality follows from $\varepsilon=2^{-125}$.
    These negative entries occur only in the last row or the last column of $T$. 
    Every positive entry of $S$ in the last row or last column has value at least $1/16$. 
    Hence, adding any negative entry of $T$ to the corresponding entry of $S$ leaves a strictly positive value.

    We now verify property [P1]. 
    If $S_{i,j}=T_{i,j}=0$, then $M_{i,j}=0$. 
    If $S_{i,j}\ne 0$ and $T_{i,j}\ge 0$, then clearly $M_{i,j}=S_{i,j}+T_{i,j}>0$. 
    If $S_{i,j}\ne 0$ and $T_{i,j}<0$, then $(i,j)$ lies in the last row or the last column, and the preceding paragraph shows that $M_{i,j}>0$. 
    Finally, if $S_{i,j}=0$ and $T_{i,j}\ne 0$, then $T_{i,j}$ cannot be negative, because all negative entries of $T$ occur only where the corresponding entry of $S$ is positive. 
    Hence $T_{i,j}>0$, and so $M_{i,j}=T_{i,j}>0$.

    Therefore, $M_{i,j}\ge 0$ for every $i,j$. 
    Moreover, if $S_{i,j}\ne 0$ or $T_{i,j}\ne 0$, then $M_{i,j}>0$. 
    This proves property [P1].

    Every nonzero entry of $T$ is one of $u_i,v_i,-r_i,-c_i,-x,-y,z$.
    Each of these has absolute value at most $\sum_{s=1}^{6}\lambda_s \le 1/640$.
    Therefore, for every $i,j$, $|M_{i,j}-S_{i,j}|=|T_{i,j}|<1/640$.
    This proves property [P2].
    
    We next prove that every row and every column of $M$ sums to $1$ (property [P3]). 
    Since every row and every column of $S$ sums to $1$, it suffices to show that every row and every column of $T$ sums to $0$.

    From the definitions of $u_i$ and $v_i$, we have $v_1=r_1$, $u_1+v_2=r_2$, $u_2+v_3=r_3$, and $u_3=y$.
    Moreover, by the definitions of $x$ and $z$, we have $c_1+c_2+c_3+x=z$.
    These identities show that every row of $T$ sums to $0$.
    
    Similarly, we have $u_1=c_1$, $v_1+u_2=c_2$, $v_2+u_3=c_3$, and $v_3=x$.
    Finally, we have $r_1+r_2+r_3+y=z$.
    Therefore, every column of $T$ also sums to $0$.
    Since every row and every column of $S$ sums to $1$, every row and every column of $M=S+T$ sums to $1$.
    This proves property [P3].

    Next, we note that $|T(X)|<1/32$ for every $X\subseteq M^+$.
    Indeed, the matrix $T$ has at most $15$ nonzero entries, and each nonzero entry has absolute value at most $\sum_{t=1}^{6}\lambda_t \le \frac{\varepsilon}{1-\varepsilon}$.
    Thus,
    \[
        |T(X)|
        \le
        15\frac{\varepsilon}{1-\varepsilon}
        <
        \frac{1}{32}.
    \]
    
    Every entry of $S$ is an integer multiple of $1/16$. 
    Hence, for every $X\subseteq M^+$, the quantity $S(X)$ is an integer multiple of $2^{-(n-1)}$.
    We now show two small consequences that we will use later:
    \begin{claim-inside-lemma}
        \label{clm:two-small-consequences}
        The following statements hold.
        \begin{enumerate}
            \item If $M(X)=1$, then $S(X)=1$.
            \item If $M(X)<1/32$, then $S(X)=0$.
        \end{enumerate}
    \end{claim-inside-lemma}
    \begin{proof}
        First, suppose $M(X)=1$. 
        Since $M(X)=S(X)+T(X)$, we have $|S(X)-1|=|T(X)|<1/32$.
        But both $S(X)$ and $1$ are integer multiples of $1/16$. 
        Therefore, if $S(X)\ne 1$, then $|S(X)-1|\ge 1/16 > 1/32$, a contradiction. 
        Hence, $S(X)=1$.
    
        Second, suppose $M(X)<1/32$. 
        We claim that $S(X)=0$. 
        If not, then since $S(X)$ is a positive integer multiple of $1/16$, we would have $S(X)\ge 1/16$.
        Therefore
        \[
            M(X)=S(X)+T(X)
            \ge S(X)-|T(X)|
            >
            1/16-1/32
            =
            1/32,
        \]
        contradicting $M(X)<1/32$. 
        Hence, $S(X)=0$.
    \end{proof}

    Next, we prove property [P4].
    First, suppose $i\in[4]$, $(i,i)\in X$, and $M(X)=1$. 
    Then, by the first part of \Cref{clm:two-small-consequences}, we have $S(X)=1$. 
    Since $S_{i,i}=1-2^{-(5-i)}$, the other entries of $X$ must have total $S$-value exactly $2^{-(5-i)}$. 
    No entry of $S$-value larger than $2^{-(5-i)}$ can be included. 
    This implies that no other diagonal entry $(k,k)$ with $k<n$ and $k\ne i$ can appear in $X$. 
    After including $(i,i)$, the remaining $S$-budget is $2^{-(5-i)}\le 1/2$, while every such diagonal entry has $S$-value at least $1/2$, with equality only for $(4,4)$. 
    If $i=4$, this equality case is exactly $(i,i)$ itself; otherwise the remaining budget is strictly smaller than $1/2$. 
    Thus every other non-last diagonal entry has $S$-value larger than the remaining budget. 
    The corner entry $(5,5)$ is different, since $S_{5,5}=2^{-(4)}$, and it may appear.
    Thus, outside $(i,i)$, the only positive $S$-entries that can appear are
    \[
        (i,5),\quad (5,i),\quad (j,5),(5,j)\text{ for }j<i,\quad\text{and }(5,5).
    \]
    If $(i,5)\in X$, we are in case (a) of [P4]. 
    If $(5,i)\in X$, we are in case (b) of [P4]. 
    If neither is present and $X$ contains both some $(j,5)$ and some $(5,k)$ with $j,k<i$, we are in case (e) of [P4]. 
    Otherwise we claim that the missing $S$-value must be made entirely from one side. 
    The identity
    \[
        2^{-(5-i)}
        =
        \sum_{h=1}^{i-1}2^{-(5-h)}+2^{-(5-1)}
    \]
    says precisely that this requires taking either
    \[
        (1,5),(2,5),\dots,(i-1,5),(5,5)
    \]
    or
    \[
        (5,1),(5,2),\dots,(5,i-1),(5,5).
    \]
    In the above identity, the last term $2^{-(4)}$ is the contribution of the entry $(5,5)$.
    The above two alternatives form cases (c) and (d) of property [P4].
    
    We now show that exactly one alternative holds. 
    After including $(i,i)$, which has $S$-value $1-2^{-(5-i)}$, the remaining entries of $X$ must have total $S$-value exactly $2^{-(5-i)}$.
    Cases (a) and (b) are mutually exclusive with all other cases, because $(i,5)$ and $(5,i)$ each have $S$-value exactly $2^{-(5-i)}$, which already exhausts the remaining $S$-budget after $(i,i)$.
    
    We first note that when $i=1$, neither case (c) nor case (d) can occur. 
    When $i=1$, cases (c) and (d) both reduce to the condition $(5,5)\in X$. 
    But if $(5,5)\in X$, then the entries $(1,1)$ and $(5,5)$ already have total $S$-value $1$. 
    Moreover, $T_{1,1}=0$ and $T_{5,5}=z>0$.
    No negative $T$-entry can be added without adding positive $S$-value and exceeding the $S$-budget. 
    Thus $M(X)>1$, contradicting the assumption that $M(X)=1$. 
    
    When $i \ge 2$, we show that cases (c) and (d) cannot occur together. 
    If $i\ge 2$, then case (c) already contributes exactly the full remaining $S$-budget $2^{-(5-i)}$.
    If case (d) also occurred, then $X$ would contain at least one additional positive-$S$ entry, which would make the $S$-value outside $(i,i)$ strictly larger than $2^{-(5-i)}$, a contradiction.
    Therefore, cases (c) and (d) do not occur simultaneously.
    
    Finally, case (e) cannot occur together with any of cases (a)--(d), because cases (a)--(d) already exhaust the remaining $S$-budget, whereas case (e) would add at least one more positive-$S$ entry.
    Therefore, exactly one of the five alternatives holds.
    This proves property [P4].
    
    Next, we prove property [P5].
    Suppose $i\in[3]$ and $M(X)=r_i$. 
    Since $r_i<2^{-5}$, the second part of \Cref{clm:two-small-consequences} gives $S(X)=0$. 
    Therefore, $X$ contains only entries whose positive value comes purely from the matrix $T$, namely entries of value $u_j$ and $v_j$.

    We will repeatedly use the following consequence of \Cref{clm:dominance}: when a target value has leading term $\lambda_t$, any subset of positive $T$-entries summing to that target must contain a positive $T$-entry with the same leading term and cannot contain any entry with a larger leading term.

    The leading term of $v_i$ is $r_i$, and this is the unique positive $T$-entry with leading term $r_i$. 
    If $X$ contained a $T$-entry with leading term larger than $r_i$, then by \Cref{clm:dominance} we would have $M(X)>r_i$.
    If it contained no entry with leading term $r_i$, then by \Cref{clm:dominance} the total value of all entries in $X$ would be less than $r_i$.
    Hence $X$ must contain the unique $T$-entry with leading term $r_i$: the entry $(i,i+1)$, whose value is $v_i$. 
    The remaining value is $r_i-v_i=u_{i-1}$, where $u_0=0$ for convenience. 
    If $i=1$, then $u_0=0$, and so $X=\{(1,2)\}$.
    If $i\ge 2$, then the remaining value $u_{i-1}$ forces the unique $T$-entry with leading term $c_{i-1}$, namely $(i,i-1)$. 
    Its value is exactly $u_{i-1}$, so no further entry can be included. 
    Hence, $X=\{(i,i-1),(i,i+1)\}$.
    This is property [P5].

    The proof of [P6] is symmetric. 
    Suppose $i\in[3]$ and $M(X)=c_i$. 
    Again, by \Cref{clm:two-small-consequences}, $S(X)=0$, so $X$ contains only $T$-entries. 
    The unique such entry with leading term $c_i$ is $(i+1,i)$, whose value is $u_i$. 
    The remaining value is $c_i-u_i=v_{i-1}$, where $v_0=0$ for convenience. 
    If $i=1$, then $X=\{(2,1)\}$.
    If $i\ge 2$, the remaining value $v_{i-1}$ forces the unique $T$-entry $(i-1,i)$, and hence $X=\{(i-1,i),(i+1,i)\}$.
    This is property [P6].

    For [P7], suppose $M(X)=x$. 
    Since $x=v_{3}<2^{-5}$, by \Cref{clm:two-small-consequences}, we have $S(X)=0$. 
    The unique $T$-entry with leading term $r_{3}$ is $(3,4)$, and its value is $v_{3}=x$. 
    Therefore, $X=\{(3,4)\}$.

    For [P8], suppose $M(X)=y$. 
    Since $y=u_{3}<2^{-5}$, we have $S(X)=0$. 
    The unique $T$-entry with leading term $c_{3}$ is $(4,3)$, and its value is $u_{3}=y$. 
    Therefore, $X=\{(4,3)\}$.

    We have now verified properties [P1]--[P8]. 
\end{proof}

We next establish the rigidity property of the core.

\begin{lemma}
    \label{lem:only-rows-or-columns-M}
    If $(X_1,\dots,X_5)$ is a partition of $M^+$ such that
        $
            \sum\limits_{(i,j) \in X_t} M_{i,j}=1
            \qquad\text{for every } t \in [5],
        $
        then $(X_1,\dots,X_5)$ is either the row partition or the column partition, up to relabeling.
\end{lemma}

\begin{proof}
    The matrix $M$ satisfies properties [P1]--[P8] stated in~\Cref{lem:KPW-eight-properties}. 
    The proof of Lemma 2.2 in \citet{DBLP:journals/jacm/KurokawaPW18} uses precisely these structural properties.
    Hence, the same induction argument applies to our matrix.
    Since the original presentation contains some minor typographical slips, including in the invocation of the structural properties and in some displayed residual sums, we give the full proof in our notation.
    
    Informally, the argument goes through as follows.
    The bundle containing $(1,1)$ must be either the first row or the first column.
    This is because applying [P4] with $i = 1$ leaves only the possibilities $(1,5) \in X$, $(5,1) \in X$, or $(5,5) \in X$.
    The last possibility is impossible because the bundle would already have value greater than $1$.
    If $(1,5) \in X$, then [P5] forces the remaining entry $(1,2)$, so the bundle is the first row.
    If $(5,1) \in X$, then [P6] forces the remaining entry $(2,1)$, so the bundle is the first column.
    If the first row is a bundle, the induction forces the second row, then the third row, and so on, until the entire partition is the row partition.
    Symmetrically, if the first column is a bundle, the induction forces the column partition.
    Therefore, every partition of $M^+$ into $5$ bundles of value $1$ is either the row partition or the column partition, up to relabeling.

    Let $R_i$ and $C_i$ denote the sets of positive entries in row $i$ and column $i$ of $M$, respectively.
    Let us consider the bundle in the partition that contains $(1,1)$, say $X$.
    We first show that $X$ is either the first row $R_1$ or the first column $C_1$.
    
    By [P4] applied with $i=1$, one of the five alternatives holds.
    Alternative (e) is impossible since there are no indices $j,k<1$.
    Alternatives (c) and (d) require $(5,5)\in X$.
    But then
    $
        M(X)
        \ge
        M_{1,1}+M_{5,5}
        =
        1+z
        >
        1,
    $
    a contradiction.
    
    Suppose alternative (a) holds, so that $(1,5)\in X$.
    Since
    $
        M_{1,1}+M_{1,5}=1-r_1,
    $
    the remaining entries of $X$ must have total value $r_1$.
    By [P5], these remaining entries are exactly $\{(1,2)\}$.
    Therefore,
    $
        X=\{(1,1),(1,2),(1,5)\}=R_1.
    $
    
    Suppose instead that alternative (b) holds, so that $(5,1)\in X$.
    Since
    $
        M_{1,1}+M_{5,1}=1-c_1,
    $
    the remaining entries of $X$ must have total value $c_1$.
    By [P6], these remaining entries are exactly $\{(2,1)\}$.
    Therefore,
    $
        X=\{(1,1),(2,1),(5,1)\}=C_1.
    $
    Thus, the bundle containing $(1,1)$ is either $R_1$ or $C_1$.
    
    We now show that if the first row is a bundle, then all rows are bundles.
    Suppose that $R_1,\dots,R_{i-1}$ are bundles in the partition for some $2\le i\le4$.
    Let $X$ be the bundle containing $(i,i)$.
    By [P4], one of the five alternatives holds.
    
    First, suppose alternative (a) holds, so that $(i,5)\in X$.
    If $i\le3$, then
    $
        M_{i,i}+M_{i,5}=1-r_i.
    $
    Hence, the remaining entries of $X$ must have total value $r_i$.
    By [P5], these remaining entries are exactly
    $
        \{(i,i-1),(i,i+1)\}.
    $
    Therefore, $X=R_i$.
    
    If $i=4$, then
    $
        M_{4,4}+M_{4,5}=1-y.
    $
    The remaining entries must have total value $y$.
    By [P8], these remaining entries consist exactly of $(4,3)$.
    Therefore, $X=R_4$.
    
    Next, suppose alternative (b) holds, so that $(5,i)\in X$.
    If $i\le3$, then
    $
        M_{i,i}+M_{5,i}=1-c_i.
    $
    The remaining entries of $X$ must have total value $c_i$.
    By [P6], $X$ must contain $(i-1,i)$.
    But this entry lies in row $i-1$, which is already one of the bundles of the partition.
    This is impossible.
    
    If $i=4$, then
    $
        M_{4,4}+M_{5,4}=1-x.
    $
    By [P7], the remaining entries must contain $(3,4)$.
    But $(3,4)$ belongs to row $3$, which is already a bundle.
    This is again impossible.
    
    Next, suppose alternative (c) holds.
    Then $X$ contains
    $
        (1,5),(2,5),\dots,(i-1,5),(5,5).
    $
    In particular, $X$ contains $(1,5)$, which belongs to the already allocated row $R_1$.
    This is impossible.
    
    Next, suppose alternative (d) holds.
    Then, $X$ contains
    $
        (5,1),(5,2),\dots,(5,i-1),(5,5).
    $
    Thus,
    \[
    \begin{aligned}
        M(X)
        &\ge
        M_{i,i}
        +
        \sum_{h=1}^{i-1}M_{5,h}
        +
        M_{5,5}\\
        &=
        \left(1-2^{-(5-i)}\right)
        +
        \sum_{h=1}^{i-1}
            \left(2^{-(5-h)}-c_h\right)
        +
        \left(\frac1{16}+z\right)\\
        &=
        1-\sum_{h=1}^{i-1}c_h+z\\
        &\ge
        1-(c_1+c_2+c_3)+z\\
        &=
        1+x\\
        &>
        1.
    \end{aligned}
    \]
    Here, we used the identity
    $
        z=c_1+c_2+c_3+x.
    $
    This contradicts the fact that $M(X)=1$.
    
    Finally, suppose alternative (e) holds.
    Then $X$ contains some $(j,5)$ with $j<i$.
    But $(j,5)$ belongs to row $j$, which is already one of the bundles of the partition.
    This is impossible.
    
    Therefore, the only possible case is $X=R_i$.
    By induction, if $R_1$ is a bundle, then $R_1,R_2,R_3,R_4$ are all bundles.
    Since the partition has exactly five bundles, the remaining positive entries form the last row $R_5$, which by [P3] has value $1$.
    Hence, the partition is the row partition, up to relabeling.
    
    To complete the proof, we show that if the first column is a bundle, then all columns are bundles.
    Suppose that $C_1,\dots,C_{i-1}$ are bundles in the partition for some $2\le i\le4$.
    Let $X$ be the bundle containing $(i,i)$.
    By [P4], one of the five alternatives holds.
    
    First, suppose alternative (a) holds, so that $(i,5)\in X$.
    If $i\le3$, then
    $
        M_{i,i}+M_{i,5}=1-r_i.
    $
    By [P5], the remaining entries of $X$ must contain $(i,i-1)$.
    But $(i,i-1)$ belongs to column $i-1$, which is already one of the bundles of the partition.
    This is impossible.
    
    If $i=4$, then
    $
        M_{4,4}+M_{4,5}=1-y.
    $
    By [P8], the remaining entries must contain $(4,3)$.
    But $(4,3)$ belongs to column $3$, which is already a bundle.
    This is again impossible.
    
    Next, suppose alternative (b) holds, so that $(5,i)\in X$.
    If $i\le3$, then
    $
        M_{i,i}+M_{5,i}=1-c_i.
    $
    The remaining entries must have total value $c_i$.
    By [P6], these remaining entries are exactly
    $
        \{(i-1,i),(i+1,i)\}.
    $
    Therefore, $X=C_i$.
    
    If $i=4$, then
    $
        M_{4,4}+M_{5,4}=1-x.
    $
    The remaining entries must have total value $x$.
    By [P7], these remaining entries consist exactly of $(3,4)$.
    Therefore, $X=C_4$.
    
    Next, suppose alternative (c) holds.
    Then, $X$ contains
    $
        (1,5),(2,5),\dots,(i-1,5),(5,5).
    $
    Thus,
    $$
    \begin{aligned}
        M(X)
        &\ge
        M_{i,i}
        +
        \sum_{h=1}^{i-1}M_{h,5}
        +
        M_{5,5}\\
        &=
        \left(1-2^{-(5-i)}\right)
        +
        \sum_{h=1}^{i-1}
            \left(2^{-(5-h)}-r_h\right)
        +
        \left(\frac1{16}+z\right)\\
        &=
        1-\sum_{h=1}^{i-1}r_h+z\\
        &\ge
        1-(r_1+r_2+r_3)+z\\
        &=
        1+y\\
        &>
        1.
    \end{aligned}
    $$
    Here, we used
    $
        z=r_1+r_2+r_3+y.
    $
    This contradicts the fact that $M(X)=1$.
    
    Next, suppose alternative (d) holds.
    Then, $X$ contains
    $
        (5,1),(5,2),\dots,(5,i-1),(5,5).
    $
    In particular, $X$ contains $(5,i-1)$, which belongs to column $i-1$, already one of the bundles of the partition.
    This is impossible.
    
    Finally, suppose alternative (e) holds.
    Then $X$ contains some $(5,k)$ with $k<i$.
    But $(5,k)$ belongs to column $k$, which is already one of the bundles of the partition.
    This is impossible.
    
    Therefore, the only possible case is $X=C_i$.
    By induction, if $C_1$ is a bundle, then $C_1,C_2,C_3,C_4$ are all bundles.
    Since the partition has exactly five bundles, the remaining positive entries form the last column $C_5$, which by [P3] has value $1$.
    Hence, the partition is the column partition, up to relabeling.
\end{proof}

We now clear denominators.
Since $\varepsilon=2^{-125}$, every term used in $T$ has a denominator that divides $2^{750}$.
The denominators in $S$ divide $16$, and hence also divide $2^{750}$.
Let
\[
    \rowsum=2^{750}.
\]
This value of $\rowsum$ clears all denominators of $M$.
For every $i,j\in[5]$, we define
\[
    A_{i,j}=\rowsum M_{i,j}.
\]
Then, each $A_{i,j}$ is a nonnegative integer.
Moreover, $A_{i,j}>0$ if and only if $(i,j)\in M^+$.
Let
$
    A^+=\{(i,j):A_{i,j}>0\}
$
be the positive entries in $A$.
Since every row and every column of $M$ sums to $1$, every row and every column of $A$ sums to $\rowsum$.

The following proposition summarizes the properties of the above construction that will be used in our reductions.

\begin{restatable}{proposition}{kpwcoreconstruction}
\label{prop:KPW-core-construction}
There exist a nonnegative integer matrix $A=(A_{i,j})_{i,j\in[5]}$ and a positive integer $\rowsum=2^{750}$ satisfying the following properties.
\begin{enumerate}
\item
The matrix $A$ has exactly $19$ positive entries.

\item
Every row and every column of $A$ sums to $\rowsum$.

\item
The integer $\rowsum$ and every entry of $A$ have at most $751$ bits.

\item
Let $A^+=\{(i,j):A_{i,j}>0\}$ and let $(X_1,\dots,X_5)$ be a partition of $A^+$.
If
\[
    \sum_{(i,j)\in X_t}A_{i,j}
    \ge\rowsum
    \qquad\text{for every }t\in[5],
\]
then $(X_1,\dots,X_5)$ is either the row partition or the column partition, up to relabeling.
\end{enumerate}
\end{restatable}

\begin{proof}
Let $M$ be the matrix constructed above.
For every $i,j\in[5]$, we set $A_{i,j}=\rowsum M_{i,j}$.
As shown above, this choice of $\rowsum$ clears all denominators of $M$.
Hence, every $A_{i,j}$ is an integer.

By property [P1] of \Cref{lem:KPW-eight-properties}, every entry of $A$ is nonnegative, and $A_{i,j}>0$ if and only if $(i,j)\in M^+$.

The support of $M$ consists of the positive entries of $S$ together with the positive entries of $T$ that occur at zero entries of $S$.
The matrix $S$ has $13$ positive entries, and the positive entries of $T$ contribute six additional entries.
The remaining nonzero entries of $T$ occur at entries already positive in $S$.
Therefore, $A$ has exactly $13+6=19$ positive entries.

By property [P3] of \Cref{lem:KPW-eight-properties}, every row and every column of $M$ sums to $1$.
Therefore, every row and every column of $A$ sums to $\rowsum$.

Since $\rowsum=2^{750}$, the integer $\rowsum$ can be stored in $751$ bits.
Moreover, every entry of $A$ is nonnegative and every row of $A$ sums to $\rowsum$, so every entry of $A$ is at most $\rowsum$.
Thus, every entry of $A$ can also be stored in at most $751$ bits.

Finally, we show the last property.
Since the total value of all positive entries of $A$ is $5\rowsum$, the assumption implies that every bundle has value exactly $\rowsum$.
Dividing by $\rowsum$, every bundle has $M$-value exactly $1$.
By \Cref{lem:only-rows-or-columns-M}, the partition is either the row partition or the column partition, up to relabeling.
Thus, we have proved the proposition.
\end{proof}

%% file: 2_no_instances/fst-core.tex
\subsection{The FST no-instance for $n = 3$}
\label{ssec:Feige}
\label{subsec:FST-no-instance-3-agents}

In this subsection, we discuss the $3$-agent no-instance of \citet{DBLP:conf/wine/FeigeST21} and establish the properties that will be needed for the weak \coNP-hardness reduction.
There are three agents, denoted by $R$, $C$, and $U$, and nine items arranged in a $3 \times 3$ matrix. 
The valuations of the three agents are given by the matrices
\begin{equation}
    \label{eqn:FST-no-instance-3-agents}
    M_R = 
    \begin{pmatrix}
      1 & 16 & 23 \\
      26 & 4 & 10 \\
      12 & 19 & 9
    \end{pmatrix},
    \qquad
    M_C = 
    \begin{pmatrix}
      1 & 16 & 22 \\
      26 & 4 & 9 \\
      13 & 20 & 9
    \end{pmatrix},
    \qquad
    M_U = 
    \begin{pmatrix}
      1 & 15 & 23 \\
      25 & 4 & 10 \\
      13 & 20 & 9
    \end{pmatrix}.
\end{equation}
For each agent, the total value of the nine items is $120$, and the maximin share is $40$. 
Moreover, no allocation gives all three agents value at least $40$, and thus there is no MMS allocation.

\begin{proposition}[{\cite{DBLP:conf/wine/FeigeST21}}]
    \label{prop:FST21-3-agent}
    In the instance described in \eqref{eqn:FST-no-instance-3-agents}, for every allocation of the nine items to $R, C, U$, at least one agent receives a bundle of value at most $39$.
\end{proposition}

\subsection{The FST core for general $n$}
\label{sec:FST-core}

In this subsection, we discuss the $n$-agent no-instance of \citet{DBLP:conf/wine/FeigeST21} and establish the properties that will be needed for the strong \coNP-hardness reduction.

Fix an integer $n \ge 4$.
The items will be represented as selected entries of an $n\times n$ matrix.
Rows are indexed from top to bottom and columns from left to right.
All unspecified entries are treated as $0$, and only the positive entries will correspond to goods.

We first define a base matrix $B$.
For every $j\in\{2, \ldots, n - 1\}$, we set
\[
    B_{1,j} = (n - 2)n,
    \qquad
    B_{j,1} = (n - 2)(n - 1),
    \qquad
    B_{j,j} = (n - 2)(n^2 - 4n + 2),
    \qquad
    B_{j,n} = (n - 2)(n - 1) + 1.
\]
We also set
\[
    B_{1,n} = 1,
    \qquad
    B_{n,1} = B_{n,2} = \cdots = B_{n,n-1} = (n - 2)^2 + 1,
    \qquad
    B_{n,n} = (n - 2)(n - 3).
\]
All other entries are $0$.

Let $E_{\mathrm{core}} = \{(i,j) : B_{i,j} > 0\}$ be the set of core goods.
The number of core goods is $5n - 7$.
Let $t_B = n(n - 2)^2 + 1$.
A direct calculation shows that every row and every column of $B$ has sum exactly $t_B$.

\begin{lemma}[{\cite{DBLP:conf/wine/FeigeST21}}]
    \label{lem:FST-base-good-partitions}
    In every partition of $E_{\mathrm{core}}$ into $n$ bundles where each bundle sums to exactly $t_B$, at least one of the following three conditions holds:
    \begin{enumerate}
        \item the bottom row is split among the $n$ bundles, one bottom-row item in each bundle;
        \item the right column is split among the $n$ bundles, one right-column item in each bundle;
        \item at least one bundle contains at least one item from the bottom row and at least one item from the right column, but does not contain the item $(n,n)$.
    \end{enumerate}
\end{lemma}

We now define two valuation functions, $V_R$ and $V_C$.
First multiply every entry of $B$ by $n$.
The row valuation $V_R$ is obtained by perturbing the last row by 1:
\[
    v_{R}(i, j) = \begin{cases}
        n B_{i, j}, & \text{if } i < n, \\
        n B_{n, j} - 1, & \text{if } i = n \text{ and } j < n, \\
        n B_{n, n} + (n - 1), & \text{if } i = j = n.
    \end{cases}
\]
Similarly, the column valuation $V_C$ is obtained by perturbing the last column by 1:
\[
    v_{C}(i, j) = \begin{cases}
        n B_{i, j}, & \text{if } j < n, \\
        n B_{i, n} - 1, & \text{if } i < n \text{ and } j = n, \\
        n B_{n, n} + (n - 1), & \text{if } i = j = n.
    \end{cases}
\]

Let $\tau = nt_B = n^2(n - 2)^2 + n$.
Every row has value $\tau$ under $V_R$, and every column has value $\tau$ under $V_C$.
Moreover, for both valuations, the total value of all core goods is $n\tau$.
Indeed, the perturbation has total change $0$: in $V_R$, we subtract $1$ from $n - 1$ bottom-row items and add $n - 1$ to $(n,n)$; the same holds symmetrically for $V_C$.

We now define the agents of the core instance.
Let $R$ and $C$ be disjoint sets of agents such that $|R| + |C| = n$, $|R| \ge 2$, and $|C| \ge 2$.
Every agent in $R$ has valuation $V_R$, and every agent in $C$ has valuation $V_C$.

\begin{proposition}[\cite{DBLP:conf/wine/FeigeST21}]
    \label{prop:FST-core-no-instance}
    In the constructed instance, every agent has maximin share exactly $\tau$.
    Moreover, no allocation of the core goods gives every agent value at least $\tau$.
\end{proposition}

%% file: 3_np/np-hardness-goods.tex
\section{\NP-hardness}
\label{sec:np}

We prove \NP-hardness for five agents by reducing from \textsc{Partition}.
Let $I = (X,T)$ be an instance of \textsc{Partition}, where $X = \{a_1,\dots,a_m\}$ is a collection of positive integers satisfying $\sum_{r=1}^m a_r = 2T$.
The question is whether $X$ can be partitioned into two sets, each of sum exactly $T$.

\subsection{Construction}

We construct a fair division instance with $n=5$ agents.
The instance contains two kinds of agents: puzzle-solving agents and \Q agents.
There are three puzzle-solving agents, denoted by $p_0,p_1,p_2$.
The remaining two agents, denoted by $q_1,q_2$, are \Q agents.

Informally, the \Q agents force the KPW core goods to be allocated by rows, while the puzzle-solving agents encode the \textsc{Partition} instance.

We first construct the $5\times 5$ \KPW integer matrix $A$ given by \Cref{prop:KPW-core-construction}.
As before, let \rowsum be the common row and column sum of $A$.
Let $R_i$ and $C_j$ denote the set of core goods in row $i$ and column $j$, respectively.

\subsubsection{Goods}

The core goods are defined by the positive entries of $A$.
For each $(i,j)\in A^+$, let $g_{i,j}$ denote the corresponding core good.
Let
$
    M_{\mathrm{core}}
    =
    \{g_{i,j}:(i,j)\in A^+\}.
$

We add the puzzle goods next.
For every number $a_r\in X$ in the \textsc{Partition} instance, we create one puzzle good $x_r$.
Let
$
    M_X = \{x_r:r\in[m]\}.
$
The set of all goods is
$
    M=M_{\mathrm{core}}~\dot\cup~M_X.
$

\subsubsection{Valuations}

We now define the valuations for each agent one type at a time.
We choose
\[
    \Gamma=4T+5
    \qquad\text{and}\qquad
    \tau=\Gamma\rowsum.
\]
The parameter $\Gamma$ is chosen so that we can distinguish between row/column partitions and every other partition of the \KPW core, even after adding the lower-order perturbations and puzzle goods.

\begin{enumerate}[leftmargin=*]

\item~[Puzzle-solving agents.]
The valuation of a puzzle-solving agent on core goods is obtained by perturbing the last column.
For every puzzle-solving agent $p_s$, the valuation on core goods is defined as follows:
\[
    v_{p_s}(g_{i,j})
    =
    \begin{cases}
        \Gamma A_{i,j},
            & \text{if } j<5,\\
        \Gamma A_{i,5}-T,
            & \text{if } j=5 \text{ and } i<5,\\
        \Gamma A_{5,5}+2T,
            & \text{if } i=j=5.
    \end{cases}
\]

Every puzzle-solving agent values the puzzle goods according to the \textsc{Partition} numbers:
\[
    v_{p_s}(x_r)=a_r
    \qquad\text{for every }r\in[m].
\]

\begin{observation}
\label{obs:NP-row-column-set-puzzle-solving}
For every puzzle-solving agent $p_s$, we can infer the following about the row and column sets.
\begin{align}
    v_{p_s}(C_j)
    &=\tau,
        &&\text{for every }j\in[4],\nonumber\\
    v_{p_s}(C_5)
    &=\tau-2T,\nonumber\\
    v_{p_s}(R_i)
    &=\tau-T,
        &&\text{for every }i\in[4],\\
    v_{p_s}(R_5)
    &=\tau+2T.\nonumber
\end{align}
Moreover, $v_{p_s}(M)=5\tau$.
\end{observation}

\item~[\Q agents.]
The valuation of a \Q agent on core goods is obtained by perturbing the last row.
For every \Q agent $q$, the valuation on core goods is defined as follows:
\[
    v_q(g_{i,j})
    =
    \begin{cases}
        \Gamma A_{i,j},
            & \text{if } i<5,\\
        \Gamma A_{5,j}-1,
            & \text{if } i=5 \text{ and } j<5,\\
        \Gamma A_{5,5}+4,
            & \text{if } i=j=5.
    \end{cases}
\]

Every \Q agent values every puzzle good at $0$:
\[
    v_q(x_r)=0
    \qquad\text{for every }r\in[m].
\]

\begin{observation}
\label{obs:NP-row-column-set-Q-agent}
For every \Q agent $q$, we can infer the following about the row and column sets.
\begin{align}
    v_q(R_i)
    &=\tau,
        &&\text{for every }i\in[5],\nonumber\\
    v_q(C_j)
    &=\tau-1,
        &&\text{for every }j\in[4],\\
    v_q(C_5)
    &=\tau+4.\nonumber
\end{align}
Moreover, $v_q(M)=5\tau$.
\end{observation}

\end{enumerate}

We note that all valuation functions are nonnegative because every positive entry of $A$ is at least $1$ and $\Gamma$ is larger than every subtractive perturbation.
The valuations are additive.
This completes the construction.

\subsection{Analysis}

We first show that the \Q agents force the core goods to be allocated by rows.

\begin{lemma}
\label{lem:NP-core-row-forcing}
In any allocation in which every agent receives value at least $\tau$, the induced partition of the core goods is the row partition, up to relabeling.
\end{lemma}

\begin{proof}
Consider the allocation of the core goods induced by such an allocation.
Suppose that this induced core allocation is neither the row partition nor the column partition.
Then, by \Cref{prop:KPW-core-construction}, some core bundle has $A$-value at most $\rowsum-1$.
After scaling by $\Gamma$, this bundle has core value at most
$
    \Gamma(\rowsum-1)=\tau-\Gamma.
$

We now upper bound how much this deficit can be overcome for each type of agent.

\begin{itemize}
\item
For a puzzle-solving agent, the positive core perturbation is at most $2T$, due to the good $g_{5,5}$.
Additionally, the total value of all puzzle goods is exactly $2T$.
Hence, even after adding all possible positive core perturbation and the value of all puzzle goods, the value of the bundle is at most
$
    \tau-\Gamma+2T+2T<\tau,
$
since $\Gamma=4T+5$.

\item
For a \Q agent, the positive core perturbation is at most $4$, due to the good $g_{5,5}$.
Moreover, \Q agents value all puzzle goods at $0$.
Since $\Gamma>4$, such a bundle has value at most
$
    \tau-\Gamma+4<\tau.
$
\end{itemize}

Thus, if the induced core allocation is neither the row partition nor the column partition, then some agent receives value strictly less than $\tau$, contradicting the assumption.
Consequently, it remains to rule out the column partition.

There are two \Q agents.
Under the column partition, by \Cref{obs:NP-row-column-set-Q-agent}, a \Q agent values the last column at $\tau+4$ and every non-last column at $\tau-1$.
Since \Q agents value all puzzle goods at $0$, at most one \Q agent can receive value at least $\tau$ under a column partition of the core goods.
This contradicts the assumption that every agent receives value at least $\tau$.

Therefore, the core goods must be allocated as rows, up to relabeling.
\end{proof}

In the following, we compute the maximin shares of the agents, where we view the partitions from each agent's valuation.
These values will be used to prove the correctness of the reduction in \Cref{lem:NP-reduction-correctness}.

\begin{lemma}
\label{lem:NP-mms-values}
The following statements hold.
\begin{enumerate}
\item
For every puzzle-solving agent $p_s$, we have $\mms_{p_s}=\tau$.
\item
For every \Q agent $q$, we have $\mms_q=\tau$.
\end{enumerate}
\end{lemma}

\begin{proof}
Consider a puzzle-solving agent $p_s$.
We first show that $\mms_{p_s}\ge\tau$.
Consider the column partition of the core goods.
By \Cref{obs:NP-row-column-set-puzzle-solving}, the first four columns have value exactly $\tau$, whereas the last column has value $\tau-2T$.
We add all puzzle goods to the bundle containing the last column.
Since the total value of all puzzle goods is exactly $2T$, this bundle now has value
$
    \tau-2T+2T=\tau.
$
Thus, we obtain a partition of all goods into five bundles, each of value exactly $\tau$.
Therefore, $\mms_{p_s}\ge\tau$.

By \Cref{obs:NP-row-column-set-puzzle-solving}, the total value of all goods to $p_s$ is exactly $5\tau$.
Hence, the averaging upper bound gives $\mms_{p_s}\le\tau$.
Thus, $\mms_{p_s}=\tau$ for every puzzle-solving agent $p_s$.

Now consider a \Q agent $q$.
The row partition of the core goods gives value exactly $\tau$ in every bundle, by \Cref{obs:NP-row-column-set-Q-agent}.
Since \Q agents value all puzzle goods at $0$, distributing the puzzle goods arbitrarily preserves this value.
Therefore, $\mms_q\ge\tau$.

The total value of all goods to $q$ is exactly $5\tau$, by \Cref{obs:NP-row-column-set-Q-agent}.
Therefore, the averaging upper bound gives $\mms_q\le\tau$.
Thus, $\mms_q=\tau$ for every \Q agent $q$.
\end{proof}

\begin{lemma}
\label{lem:NP-reduction-correctness}
The \textsc{Partition} instance is a yes-instance if and only if the constructed goods instance admits an MMS allocation.
\end{lemma}

\begin{proof}
Suppose first that the \textsc{Partition} instance is a yes-instance.
Let $I_1,I_2$ be a partition of $[m]$ such that
$
    \sum_{r\in I_1}a_r
    =
    \sum_{r\in I_2}a_r
    =
    T.
$
We construct an MMS allocation.

We allocate the core goods by rows.
\begin{itemize}
\item
We give the last row $R_5$ to $p_0$.
Then,
$
    v_{p_0}(R_5)
    =
    \tau+2T
    \ge \tau
    =
    \mms_{p_0}.
$

\item
We give two distinct non-last rows to $p_1$ and $p_2$, one row to each.
For every $j\in[2]$, we also give $p_j$ the puzzle goods
$
    X_j=\{x_r:r\in I_j\}.
$
Each such agent receives value
$
    \tau-T+T
    =
    \tau
    =
    \mms_{p_j}.
$

\item
We give the remaining two rows to the \Q agents, one row to each.
Every \Q agent values every row at exactly $\tau$, which is her maximin share.
\end{itemize}

Therefore, every agent receives value at least her MMS, and an MMS allocation exists.

Conversely, suppose that the constructed goods instance admits an MMS allocation.
By \Cref{lem:NP-mms-values}, every agent has maximin share $\tau$.
Hence, every agent receives value at least $\tau$ in the MMS allocation.
By \Cref{lem:NP-core-row-forcing}, the core goods are allocated as rows, up to relabeling.

There are three puzzle-solving agents and only one last row.
Therefore, at least two puzzle-solving agents receive non-last rows.
A puzzle-solving agent receiving a non-last row gets core value exactly $\tau-T$.
Since her maximin share is $\tau$, she must receive puzzle goods of total value at least $T$.

The total value of all puzzle goods is exactly $2T$.
Therefore, at most two puzzle-solving agents can receive non-last rows.
Consequently, exactly two puzzle-solving agents receive non-last rows, each of them receives puzzle goods of total value exactly $T$, and no puzzle value is wasted on any other agent.

Thus, the puzzle goods are partitioned into two bundles, each of total value $T$.
These two bundles correspond to a partition of the numbers $a_1,\dots,a_m$ into two sets of equal sum.
Hence, the original \textsc{Partition} instance is a yes-instance.
\end{proof}

\begin{theorem}
\label{thm:NP-hardness-goods}
Deciding whether an additive goods instance admits an MMS allocation is \NPH, even for five agents with only two distinct valuation functions.
\end{theorem}

\begin{proof}
By \Cref{lem:NP-reduction-correctness}, the constructed goods instance admits an MMS allocation if and only if the original \textsc{Partition} instance is a yes-instance.

We now show that the reduction runs in polynomial time.
By \Cref{prop:KPW-core-construction}, the KPW core matrix $A$ is a fixed $5\times5$ integer matrix with exactly $19$ positive entries.
Moreover, every entry of $A$ and the common row sum $\rowsum$ have constant bit length.

The scaling factor $\Gamma=4T+5$ has $\bigoh(\log T)$ bits.
Every core value in the construction is of the form $\Gamma A_{i,j}$ plus or minus a perturbation of absolute value at most $2T$.
Therefore, every core value has polynomial bit length.
The puzzle-good values are the positive integers $a_r$, so these also have polynomial bit length.
Moreover, all values are nonnegative by construction, and the valuations are additive.

Finally, the instance has exactly five agents and $19+m$ goods.
Hence, the entire instance can be constructed and encoded in polynomial time.

The construction uses only two distinct valuation functions, one shared by the three puzzle-solving agents and the other shared by the two \Q agents.
Therefore, the reduction proves \NP-hardness even for five agents with only two agent types.
\end{proof}

%% file: 5_delta2p/delta2p-hardness-additive-goods.tex
\section{$\Delta_2^P$-hardness}
\label{sec:delta2p-additive}

We prove $\Delta_2^P$-hardness for instances with additive valuations, even when there are only five agents.
We first establish the $\Delta_2^P$-hardness of a lexicographic variant of \fourpartition and then present a polynomial-time reduction from that problem to the \decision.
The overall sequence of reductions is\footnote{Although the ingredients in this chain of reductions are standard, we are not aware of a reference establishing the precise promised lexicographic variants needed here.
We therefore include the intermediate reductions for completeness.}
\begin{equation*}
\lexmaxsat
\rightarrow
\lexmaxsatpromise
\rightarrow
\lexmaxthreesatpromise
\rightarrow
\lexmaxfourpartitionpromise.
\end{equation*}

In an instance of \fourpartition, we are given a multiset $W = \{w_1, \dots, w_{4q}\}$ of positive integers and an integer $T$ such that $\sum_{j=1}^{4q} w_j = qT$ and $\frac{T}{5} < w_j < \frac{T}{3}$ for every $j \in [4q]$.
A solution $\solS$ is an ordered tuple $(S_1, \dots, S_q)$ of pairwise disjoint four-element submultisets whose union is $W$ and such that every $S_j$ sums to exactly $T$.
The \emph{incidence matrix} of $\solS$ is defined as a matrix $\incmat(\solS) \in \{0,1\}^{q \times 4q}$, such that $\incmat(\solS)_{j,i} = 1$ if and only if the number $w_i \in S_j$.
The incidence string $\inc(\solS) \in \{0,1\}^{4q^2}$ is obtained by reading the matrix $\incmat(\solS)$ column by column:
\begin{equation*}
\inc(\solS)
=
\left(
\incmat(\solS)_{1,1}, \dots, \incmat(\solS)_{q,1},
\dots,
\incmat(\solS)_{1,4q}, \dots, \incmat(\solS)_{q,4q}
\right).
\end{equation*}

In the \lexmaxfourpartitionpromise problem, we are given an instance $(W = \{w_1, \dots, w_{4q}\}, T)$ of \fourpartition that is promised to admit a solution.
Let $\solS^*$ be a solution such that the incidence string $\inc(\solS^*)$ is lexicographically maximum.
The question is whether the last bit of $\inc(\solS^*)$ is equal to $1$; that is, whether $\incmat(\solS^*)_{q,4q} = 1$.
We will prove the following

\begin{restatable}[\app]{lemma}{lexmaxfourpartitionpromisehardness}
\label{lem:lexmaxfourpartitionpromise-delta2p-hardness}
\lexmaxfourpartitionpromise is $\Delta_2^P$-hard even when restricted to instances with $4q$ numbers in which the target $T$ and every number are $\bigoh(q^4)$.
\end{restatable}

We now present a polynomial-time reduction from \lexmaxfourpartitionpromise to the \decision.
Let $(W,T)$ be an instance of \lexmaxfourpartitionpromise, where $W = \{a_1, \dots, a_{4\ell}\}$, $\sum_{i=1}^{4\ell} a_i = \ell T$, and $\frac{T}{5} < a_i < \frac{T}{3}$ for every $i \in [4\ell]$.
We refer to the numbers $a_1, \dots, a_{4\ell}$ as the \emph{source numbers}.
The bounds $\frac{T}{5}<a_i<\frac{T}{3}$ imply that every set summing to $T$ contains exactly four source numbers.
Recall that a solution is an ordered tuple $\solS = (S_1, \dots, S_\ell)$, and its incidence matrix is $\incmat(\solS) \in \{0,1\}^{\ell \times 4\ell}$, where $\incmat(\solS)_{j,i} = 1$ if and only if the number $a_i \in S_j$.

We first associate a numerical score with every source solution.
Let $m=4\ell$.
Let
$
    R=m\ell=4\ell^2
$
denote the length of the incidence string.
For every $(i,j)\in[m]\times[\ell]$, let \Ma{the associated weight be}
$
    \omega_{i,j}
    =
    2^{R-(\ell(i-1)+j)}.
$
For a solution $\solS$, we define its \emph{score} by
\[
    \score(\solS)
    =
    \sum_{i=1}^{m}\sum_{j=1}^{\ell}
    \omega_{i,j}\cdot\incmat(\solS)_{j,i}.
\]

The weights $\omega_{i,j}$ are precisely the powers
$2^{R-1},2^{R-2},\dots,2^0$ in the order in which the corresponding bits occur in the incidence string.
Consequently,
\[
    \inc(\solS_1)>_{\mathrm{lex}}\inc(\solS_2)
    \quad\Longleftrightarrow\quad
    \score(\solS_1)>\score(\solS_2).
\]
Thus, the lexicographically maximum solution is precisely a solution of maximum score.

Let
\[
    P=2^R
    \qquad\text{and}\qquad
    W_{\mathrm{sc}}=P-1.
\]
The total of all score weights is $W_{\mathrm{sc}}$, and hence every solution has score at most $W_{\mathrm{sc}}$.

We define
\[
    \mathrm{OPT}
    =
    \max\{
    \score(\solS):
    \solS\text{ is a solution of the source instance}
    \}.
\]
Since the source instance is promised to admit a solution, $\mathrm{OPT}$ is well defined.
Moreover, we have
$
    0\le\mathrm{OPT}\le W_{\mathrm{sc}}.
$

\subsection{Construction}

We construct a fair division instance with five agents.
The instance contains one puzzle-solving agent $p$, one \emph{lex-opt indicator} agent $a^*$, two row-forcing agents $q_1,q_2$, and one anchor agent $f$.

Informally, the two row-forcing agents ensure that the KPW core goods are allocated by rows, while the anchor agent consumes the exceptional last row.
The puzzle-solving agent encodes a feasible solution of the source instance.
The lex-opt indicator compares the score of the encoded solution with the maximum possible score, which is reflected in her maximin share.

\subsubsection{Goods}

We first construct the $5\times5$ \KPW integer matrix $A$ given by \Cref{prop:KPW-core-construction}.
As before, let \rowsum be the common row and column sum of $A$.
Let $R_r$ and $C_s$ denote the sets of core goods in row $r$ and column $s$, respectively.

The core goods are defined by the positive entries of $A$.
For every $(r,s)\in A^+$, let $g_{r,s}$ denote the corresponding core good.
Let
$
    M_{\mathrm{core}}
    =
    \{g_{r,s}:(r,s)\in A^+\}.
$

We add the assignment goods next.
For every pair $(i,j)\in[m]\times[\ell]$, we create one assignment good $h_{i,j}$.
Let
$
    M_{\mathrm{asg}}
    =
    \{h_{i,j}:i\in[m],\,j\in[\ell]\}.
$
Informally, selecting the good $h_{i,j}$ represents assigning source number $a_i$ to part $j$.

The set of all goods is
$
    M
    =
    M_{\mathrm{core}}
    ~\dot\cup~
    M_{\mathrm{asg}}.
$

\subsubsection{Parameters}

We first define integers that encode the constraints of a feasible source solution.
\Ma{We define the base}
\[
    D=\ell T+\ell+1.
\]
For every $(i,j)\in[m]\times[\ell]$, we define the
\emph{consistency number} using the base $D$
\[
    c_{i,j}  = D^{i-1} + a_iD^{m+j-1}.
\]
We also define the \emph{target} using the base $D$
\[
    B = \sum_{i=1}^{m}D^{i-1} + T\sum_{j=1}^{\ell}D^{m+j-1}.
\]

The consistency numbers encode two requirements.
The first $m$ base-$D$ digits ensure that every source index is selected exactly once.
The remaining $\ell$ digits ensure that the source numbers assigned to each part sum to exactly $T$.
The choice of $D$ ensures that no carries occur between these positions, even if all assignment goods are selected. 

We do not need to encode the requirement that each part contains exactly four source numbers separately.
Indeed, since every source number lies strictly between $T/5$ and $T/3$, every subset summing to $T$ contains exactly four numbers.

For a set $S\subseteq M_{\mathrm{asg}}$, we denote \Ma{the sum of the corresponding consistency numbers and their weights by} 
\[
    c(S)
    =
    \sum_{h_{i,j}\in S}c_{i,j},
    \qquad
    \omega(S)
    =
    \sum_{h_{i,j}\in S}\omega_{i,j}.
\]

We first establish the property of the consistency numbers that will be used throughout the reduction. 
\Ma{It shows that a solution for the source instance yields a subset of the assignment goods that precisely captures the correspondence between a number and the tuple to which it belongs.
}

\begin{lemma}
\label{lem:delta-consistency-encoding}
A subset $S\subseteq M_{\mathrm{asg}}$ satisfies $c(S)=B$ if and only if there exists a solution $\solS=(S_1,\dots,S_\ell)$ of the source instance such that
$
    S
    =
    \{h_{i,j}:a_i\in S_j,\ j\in[\ell]\}.
$
Under this correspondence, $\omega(S)=\score(\solS)$.
\end{lemma}

\begin{proof}
Suppose first that $c(S)=B$.
By the choice of $D$, no carries occur in the base-$D$ representation of $c(S)$.
This is because the coefficient of an item digit is at most $\ell$, while the coefficient of a part-sum digit is at most $\sum_{i=1}^{m}a_i=\ell T$.
Both quantities are strictly smaller than $D$.

For every $i\in[m]$, the coefficient of $D^{i-1}$ is exactly $1$.
Hence, exactly one assignment good among $h_{i,1},\dots,h_{i,\ell}$ belongs to $S$.

For every $j\in[\ell]$, the coefficient of $D^{m+j-1}$ is exactly $T$.
Consequently, the source numbers assigned to part $j$ sum to exactly $T$.

Since every source number lies strictly between $T/5$ and $T/3$, each such part contains exactly four numbers.
Thus, the selected assignment goods describe a solution $\solS$ of the source instance.

Conversely, consider any solution $\solS=(S_1,\dots,S_\ell)$.
We select $h_{i,j}$ precisely when $a_i\in S_j$.
Every source index is selected exactly once, and the numbers assigned to every part sum to $T$.
Therefore, the selected consistency numbers sum to $B$.

The equality $\omega(S)=\score(\solS)$ follows directly from the definitions.
\end{proof}

We next introduce the final-bit penalty, \Ma{whose role is to check that the last bit of the incidence string is set to 1.}
For every $(i,j)\in[m]\times[\ell]$, let
\[
    \varepsilon_{i,j}
    =
    \begin{cases}
        1, & \text{if }i=m\text{ and }j\ne\ell,\\
        0, & \text{otherwise}.
    \end{cases}
\]
For a set $S\subseteq M_{\mathrm{asg}}$, we denote the sum of the penalties as 
$
    \varepsilon(S)
    =
    \sum_{h_{i,j}\in S}\varepsilon_{i,j}.
$

If $c(S)=B$, then exactly one assignment good with source index $m$ belongs to $S$.
Hence,
\[
    \varepsilon(S)=0
    \quad\Longleftrightarrow\quad
    h_{m,\ell}\in S.
\]
Equivalently, the corresponding source solution has final incidence bit equal to $1$.

Let
\[
    C_{\mathrm{tot}}
    =
    \sum_{i=1}^{m}\sum_{j=1}^{\ell}c_{i,j},
    \qquad
    V_P
    =
    C_{\mathrm{tot}}-(\ell-1).
\]
We state this simple implication that follows by rearranging the terms.
\begin{observation}
We note that
\[
\begin{aligned}
    C_{\mathrm{tot}}
    =
    \ell\sum_{i=1}^{m}D^{i-1}
    +
    \left(\sum_{i=1}^{m}a_i\right)
    \sum_{j=1}^{\ell}D^{m+j-1}
    =
    \ell B,
\end{aligned}
\]
where $\sum_{i=1}^{m}a_i=\ell T$.
\end{observation}

The quantity $V_P$ will be the total value of all assignment goods to the puzzle-solving agent.

We further define
\[
    K=2W_{\mathrm{sc}}+2,
    \qquad
    F=2W_{\mathrm{sc}}+1,
    \qquad
    U=KC_{\mathrm{tot}}+F+1.
\]
Thus, we have
$
    K>F>W_{\mathrm{sc}}\ge\mathrm{OPT},
    $ and $
    F>2W_{\mathrm{sc}}.
$

The parameter $K$ separates a unit discrepancy in the consistency number from every possible score contribution.
The parameter $F$ ensures that the intended MMS-witnessing partition of the lex-opt indicator has its minimum value in the second column bundle.
Finally, $U$ makes certain core rows too small for the lex-opt indicator, even if all assignment goods are added to them.

We will scale the KPW core by a sufficiently large integer $\Gamma$, to be specified after defining all lower-order perturbations.
Finally, we define $\tau=\Gamma\rowsum$. 

\subsubsection{Valuations}

We now define the valuations for each agent one type at a time.

For every agent $z$ and every core good $g_{r,s}$, the valuation will be of the form
\[
    v_z(g_{r,s})
    =
    \Gamma A_{r,s}+\delta_z(r,s),
\]
where $\delta_z(r,s)$ is an integer perturbation.
All perturbations not explicitly specified below are zero.

\begin{enumerate}[leftmargin=*]

\item~[Puzzle-solving agent.]
For every assignment good $h_{i,j}$, we define
\[
    v_p(h_{i,j})
    =
    c_{i,j}-\varepsilon_{i,j}.
\]
Thus,
\[
    v_p(M_{\mathrm{asg}})
    =
    V_P.
\]

The valuation of $p$ on core goods is obtained by perturbing the last column.
We define
\[
    \delta_p(r,5)=-B
    \qquad\text{for every }r\in[4],
\]
and
\[
    \delta_p(5,5)=4B-V_P.
\]

\begin{observation}
\label{obs:delta-puzzle-values}
For the puzzle-solving agent $p$,
\begin{align}
    v_p(R_r)
    &=\tau-B,
        &&\text{for every }r\in[4],\nonumber\\
    v_p(R_5)
    &=\tau+4B-V_P,\\
    v_p(C_s)
    &=\tau,
        &&\text{for every }s\in[4],\nonumber\\
    v_p(C_5)
    &=\tau-V_P.\nonumber
\end{align}
Moreover, $v_p(M)=5\tau$.
\end{observation}

\item~[Lex-opt indicator.]
For every assignment good $h_{i,j}$, we define
\[
    v_{a^*}(h_{i,j})
    =
    Kc_{i,j}-\omega_{i,j}.
\]
Consequently,
\[
    v_{a^*}(M_{\mathrm{asg}})
    =
    KC_{\mathrm{tot}}-W_{\mathrm{sc}}.
\]

We now define the perturbations of the core goods.
We first prescribe the desired offsets of the column sums:
\[
\begin{aligned}
    \beta_1
    &=
    F-KB,\\
    \beta_2
    &=
    W_{\mathrm{sc}}-K(C_{\mathrm{tot}}-B),\\
    \beta_3
    &=
    \beta_4
    =
    \beta_5
    =
    F.
\end{aligned}
\]
We also prescribe the offsets of the first four row sums:
\[
\begin{aligned}
    \gamma_1
    &=
    W_{\mathrm{sc}}-K(C_{\mathrm{tot}}-B),\\
    \gamma_2
    &=
    \gamma_3
    =
    \gamma_4
    =
    -U.
\end{aligned}
\]

We realize these offsets by perturbing only the last row and the last column.
Specifically, let
\[
\begin{aligned}
    \delta_{a^*}(r,5)
    &=
    \gamma_r
    &&\text{for every }r\in[4],\\
    \delta_{a^*}(5,s)
    &=
    \beta_s
    &&\text{for every }s\in[4],\\
    \delta_{a^*}(5,5)
    &=
    \beta_5-\sum_{r=1}^{4}\gamma_r.
\end{aligned}
\]

Observe that every position whose value is perturbed corresponds to a positive entry of the KPW core.

\begin{observation}
\label{obs:delta-indicator-values}
For the lex-opt indicator $a^*$, the core column values satisfy
\[
    v_{a^*}(C_s)
    =
    \tau+\beta_s
    \qquad\text{for every }s\in[5].
\]
Moreover,
\[
    v_{a^*}(R_r)
    =
    \tau+\gamma_r
    \qquad\text{for every }r\in[4],
\]
and
\[
    v_{a^*}(R_5)
    =
    \tau+4F-KB+3U.
\]
\end{observation}

\item~[Row-forcing agents.]
The two row-forcing agents $q_1,q_2$ have identical valuations.
For each such agent $q$, we perturb the last row by setting
\[
    \delta_q(5,s)=-1
    \qquad\text{for every }s\in[4],
\]
and
\[
    \delta_q(5,5)=4.
\]
Every row-forcing agent values every assignment good at zero.

\begin{observation}
\label{obs:delta-row-forcing-values}
For every row-forcing agent $q$,
\[
    v_q(R_r)=\tau
    \qquad\text{for every }r\in[5],
\]
and
\[
    v_q(C_s)
    =
    \begin{cases}
        \tau-1, & \text{if }s\in[4],\\
        \tau+4, & \text{if }s=5.
    \end{cases}
\]
Moreover, $v_q(M)=5\tau$.
\end{observation}

\item~[Anchor agent.]
The valuation of the anchor agent $f$ on core goods is obtained by perturbing the last column.
We define
\[
    \delta_f(r,5)=-1
    \qquad\text{for every }r\in[4],
\]
and
\[
    \delta_f(5,5)=4.
\]
The anchor agent values every assignment good at zero.

\begin{observation}
\label{obs:delta-anchor-values}
For the anchor agent $f$,
\[
    v_f(R_r)
    =
    \begin{cases}
        \tau-1, & \text{if }r\in[4],\\
        \tau+4, & \text{if }r=5,
    \end{cases}
\]
and
\[
    v_f(C_s)=\tau
    \qquad\text{for every }s\in[5].
\]
Moreover, $v_f(M)=5\tau$.
\end{observation}

\end{enumerate}

It remains to specify the scaling factor $\Gamma$.
For every agent $z$, let
$
    V_z=v_z(M_{\mathrm{asg}}).
$
We \Ma{set} 
\[
    \Gamma
    =
    1+
    \max_{z\in\{p,a^*,q_1,q_2,f\}}
    \left(
        \sum_{(r,s)\in A^+}|\delta_z(r,s)|
        +
        V_z
    \right).
\]

Observe that this definition depends only on the source instance and the explicitly defined parameters.
In particular, it does not depend on $\mathrm{OPT}$.

\begin{observation}
All valuations are nonnegative.    
\end{observation}

\begin{proof}
Since every consistency number is greater than $1$, and $K>W_{\mathrm{sc}}$, all assignment-good values are nonnegative.
Moreover, every positive entry of the integer KPW matrix is at least $1$, and our choice of $\Gamma$ ensures that the scaled core values remain strictly positive after the perturbations.
\end{proof}

The following inequality will be useful in the analysis.
For every agent $z$, every subset $H\subseteq M_{\mathrm{core}}$ whose unscaled $A$-value is at most $\rowsum-1$, and every subset $S\subseteq M_{\mathrm{asg}}$, we have
\begin{align}
    v_z(H\cup S)
    \le
    \Gamma(\rowsum-1)
    +
    \sum_{(r,s)\in A^+}|\delta_z(r,s)|
    +
    V_z
    <
    \tau.
    \label{eq:delta-core-gap}
\end{align}

This completes the construction.

\subsection{Analysis}

We first compute the maximin shares of the agents.
We then prove that every MMS allocation encodes a maximum-score solution of the source instance whose final incidence bit is one.
Using this, we prove the correctness of the reduction.

\subsubsection{MMS values of the agents}

We first compute the maximin shares of all agents except the lex-opt indicator.
For these agents, the row or column partitions of the core goods provide the required lower bounds, while the averaging upper bound establishes their optimality.

\begin{lemma}
\label{lem:delta-non-indicator-MMS}
The following statements hold.
\begin{enumerate}
\item
For the puzzle-solving agent $p$, we have $\mms_p=\tau$.
\item
For every row-forcing agent $q$, we have $\mms_q=\tau$.
\item
For the anchor agent $f$, we have $\mms_f=\tau$.
\end{enumerate}
\end{lemma}

\begin{proof}
First, consider the puzzle-solving agent $p$.
By \Cref{obs:delta-puzzle-values}, the first four columns of the core have value exactly $\tau$, whereas the last column has value $\tau-V_P$.
We consider the column partition of the core goods and add all assignment goods to the bundle containing the last column.
Since the total value of the assignment goods to $p$ is exactly $V_P$, this bundle now has value $\tau$.
Thus, we obtain a partition of all goods into five bundles, each of value exactly $\tau$.
Therefore, $\mms_p\ge\tau$.

Moreover, by \Cref{obs:delta-puzzle-values}, the total value of all goods to $p$ is exactly $5\tau$.
Hence, the averaging upper bound gives $\mms_p\le\tau$.
Thus, $\mms_p=\tau$.

Next, consider a row-forcing agent $q$.
By \Cref{obs:delta-row-forcing-values}, every row of the core has value exactly $\tau$ to $q$.
Since $q$ values every assignment good at zero, the row partition of the core goods, together with an arbitrary distribution of the assignment goods, witnesses $\mms_q\ge\tau$.

Moreover, by \Cref{obs:delta-row-forcing-values}, the total value of all goods to $q$ is exactly $5\tau$.
Therefore, the averaging upper bound gives $\mms_q\le\tau$.
Thus, $\mms_q=\tau$.

Finally, consider the anchor agent $f$.
By \Cref{obs:delta-anchor-values}, every column of the core has value exactly $\tau$ to $f$.
Since $f$ values every assignment good at zero, the column partition of the core goods witnesses $\mms_f\ge\tau$.

Moreover, by \Cref{obs:delta-anchor-values}, the total value of all goods to $f$ is exactly $5\tau$.
Hence, the averaging upper bound gives $\mms_f\le\tau$.
Thus, $\mms_f=\tau$.
\end{proof}

We next compute the maximin share of the lex-opt indicator.
Unlike the other agents, her maximin share depends on the source instance.
The key property of our construction is that her maximin share is determined by the maximum score of a feasible source solution.

\begin{lemma}
\label{lem:delta-indicator-MMS}
For the lex-opt indicator, we have 
$
    \mms_{a^*}=\tau+\mathrm{OPT}.
$
\end{lemma}

\begin{proof}
We prove the lemma in two parts.
We first establish the lower bound by constructing a suitable partition and then prove the matching upper bound by analyzing an arbitrary partition.

\begin{claim-inside-lemma}
\label{clm:delta-indicator-MMS-lower}
\[\mms_{a^*}\ge\tau+\mathrm{OPT}.\]
\end{claim-inside-lemma}

\begin{proof}
Let $\solS$ be a maximum-score solution of the source instance.
By \Cref{lem:delta-consistency-encoding}, there exists a set $S\subseteq M_{\mathrm{asg}}$ satisfying
    $c(S)=B$
    and
    $\omega(S)=\mathrm{OPT}$.

Consider the column partition of the core goods.
We add the assignment goods in $S$ to the first column and all remaining assignment goods to the second column.
Thus, we obtain the following partition of all goods:
\[
    C_1\cup S,
    \quad
    C_2\cup(M_{\mathrm{asg}}\setminus S),
    \quad
    C_3,
    \quad
    C_4,
    \quad
    C_5.
\]

We now compute the values of these bundles to the lex-opt indicator.
By \Cref{obs:delta-indicator-values}, the first column has value $\tau+\beta_1=\tau+F-KB$.
Moreover, the assignment goods in $S$ contribute $Kc(S)-\omega(S)$.
Hence, the value of the first bundle is
\begin{align*}
    v_{a^*}(C_1\cup S)
    =
    \tau+F-KB+Kc(S)-\omega(S)
    =
    \tau+F-\mathrm{OPT},
\end{align*}
where the second equality follows from $c(S)=B$ and $\omega(S)=\mathrm{OPT}$.

Similarly, the second column has value
$
    \tau+\beta_2
    =
    \tau+W_{\mathrm{sc}}-K(C_{\mathrm{tot}}-B).
$
The assignment goods outside $S$ contribute
$
    K(C_{\mathrm{tot}}-c(S))
    -
    (W_{\mathrm{sc}}-\omega(S)).
$
Therefore, the value of the second bundle is
\begin{align*}
    v_{a^*}
    \bigl(C_2\cup(M_{\mathrm{asg}}\setminus S)\bigr)
    =
    \tau+W_{\mathrm{sc}}-K(C_{\mathrm{tot}}-B)
    +K(C_{\mathrm{tot}}-c(S))
    -(W_{\mathrm{sc}}-\omega(S))
    =
    \tau+\mathrm{OPT}.
\end{align*}

Finally, by \Cref{obs:delta-indicator-values}, each of the remaining three columns has value exactly $\tau+F$.

Recall that $\mathrm{OPT}\le W_{\mathrm{sc}}$ and $F=2W_{\mathrm{sc}}+1$.
Hence, we have $F-\mathrm{OPT}>\mathrm{OPT}$.
Consequently, the second bundle is the unique bundle of minimum value, and its value is exactly $\tau+\mathrm{OPT}$.
Thus, the constructed partition witnesses the fact that
$
    \mms_{a^*}\ge\tau+\mathrm{OPT}.
$
\end{proof}

We now establish the reverse inequality.
The main idea is to show that any partition in which every bundle has value at least $\tau$ must allocate the core goods by columns.
The values of the first two column bundles then force the assignment goods in the first bundle to encode a feasible source solution.

\begin{claim-inside-lemma}
\label{clm:delta-indicator-MMS-upper}
\[\mms_{a^*}\le\tau+\mathrm{OPT}.\]
\end{claim-inside-lemma}

\begin{proof}
Consider an arbitrary partition $\mathcal P$ of all goods into five bundles, viewed from the valuation of the lex-opt indicator.
We show that some bundle of $\mathcal P$ has value at most $\tau+\mathrm{OPT}$.

If some bundle has value less than $\tau$, the desired upper bound follows immediately, since $\mathrm{OPT}\ge0$.
We may therefore assume that every bundle has value at least $\tau$.
By \Cref{eq:delta-core-gap} and \Cref{prop:KPW-core-construction}, the partition induced by $\mathcal P$ on the core goods must be either the row partition or the column partition, up to relabeling.

We first rule out the row partition.
By \Cref{obs:delta-indicator-values}, the second core row has value
$
    v_{a^*}(R_2)=\tau-U.
$
Even if all assignment goods are added to the bundle containing $R_2$, its value is at most
\begin{align*}
    \tau-U+v_{a^*}(M_{\mathrm{asg}})
    =
    \tau-U+KC_{\mathrm{tot}}-W_{\mathrm{sc}}
    =
    \tau-F-W_{\mathrm{sc}}-1
    <
    \tau,
\end{align*}
where we used $U=KC_{\mathrm{tot}}+F+1$.
This contradicts our assumption that every bundle has value at least $\tau$.
Hence, the core goods must be allocated by columns.

Let $S\subseteq M_{\mathrm{asg}}$ denote the assignment goods in the bundle containing the first column $C_1$.
We first show that $c(S)=B$.

By \Cref{obs:delta-indicator-values}, the value of the bundle containing $C_1$ is
\begin{align}
    v_{a^*}(C_1\cup S)
    =
    \tau+\beta_1+Kc(S)-\omega(S)
    =
    \tau+F+K(c(S)-B)-\omega(S).
    \label{eq:delta-indicator-first-column}
\end{align}
Suppose that $c(S)\le B-1$.
Since $\omega(S)\ge0$, the value of this bundle is at most
$
    \tau+F-K<\tau,
$
where the last inequality follows from $K>F$.
This contradicts our assumption.
Hence,
$c(S)\ge B$. 

Next, consider the bundle containing the second column $C_2$.
The assignment goods in this bundle form a subset of $M_{\mathrm{asg}}\setminus S$.
Since the lex-opt indicator values every assignment good positively, the value of this bundle is at most what she would receive if all assignment goods outside $S$ were added to $C_2$.
By \Cref{obs:delta-indicator-values}, this upper bound is
\begin{align}
    v_{a^*}
    \bigl(C_2\cup(M_{\mathrm{asg}}\setminus S)\bigr)
    =
    \tau+\beta_2
    +K(C_{\mathrm{tot}}-c(S))
    -(W_{\mathrm{sc}}-\omega(S))
    =
    \tau+K(B-c(S))+\omega(S).
    \label{eq:delta-indicator-second-column}
\end{align}

Suppose that $c(S)\ge B+1$.
Since $\omega(S)\le W_{\mathrm{sc}}$, this upper bound is at most
$
    \tau-K+W_{\mathrm{sc}}<\tau,
$
where the last inequality follows from $K>W_{\mathrm{sc}}$.
This is again a contradiction.
Therefore,
$
    c(S)\le B.
$
Combining the inequalities obtained from the first two column bundles, we obtain $c(S)=B$. 

By \Cref{lem:delta-consistency-encoding}, the assignment goods in $S$ encode a feasible solution $\solS$ of the source instance.
Moreover, the score of this solution is $\omega(S)$.
By the definition of $\mathrm{OPT}$, we have
$
    \omega(S)=\score(\solS)\le\mathrm{OPT}.
$
Substituting $c(S)=B$ into \cref{eq:delta-indicator-second-column}, we conclude that the value of the bundle containing $C_2$ is at most
$
    \tau+\omega(S)
    \le
    \tau+\mathrm{OPT}.
$

Thus, every partition of all goods contains a bundle of value at most $\tau+\mathrm{OPT}$.
Consequently, we have
$
    \mms_{a^*}\le\tau+\mathrm{OPT}.
$
\end{proof}

Together, \Cref{clm:delta-indicator-MMS-lower,clm:delta-indicator-MMS-upper} establish that
$
    \mms_{a^*}=\tau+\mathrm{OPT}.
$
\end{proof}

\subsubsection{Structure of an MMS allocation}

We next establish the structural properties of allocations satisfying every agent's maximin share.

\begin{lemma}
\label{lem:delta-core-row-forcing}
In any allocation in which every agent receives value at least $\tau$, the KPW core goods are allocated by rows, up to relabeling.
\end{lemma}

\begin{proof}
Consider the allocation of the core goods induced by such an allocation.
Suppose that this induced core allocation is neither the row partition nor the column partition.

By \Cref{prop:KPW-core-construction}, some core bundle has unscaled $A$-value at most $\rowsum-1$.
By \cref{eq:delta-core-gap}, the recipient of this bundle receives value strictly less than $\tau$, even after adding any assignment goods.
This contradicts the assumption.

Thus, the core goods must be allocated either by rows or by columns.
It remains to rule out the column partition.

There are two row-forcing agents.
Under the column partition, by \Cref{obs:delta-row-forcing-values}, a row-forcing agent values every non-last column at $\tau-1$ and the last column at $\tau+4$.
Moreover, row-forcing agents value every assignment good at zero.
Therefore, at most one row-forcing agent can receive value at least $\tau$ under a column partition.
This contradicts the assumption that every agent receives value at least $\tau$.
Hence, the core goods must be allocated as rows, up to relabeling.
\end{proof}

We now determine which rows must be allocated to the anchor and the lex-opt indicator.

\begin{lemma}
\label{lem:delta-allocation-structure}
In every MMS allocation, the anchor agent receives the last row $R_5$, and the lex-opt indicator receives the first row $R_1$.
Consequently, the puzzle-solving agent receives one of the  other rows $R_2,R_3,R_4$.
\end{lemma}

\begin{proof}
By \Cref{lem:delta-non-indicator-MMS,lem:delta-indicator-MMS}, every agent has maximin share at least $\tau$.
Therefore, every agent receives value at least $\tau$ in an MMS allocation.
By \Cref{lem:delta-core-row-forcing}, the core goods are allocated as rows.

Consider the anchor agent $f$.
By \Cref{obs:delta-anchor-values}, she values every non-last row at $\tau-1$ and values the last row at $\tau+4$.
Moreover, she values every assignment good at zero.
Since $\mms_f=\tau$, she must receive the last row $R_5$.

Consequently, the lex-opt indicator receives a non-last row.
By \Cref{obs:delta-indicator-values}, each of the rows $R_2,R_3,R_4$ has value $\tau-U$ to the lex-opt indicator.
Even after adding all assignment goods, such a row has value at most
$
    \tau-U+KC_{\mathrm{tot}}-W_{\mathrm{sc}}
    <
    \tau.
$
Since $\mms_{a^*}\ge\tau$, the lex-opt indicator cannot receive any of these rows.
Therefore, she must receive the first row $R_1$.

The remaining rows $R_2,R_3,R_4$ are allocated to the puzzle-solving agent and the two row-forcing agents, one row to each.
Thus, the puzzle-solving agent receives one of these three rows.
\end{proof}

\subsubsection{Correctness of the reduction}

We now establish the correctness of the reduction.

\begin{lemma}
\label{lem:delta2p-additive-reduction-correctness}
The constructed goods instance admits an MMS allocation if and only if the last bit of the lexicographically maximum solution of the source \lexmaxfourpartitionpromise instance is one.
\end{lemma}

\begin{proof}
Suppose first that the constructed goods instance admits an MMS allocation.
By \Cref{lem:delta-allocation-structure}, the puzzle-solving agent receives one of the rows $R_2,R_3,R_4$, while the lex-opt indicator receives $R_1$.

Let $S\subseteq M_{\mathrm{asg}}$ be the set of assignment goods received by the puzzle-solving agent.
Her core row has value $\tau-B$, and by \Cref{lem:delta-non-indicator-MMS}, her maximin share is exactly $\tau$.
Therefore,
$
    \tau-B+c(S)-\varepsilon(S)
    \ge
    \tau.
$
Equivalently, we have
\begin{equation}
    c(S)-\varepsilon(S)\ge B.
    \label{eq:delta-puzzle-necessary}
\end{equation}

Next, consider the lex-opt indicator.
She receives the core row $R_1$ and some subset of the assignment goods not received by the puzzle-solving agent.
Since she values every assignment good positively, her bundle has value at most what she would receive if she were given all assignment goods outside $S$.
By the definition of her valuations, this upper bound is
\begin{align*}
    \tau+W_{\mathrm{sc}}-K(C_{\mathrm{tot}}-B)
    +
    K(C_{\mathrm{tot}}-c(S))
    -
    (W_{\mathrm{sc}}-\omega(S))
    =
    \tau+K(B-c(S))+\omega(S).
\end{align*}

By \Cref{lem:delta-indicator-MMS}, the lex-opt indicator has maximin share $\tau+\mathrm{OPT}$.
Therefore,
\begin{equation}
    K(B-c(S))+\omega(S)
    \ge
    \mathrm{OPT}.
    \label{eq:delta-indicator-necessary}
\end{equation}

We now combine \cref{eq:delta-puzzle-necessary} and \cref{eq:delta-indicator-necessary}.
Since $\varepsilon(S)\ge0$, the first inequality gives
$
    c(S)\ge B.
$
Suppose that $c(S)>B$.
Since the consistency numbers are integral, we have $c(S)\ge B+1$.
Consequently,
\[
    K(B-c(S))+\omega(S)
    \le
    -K+W_{\mathrm{sc}}
    <
    0
    \le
    \mathrm{OPT},
\]
contradicting \cref{eq:delta-indicator-necessary}.
Therefore, we have $c(S)\le B$, and consequently, it follows that $c(S)=B$. 

Substituting this equality into \cref{eq:delta-puzzle-necessary} gives
$
    \varepsilon(S)\le0.
$
Since the penalties are nonnegative, we obtain
$
    \varepsilon(S)=0.
$
By \Cref{lem:delta-consistency-encoding}, the set $S$ represents a feasible solution $\solS$ of the source instance.
Moreover, $\varepsilon(S)=0$ implies that the final incidence bit of this solution is one.

Finally, substituting $c(S)=B$ into \eqref{eq:delta-indicator-necessary} gives
$
    \omega(S)\ge\mathrm{OPT}.
$
But $\omega(S)=\score(\solS)$, and no feasible source solution has score greater than $\mathrm{OPT}$.
Therefore, we have
$
    \score(\solS)=\mathrm{OPT}.
$

Thus, $\solS$ is a maximum-score solution with final incidence bit one.
Since maximizing the score is equivalent to lexicographically maximizing the incidence string, the original \lexmaxfourpartitionpromise instance is a yes-instance.

Conversely, suppose that the source instance is a yes-instance.
Let $\solS^\star=(S_1^\star,\dots,S_\ell^\star)$ be a lexicographically maximum solution whose final incidence bit is one.

Let
$
    S^\star
    =
    \{h_{i,j}:a_i\in S_j^\star,\ j\in[\ell]\}.
$
By \Cref{lem:delta-consistency-encoding} and the definition of the penalties, we have 
    $c(S^\star)=B$,
    $\varepsilon(S^\star)=0$,
    and
    $\omega(S^\star)=\mathrm{OPT}$.

We construct an MMS allocation as follows.
We allocate the core goods by rows.
\begin{itemize}

\item
We give the last row $R_5$ to the anchor agent $f$.
She receives value $\tau+4\ge\tau=\mms_f$.

\item
We give the first row $R_1$ to the lex-opt indicator $a^*$, together with every assignment good not belonging to $S^\star$.
She receives value
\begin{align*}
    v_{a^*}
    \bigl(R_1\cup(M_{\mathrm{asg}}\setminus S^\star)\bigr)
    =
    \tau+K(B-c(S^\star))+\omega(S^\star)
    =
    \tau+\mathrm{OPT}
    =
    \mms_{a^*}.
\end{align*}

\item
We give the second row $R_2$ to the puzzle-solving agent $p$, together with all assignment goods in $S^\star$.
She receives value
\begin{align*}
    v_p(R_2\cup S^\star)
    =
    \tau-B+c(S^\star)-\varepsilon(S^\star)
    =
    \tau
    =
    \mms_p.
\end{align*}

\item
We give the remaining rows $R_3,R_4$ to the two row-forcing agents, one row to each.
Every row-forcing agent values every row at exactly $\tau$, which is her maximin share.

\end{itemize}

This allocates every core good and every assignment good exactly once.
Every agent receives value at least her MMS.
Therefore, the constructed instance admits an MMS allocation.
\end{proof}

We now prove the last ingredient of the reduction.
\begin{lemma}\label{lem:polynomial-size}
    The reduction can be executed in polynomial time.
\end{lemma}

\begin{proof}
By \Cref{lem:lexmaxfourpartitionpromise-delta2p-hardness}, the source problem remains $\Delta_2^P$-hard when $T$ and every source number are polynomially bounded in $\ell$.
Since $m=4\ell$, the base
$D=\ell T+\ell+1$
is polynomially bounded in the input size.

Every consistency number $c_{i,j}$ and the target $B$ are sums of powers of $D$ with exponents at most $m+\ell-1$.
Consequently, their binary representations have polynomial length.
The same holds for $C_{\mathrm{tot}}$ and $V_P$.

The score uses $R=4\ell^2$ binary positions, so $W_{\mathrm{sc}}=2^R-1$ has $\bigoh(\ell^2)$ bits.
The parameters $K,F,U$ are obtained from these quantities using a constant number of additions and multiplications.
Thus, they also have polynomial bit length.

By \Cref{prop:KPW-core-construction}, the KPW core matrix $A$ is a fixed $5\times5$ integer matrix with exactly $19$ positive entries.
Every entry of $A$, as well as $\rowsum$, has constant bit length.

All the perturbations $\delta_z(r,s)$ are obtained from the previously defined parameters by polynomially many additions and multiplications.
The scaling factor $\Gamma$ is computed from these perturbations and the total assignment-good values of the five agents.
Therefore, $\Gamma$ and every constructed valuation have polynomial bit length.

Finally, the constructed instance contains exactly five agents and
$
    19+m\ell
    =
    19+4\ell^2
$
goods.
All valuation functions are nonnegative additive integer valuations.
Hence, the entire instance can be constructed and encoded in polynomial time.
\end{proof}

Now we have all the ingredients to prove our final result. 

\begin{theorem}
\label{thm:delta2p-five-additive}
Deciding whether an additive goods instance admits an MMS allocation is $\Delta_2^P$-complete, even for five agents.
\end{theorem}

\begin{proof}
The membership of the problem in $\Delta_2^P$ follows from \Cref{prop:circuit-membership-delta2p}.

By \Cref{lem:delta2p-additive-reduction-correctness}, the constructed goods instance admits an MMS allocation if and only if the last bit of the lexicographically maximum solution of the source \lexmaxfourpartitionpromise instance is one.
By \Cref{lem:polynomial-size} we know that the reduction can be run in polynomial time. 
The reduction therefore proves $\Delta_2^P$-hardness even for five agents.
Together with membership in $\Delta_2^P$, this establishes the theorem.
\end{proof}


\subsection*{Role of AI in this Section}
In an earlier version of this work, following the proofs of \NP-hardness and \DP-hardness, ChatGPT designed a gadget yielding a hardness proof for $BH_3$, the third level of the Boolean hierarchy.
By examining this gadget, we observed how to extend the construction to obtain $\Theta_2^P$-hardness.
We subsequently developed a complete reduction strategy for $\Delta_2^P$-hardness (with an unbounded number of agents) and conveyed it to ChatGPT, which determined the parameters realizing the reduction.
Specifically, its assistance in designing the puzzle scale parameters and in scaling the different parameters so that the resulting constraints were simultaneously satisfied was important in completing that construction.

While working on another problem related to MMS we found a much simpler \NP-hardness proof with five agents. 
That prompted us to extend our framework of the $\Delta_2^P$-hardness proof--that built on the earlier \NP-hardness gadget--using the principles of the new \NP-hardness gadget. 
This yielded a much more streamlined and simpler proof of the main $\Delta_2^P$-hardness result. 
ChatGPT determined the specific scaling parameters and the perturbations of the core matrix, following the same principles used in the new \NP-hardness proof, which is presented in this version.

\hide{In the current version, following the new proof of \NP-hardness, we suggested a similar construction from \lexmaxfourpartitionpromise for establishing $\Delta_2^P$-hardness for five agents.}

%% file: 7_conp/conp-hardness-goods.tex
\section{\coNP-hardness}
\label{sec:conp}

\subsection{Weak \coNP-hardness when $n=3$}

To obtain the \coNP-hardness result, we reduce from the following weakly \NPH problem, which we call \threewaypartition. 
The input is a multiset $W=\{w_1, \dots, w_m\}$ of positive integers and an integer $T$ such that $\sum_{j = 1}^m w_j = 3T$.
The question is whether $W$ can be partitioned into three parts, each of sum exactly $T$. 
This problem is weakly \NPH by a straightforward reduction from $\textsc{Partition}$: given a partition instance with total sum $2T$, add the number $T$ to the multiset.

We now describe the reduction. 
We start with the 3-agent 9-item no-instance discussed in \Cref{subsec:FST-no-instance-3-agents}. 
We call its nine items the \emph{core items}, and denote their set by $E_{\mathrm{core}}$. 
For every number $w_j \in W$, we add one new good $x_j$. 
We call these new goods the \emph{puzzle items}, and denote their set by $E_{\mathrm{puzz}}$.
For a bundle $A$ of goods, we define its \emph{puzzle weight} by
\[
\omega(A) = \sum_{x_j \in A \cap E_{\mathrm{puzz}}} w_j.
\]

All three agents value the puzzle items identically. 
For every agent $i \in \{R, C, U\}$, we define
\[
    v_i(e) = M_i[e] \qquad \text{for every } e \in E_{\mathrm{core}},
\]
where $M_i[e]$ is the entry of the corresponding matrix in \eqref{eqn:FST-no-instance-3-agents}, and
\[
    v_i(x_j) = 121~w_j \qquad \text{for every } x_j \in E_{\mathrm{puzz}}.
\]
The valuations are additive.

\begin{lemma}
    \label{lem:coNP-reduction-correctness}
    The multiset $W$ admits a three-way partition into parts of sum $T$ if and only if the constructed instance does not admit an MMS allocation.
\end{lemma}

\begin{proof}
    For the forward direction, suppose that $W$ has a partition $(P_1, P_2, P_3)$ such that each $P_j$ has sum $T$. 
    We claim that every agent has maximin share exactly $121T + 40$.
    Indeed, for agent $R$, take the row partition of the core items, and append the puzzle bundle $P_j$ to the $j$-th row. 
    Each resulting bundle has value $40 + 121T$. 
    Thus, $\mms_R\ge 121T + 40$. 
    The same argument applies to $C$, using the column partition of the core items, and to $U$, using its MMS partition in the Feige--Sapir--Tauber instance. 
    On the other hand, for every agent the total value of all items is $120 + 121\cdot 3T = 3(121T + 40)$, so no agent can have MMS larger than $121T + 40$. 
    Hence
    \[
        \mms_R = \mms_C = \mms_U = 121T + 40.
    \]

    We now show that no MMS allocation exists. 
    Consider an arbitrary allocation $(A_R, A_C, A_U)$. 
    For every agent $i \in \{R, C, U\}$, let $p_i = \omega(A_i)$ be the puzzle weight received by agent $i$.
    Since $p_R + p_C + p_U = 3T$, either some agent receives puzzle weight at most $T - 1$, or all three agents receive puzzle weight exactly $T$.

    In the first case, suppose agent $i$ receives puzzle weight at most $T - 1$. 
    Even if $i$ received all core items, their total value would be at most
    \[
        121(T - 1) + 120
        =
        121T - 1
        <
        121T + 40
        =
        \mms_i.
    \]
    Thus, agent $i$ does not receive their MMS value and the allocation is not an MMS allocation.

    In the second case, every agent receives puzzle weight exactly $T$. 
    Then, in order to reach their MMS $121T + 40$, every agent must receive core value at least $40$. 
    But the allocation of the core items induced by $(A_R, A_C, A_U)$ is an allocation of the nine-item instance in \eqref{eqn:FST-no-instance-3-agents}, and by \Cref{prop:FST21-3-agent}, some agent receives core value at most $39$. 
    That agent receives total value at most
    \[
        121T + 39
        <
        121T + 40.
    \]
    Therefore, the constructed instance does not admit an MMS allocation.

    For the reverse direction, suppose that $W$ does not admit a three-way partition into parts of sum $T$. 
    The value of the puzzle items in a bundle $A$ is $121\omega(A)$ for every agent.

    Let
    \[
        \alpha
        =
        \max_{(Q_1, Q_2, Q_3) \in \Pi_3(W)}
        \min_{j \in \{1,2,3\}}
        \sum_{w_k \in Q_j} w_k.
    \]
    Since $W$ has total sum $3T$ and has no partition into three parts of sum exactly $T$, we have  $\alpha \le T-1$.

    We choose a partition $(Q_1, Q_2, Q_3)$ attaining the value $\alpha$, and among all such partitions choose one minimizing the number of parts whose sum is exactly $\alpha$. 
    Let this number be $r$. 
    Since the total sum is $3T$ and $\alpha < T$, we have $r \in \{1, 2\}$.
    We now distinguish two cases based on the value of $r$.
    The role of the following case distinction is to identify how many bottleneck puzzle bundles must be supported by core items.

    \paragraph{Case 1: $r=1$.}
    We first show that, for every agent $i\in\{R, C, U\}$,
    \[
        \mms_i \le 121\alpha + 120.
    \]
    Consider an arbitrary three-partition $(A_1, A_2, A_3)$ of all items. 
    The three numbers $\omega(A_1)$, $\omega(A_2)$, and $\omega(A_3)$ form the weights of a three-way partition of the puzzle items. 
    By the definition of $\alpha$, at least one bundle has puzzle weight at most $\alpha$.

    If some bundle has puzzle weight at most $\alpha - 1$, then even after adding all core items to that bundle, its value to any agent is at most
    \[
        121(\alpha - 1) + 120 < 121\alpha \le 121\alpha + 120.
    \]
    Otherwise, the minimum puzzle weight is exactly $\alpha$. 
    Then some bundle with puzzle weight exactly $\alpha$ has core value at most $120$, and hence total value at most $121\alpha + 120$. 
    Thus, $\mms_i \le 121\alpha + 120$ for every agent $i$.

    Now we allocate the items as follows. 
    We give the unique part among $(Q_1, Q_2, Q_3)$ whose sum is $\alpha$, together with all core items, to one arbitrary agent. 
    We give the other two puzzle parts to the other two agents. 
    The first agent receives value exactly $121\alpha + 120$, while each other agent receives puzzle weight at least $\alpha + 1$, and hence value at least
    \[
        121(\alpha+1) = 121\alpha + 121 > 121\alpha + 120.
    \]
    Therefore, every agent receives at least their MMS.

    \paragraph{Case 2: $r=2$.}
    For every agent $i \in \{R, C, U\}$, let $\beta_i$ denote their two-agent maximin share on the core items:
    \[
        \beta_i
        =
        \max_{(B_1, B_2) \in \Pi_2(E_{\mathrm{core}})}
        \min\{v_i(B_1), v_i(B_2)\}.
    \]
    Since the total core value of every agent is $120$, we have $\beta_i \le 60$.

    We claim that, for every agent $i$,
    \[
        \mms_i \le 121\alpha + \beta_i.
    \]
    Consider an arbitrary three-partition $(A_1 ,A_2, A_3)$ of all goods. 
    If some bundle has puzzle weight at most $\alpha - 1$, then that bundle has value at most
    \[
        121(\alpha - 1) + 120 < 121\alpha \le 121\alpha + \beta_i.
    \]
    Otherwise, the minimum puzzle weight is exactly $\alpha$. 
    Since $r=2$, every partition of the puzzle items whose minimum weight is $\alpha$ has at least two parts of weight exactly $\alpha$. 
    Hence at least two of the bundles $A_1, A_2, A_3$ have puzzle weight exactly $\alpha$.

    Consider the core items lying in these two bundles. 
    The smaller core value among the two bundles is at most $\beta_i$. 
    Otherwise, by moving all remaining core items into these two bundles, we would obtain a two-partition of the core items in which both bundles have value more than $\beta_i$, contradicting the definition of $\beta_i$. 
    Therefore, some bundle has total value at most $121\alpha + \beta_i$, and so $\mms_i \le 121\alpha + \beta_i$.

    We now construct an MMS allocation. 
    In the chosen partition $(Q_1, Q_2, Q_3)$, exactly two parts have sum $\alpha$, and the third part has sum at least $\alpha + 1$. 
    We give the two bottleneck puzzle parts to two agents, say $i$ and $j$.
    
    Let $(H_1,H_2)$ be a two-partition of the core goods witnessing $\beta_i$ for agent $i$; that is, $v_i(H_1) \ge \beta_i$ and $v_i(H_2) \ge \beta_i$.
    Agent $j$ chooses their preferred bundle among $H_1$ and $H_2$, and agent $i$ receives the other bundle.
    By the choice of $(H_1, H_2)$, agent $i$ receives core value at least $\beta_i$.
    Moreover, agent $j$ receives a bundle of value at least half of their total value for the core goods.
    Since no two-agent maximin share can exceed half of the total value, this value is at least $\beta_j$.
    Thus, agent $j$ also receives core value at least $\beta_j$.
    
    We give the remaining puzzle part to the third agent, say $k$.
    Agents $i$ and $j$ receive values at least $121\alpha + \beta_i$ and $121\alpha + \beta_j$, respectively. 
    Agent $k$ receives puzzle value at least $121(\alpha + 1) = 121\alpha + 121$.
    Since $\beta_k \le 60$, this is strictly larger than $121\alpha + \beta_k$. 
    Thus, all three agents receive at least their MMS values.

    Therefore, if $W$ has no three-way partition into parts of sum $T$, the constructed instance admits an MMS allocation.
\end{proof}

\weakconphardness*

\begin{proof}
    By \Cref{lem:coNP-reduction-correctness}, the constructed instance admits an MMS allocation if and only if the \threewaypartition instance is a no-instance. 
    Since \threewaypartition is weakly \NPH, deciding existence of an MMS allocation is weakly \coNPH.
\end{proof}

\subsection{Strong \coNP-hardness}

We now prove strong \coNP-hardness using the FST core.

We reduce from the strongly \NPH problem \threepartition.
An instance of \threepartition consists of a multiset $W = \{w_1, \ldots, w_{3n}\}$
of positive integers and an integer $T$ such that $\sum_{j = 1}^{3n} w_j = nT$ and $\frac{T}{4} < w_j < \frac{T}{2}$ for every $j \in [3n]$.
The question is whether $W$ can be partitioned into $n$ triples, each of sum exactly $T$.
We assume that $n \ge 4$.

We construct an additive goods instance as follows.
We start with the FST core instance from \Cref{sec:FST-core}, with agent set $Q = R \uplus C$, where $|R| = n - 2$ and $|C| = 2$.
Recall that $R$ is the set of row perturbation agents and $C$ is the set of column perturbation agents.
For every number $w_j \in W$, add one new good $x_j$.
We call these goods the \emph{puzzle items}, and denote their set by $E_{\mathrm{puzz}}$.
For a bundle $A$ of goods, we define its \emph{puzzle weight} by $\omega(A) = \sum_{x_j \in A \cap E_{\mathrm{puzz}}} w_j$.

All $n$ agents value the puzzle items identically.
Let $L = n\tau + 1$.
For every agent $i \in Q$, we keep their core valuation unchanged, and define
\[
    v_i(x_j) = Lw_j
    \qquad\text{for every } j \in [3n].
\]
The valuations are additive.

\begin{lemma}
    \label{lem:strong-coNP-reduction-correctness}
    The \threepartition instance is a yes-instance if and only if the constructed instance does not admit an MMS allocation.
\end{lemma}

\begin{proof}
    For the forward direction, suppose that $W$ has a partition $(P_1, \ldots, P_n)$ such that each $P_j$ has sum exactly $T$.
    We claim that every agent has maximin share exactly $LT + \tau$.

    Fix an agent $i \in Q$.
    By \Cref{prop:FST-core-no-instance}, agent $i$ has an $n$-partition $(G_1, \ldots, G_n)$ of the core goods such that $v_i(G_j) \ge \tau$ for every $j \in [n]$.
    We append the puzzle bundle $P_j$ to the core bundle $G_j$.
    Since each $P_j$ has puzzle weight $T$, each resulting bundle has value at least $LT + \tau$.
    Thus, we have $\mms_i \ge LT + \tau$.
    On the other hand, the total value of all goods for agent $i$ is
    \[
        n\tau + L\sum_{j = 1}^{3n} w_j
        =
        n\tau + LnT
        =
        n(LT + \tau).
    \]
    Therefore, no agent can have maximin share larger than $LT + \tau$.
    Combining the lower and upper bounds, we have $\mms_i = LT + \tau$ for every $i \in Q$.

    We now show that no MMS allocation exists.
    Consider an arbitrary allocation $(A_1, \ldots, A_n)$ of all goods.
    For every agent $i \in Q$, let $p_i = \omega(A_i)$ be the puzzle weight received by agent $i$.
    Since $\sum_{i = 1}^n p_i = nT$,
    either some agent receives puzzle weight at most $T - 1$, or every agent receives puzzle weight exactly $T$.

    In the first case, suppose agent $i$ receives puzzle weight at most $T - 1$.
    Even if agent $i$ received all core goods, their total value would be at most
    \[
        L(T - 1) + n\tau
        =
        LT - L + n\tau
        =
        LT - 1
        <
        LT + \tau
        =
        \mms_i.
    \]
    Thus, agent $i$ does not receive their MMS value and the allocation is not an MMS allocation.

    In the second case, every agent receives puzzle weight exactly $T$.
    Then, in order to reach their maximin share of $LT + \tau$, every agent must receive core value at least $\tau$.
    But by \Cref{prop:FST-core-no-instance}, no allocation of the core goods gives every agent value at least $\tau$.
    Hence, some agent receives value strictly less than $LT + \tau$.
    Therefore, the constructed instance does not admit an MMS allocation.

    For the reverse direction, suppose that the \threepartition instance is a no-instance.
    The value of the puzzle items in a bundle $A$ is $L\omega(A)$ for every agent. Let
    \[
        \alpha
        =
        \max_{(Y_1, \ldots, Y_n) \in \Pi_n(W)}
        \min_{t \in [n]}
        \sum_{w_j \in Y_t} w_j.
    \]

    Given the range of each number, $T/4 < w_j < T/2$, every bundle of sum exactly $T$ must contain exactly three numbers.
    Hence, any partition into $n$ parts of sum exactly $T$ would be a valid solution to the original \threepartition instance.
    Since the \threepartition instance is a no-instance, no such partition exists.
    As $W$ has total sum $nT$ and the numbers are integral, we have $\alpha \le T - 1$.

    We choose a partition $(Y_1, \ldots, Y_n)$ attaining the value $\alpha$, and among all such partitions choose one minimizing the number of parts whose sum is exactly $\alpha$.
    Let this number be $r$.
    Since the total sum is $nT$ and $\alpha < T$, we have $1 \le r \le n - 1$.
    We relabel the parts so that
    \[
        \sum_{w_j \in Y_t} w_j = \alpha
        \quad\text{for every } t \in [r],
    \qquad
    \text{and}
    \qquad
        \sum_{w_j \in Y_t} w_j \ge \alpha + 1
        \quad\text{for every } t \in \{r + 1, \ldots, n\}.
    \]
    For every agent $i \in Q$, we consider the $r$-agent maximin share on the core items:
    \[
        \beta_i^r
        =
        \max_{(H_1, \ldots, H_r) \in \Pi_r(E_{\mathrm{core}})}
        \min_{t \in [r]} v_i(H_t).
    \]
    Since the total core value of every agent is $n\tau$, we have $\beta_i^r \le n\tau$.

    We claim that, for every agent $i \in Q$,
    \[
        \mms_i \le L\alpha + \beta_i^r.
    \]
    Consider an arbitrary $n$-partition $(A_1, \ldots, A_n)$ of all goods.
    If some bundle has puzzle weight at most $\alpha - 1$, then that bundle has value at most
    \[
        L(\alpha - 1) + n\tau
        =
        L\alpha - L + n\tau
        =
        L\alpha - 1
        <
        L\alpha
        \le
        L\alpha + \beta_i^r.
    \]
    Otherwise, the minimum puzzle weight is exactly $\alpha$.
    By the choice of $r$, every partition of the puzzle items whose minimum weight is $\alpha$ has at least $r$ parts of weight exactly $\alpha$.
    Hence, at least $r$ many bundles among $A_1, \ldots, A_n$ have puzzle weight exactly $\alpha$.

    Consider the core goods lying in these $r$ bundles.
    The minimum core value among these $r$ bundles is at most $\beta_i^r$.
    Otherwise, by moving all remaining core goods into these $r$ bundles, we would obtain an $r$-partition of the core goods in which every part has value strictly larger than $\beta_i^r$, contradicting the definition of $\beta_i^r$.
    Therefore, some bundle has total value at most $L\alpha + \beta_i^r$, and so $\mms_i \le L\alpha + \beta_i^r$ for every $i \in Q$.

    We now show a simple feasibility property of the FST core that will be used to deal with the $r$ bottleneck puzzle bundles.

    \begin{claim-inside-lemma}
        \label{clm:FST-subcore-feasibility}
        For every $r \in [n - 1]$, there exists a set $S_r \subseteq Q$ of size $r$ such that the core goods can be allocated among the agents in $S_r$ so that every agent $i \in S_r$ receives value at least $\beta_i^r$.
    \end{claim-inside-lemma}

    \begin{proof}
        We distinguish two cases.

        \paragraph{Case 1: $r \le n - 2$.}
        We choose $S_r$ to be any set of $r$ row agents.
        All agents in $S_r$ have the same valuation $V_R$.
        Let $(H_1, \ldots, H_r)$ be an $r$-partition of the core goods attaining the $r$-agent maximin share for $V_R$.
        We assign these bundles arbitrarily to the agents in $S_r$.
        Since all agents in $S_r$ have valuation $V_R$, every agent $i \in S_r$ receives value at least $\beta_i^r$.

        \paragraph{Case 2: $r = n - 1$.}
        We choose $S_r$ to be the set of all $n - 2$ row agents and one column agent.
        Let $c$ be the column agent in $S_r$.
        Let $(H_1, \ldots, H_r)$ be an $r$-partition of the core goods attaining the $r$-agent maximin share for the row valuation $V_R$.
        Thus, every bundle $H_t$ has row value at least the $r$-agent maximin share of a row agent.

        We give the column agent $c$ their favorite bundle among $H_1, \ldots, H_r$.
        Since the total core value of $c$ is $n\tau$, their favorite bundle has value at least $\frac{n\tau}{r}$.
        On the other hand, $\beta_c^r \le \frac{n\tau}{r}$ because no $r$-agent maximin share can exceed the average value of the whole set of core goods.
        Therefore, $c$ receives value at least $\beta_c^r$.

        We arbitrarily assign the remaining $r - 1 = n - 2$ bundles to the row agents.
        Each remaining bundle has row value at least the $r$-agent maximin share of a row agent.
        Hence, every row agent in $S_r$ receives value at least their $\beta_i^r$.
    \end{proof}

    We now construct an MMS allocation.
    Let $S_r \subseteq Q$ be the set of $r$ agents promised by Claim~\ref{clm:FST-subcore-feasibility}.
    We give the $r$ bottleneck puzzle bundles $Y_1, \ldots, Y_r$ to the agents in $S_r$, one bundle per agent.
    By Claim~\ref{clm:FST-subcore-feasibility}, the core goods can be allocated among the agents in $S_r$ so that every agent $i \in S_r$ receives core value at least $\beta_i^r$.
    Therefore, every agent $i \in S_r$ receives total value at least $L\alpha + \beta_i^r \ge \mms_i$.

    We give the remaining puzzle bundles $Y_{r + 1}, \ldots, Y_n$ to the agents in $Q \setminus S_r$, one bundle per agent.
    Each such agent receives puzzle weight at least $\alpha + 1$, and hence receives value at least
    \[
        L(\alpha + 1)
        =
        L\alpha + L
        >
        L\alpha + n\tau
        \ge
        L\alpha + \beta_i^r
        \ge
        \mms_i.
    \]
    Thus, every agent receives at least their MMS value.
    Therefore, if the \threepartition instance is a no-instance, then the constructed instance admits an MMS allocation.
    
    This completes the proof.
\end{proof}

\strongconphardness*

\begin{proof}
    By \Cref{lem:strong-coNP-reduction-correctness}, the constructed instance admits an MMS allocation if and only if the \threepartition instance is a no-instance.

    All values in the FST core are polynomially bounded in $n$.
    Moreover, the scaling factor $L$ is polynomially bounded in $n$ because $L = n\tau + 1$, where $\tau = n^2(n - 2)^2 + n$.
    The puzzle-item values are $Lw_j$, and hence the construction only multiplies the input numbers by a polynomially bounded factor.
    Since \threepartition is strongly \NPH, deciding existence of an MMS allocation is strongly \coNPH.
\end{proof}

\subsection{Inapproximability}

The strong \coNP-hardness reduction also gives an approximation consequence.
In the no-case of \threepartition, the constructed instance admits an exact MMS allocation.
In the yes-case, every allocation leaves some agent short of their MMS value by at least one unit.
Since the constructed numerical values are polynomially bounded, this one-unit gap becomes an inverse-polynomial gap in the optimal MMS approximation ratio.

\nofptas*
\begin{proof}
    For an instance $I$ with agent set $Q$, let $\lambda_I$ be the largest $\lambda$ for which there exists an allocation giving every agent at least a $\lambda$-fraction of their MMS. 
    Equivalently,
    \[
        \lambda_I
        =
        \max_{A \in \Pi_n(M)}
        \min_{i \in Q}
        \frac{v_i(A_i)}{\mms_i}.
    \]
    
    Consider the instances constructed in the proof of \Cref{thm:strong-coNP-hardness}.
    Recall that $\tau = n^2(n - 2)^2 + n$ and $L = n\tau + 1$, from \Cref{thm:strong-coNP-hardness}.

    If the \threepartition instance is a no-instance, then by \Cref{lem:strong-coNP-reduction-correctness}, the constructed instance admits an MMS allocation.
    Hence, $\lambda_I \ge 1$.

    If the \threepartition instance is a yes-instance, then every agent has MMS exactly $LT + \tau$.
    Moreover, the proof of \Cref{lem:strong-coNP-reduction-correctness} shows that in every allocation, some agent receives value at most $LT + \tau - 1$.
    Therefore,
    \[
        \lambda_I
        \le
        \frac{LT + \tau - 1}{LT + \tau}
        =
        1 - \frac{1}{LT + \tau}.
    \]

    Thus, the reduction creates an inverse-polynomial gap between the two cases: $\lambda_I \ge 1$ in the no-case of \threepartition, and $\lambda_I \le 1 - 1/(LT + \tau)$ in the yes-case of \threepartition.
    Since \threepartition is strongly \NPH, it remains hard even when $T$ is polynomially bounded in $n$.
    Also, $\tau$ and $L$ are polynomially bounded in $n$.
    Hence, $LT + \tau$ is polynomially bounded in the size of the constructed instance.

    Suppose there were an FPTAS for computing $\lambda_I$.
    Consider running it with $\varepsilon = {1}/{(2(LT + \tau))}$.
    Since $1/\varepsilon$ is polynomially bounded, the running time would be polynomial.
    
    In the no-case, we have $\lambda_I \ge 1$, and hence a $(1 - \varepsilon)$-approximation returns a value at least $1 - \varepsilon = 1 - {1}/{(2(LT + \tau))}$.
    In the yes-case, every feasible allocation has ratio at most $1 - {1}/{(LT + \tau)} < 1 - {1}/{(2(LT + \tau))}$.
    Therefore, we would be able to distinguish the case $\lambda_I \ge 1$ from the case $\lambda_I \le 1 - {1}/{(LT + \tau)}$, thereby deciding \threepartition in polynomial time.
    This contradicts the strong \NP-hardness of \threepartition unless \P = \NP.
\end{proof}

%% file: 8_goods_to_chores/goods_to_chores.tex
\newcommand{\Cmath}{{\mathrm C}}

\section{Results in the Chores Setting}
\label{sec:goods-to-chores-transfer}

We now present a polynomial-time reduction from the goods setting to the chores setting, using which we can transfer our hardness results from the goods setting to the chores setting.
The idea is to pad the instance with zero-valued items so that size-wise balanced MMS allocations could exist, and then subtract a large cardinality-dependent term from the valuations.
The large term forces every MMS allocation in the constructed chores instance to give exactly the same number of items to every agent.
On such balanced allocations, the chore utility differs from the original goods utility only by an agent-independent constant.

\begin{lemma}
\label{lem:goods-to-chores-transfer}
Let $I$ be an additive goods instance with agent set $N$, item set $G$, and nonnegative valuations $\{v_i\}_{i \in N}$.
Let $|N| = n$ and $|G| = m$.
Then one can construct, in polynomial time, an additive chores instance $I^C$ with nonpositive utilities $\{u_i^C\}_{i \in N}$ such that
\[
    I \text{ admits an MMS allocation}
    \quad\Longleftrightarrow\quad
    I^C \text{ admits an MMS allocation}.
\]
Moreover, the construction preserves the number of agents and the number of agent types.
\end{lemma}

\begin{proof}
We first pad the goods instance.
Let $q = m$.
We add $(n - 1)m$ new dummy items so that the total number of items is exactly $nq$.
Each dummy item has value $0$ for every agent.
Let $H$ denote the padded item set.

This padding does not change any agent's maximin share in the goods instance.
Every partition of the original item set can be extended to a partition of $H$ by adding the dummy items, without changing the value of any bundle.
Thus, the MMS value cannot decrease.
Conversely, every partition of $H$ induces a partition of the original item set after discarding the dummy items, again without changing the value of any bundle.
Thus, the MMS value cannot increase.
Let $\mms_i^{\mathrm{goods}}$ denote agent $i$'s MMS value in the original goods instance, which is also the MMS value in the padded goods instance.

Let $R = \max_{i \in N} v_i(H)$.
Since the goods valuations are nonnegative and monotone, every bundle has value at most $R$ for every agent.

We set $\alpha_1 = R + 1$.
For every agent $i$ and every bundle $S \subseteq H$, we define
\[
u_i^C(S) = v_i(S) - \alpha_1 |S|.
\]
This is an additive utility function.
Equivalently, the chore cost is $c_i^C(S) = \alpha_1 |S| - v_i(S)$.
Since $\alpha_1 > R$, the utility $u_i^C$ is nonpositive on every bundle, and the cost $c_i^C$ is nonnegative and monotone.

For the rest of the proof, we define 
$$
\Gamma_t = \alpha_1 t.
$$
Thus, $\Gamma_t$ is strictly increasing in $t$.
We have $u_i^C(S) = v_i(S) - \Gamma_{|S|}$.
In particular, for every bundle $S$ with exactly $q$ items, $u_i^C(S) = v_i(S) - \Gamma_q$.

\begin{claim-inside-lemma}
Agent $i$'s maximin share in the constructed chores instance is $\mms_i^{\mathrm{chores}} = \mms_i^{\mathrm{goods}} - \Gamma_q$.
\end{claim-inside-lemma}

\begin{proof}
First, consider an MMS-witnessing partition for agent $i$ in the padded goods instance.
Since every original bundle has size at most $m = q$, we can add dummy items to the bundles so that every bundle has size exactly $q$.
Let this balanced partition be $(P_1, \ldots, P_n)$.
For every $j \in [n]$, we have $v_i(P_j) \ge \mms_i^{\mathrm{goods}}$.
Therefore,
\[
u_i^C(P_j)
=
v_i(P_j) - \Gamma_q
\ge
\mms_i^{\mathrm{goods}} - \Gamma_q.
\]
Hence, we have $
\mms_i^{\mathrm{chores}}
\ge
\mms_i^{\mathrm{goods}} - \Gamma_q$.

For the reverse inequality, consider any partition $\mathcal{P} = (P_1, \ldots, P_n)$ of $H$.
If some bundle $P_j$ has size at least $q + 1$, then
\[
\begin{aligned}
u_i^C(P_j)
=
v_i(P_j) - \Gamma_{|P_j|} 
\le
R - \Gamma_{q + 1} 
=
R - \left(\Gamma_q + \alpha_1\right)
<
-\Gamma_q
\le
\mms_i^{\mathrm{goods}} - \Gamma_q.
\end{aligned}
\]
The strict inequality follows from $\alpha_1 = R + 1$, and the last inequality follows from $\mms_i^{\mathrm{goods}} \ge 0$.
Therefore, any partition with a bundle of size at least $q + 1$ has $\min_{j \in [n]} u_i^C(P_j) < \mms_i^{\mathrm{goods}} - \Gamma_q$.

It remains to consider partitions in which no bundle has more than $q$ items.
Since there are exactly $nq$ items and $n$ bundles, every bundle must have exactly $q$ items.
For such a balanced partition,
\[
\begin{aligned}
\min_{j \in [n]} u_i^C(P_j)
&=
\min_{j \in [n]} \bigl(v_i(P_j) - \Gamma_q\bigr) \\
&=
\left(\min_{j \in [n]} v_i(P_j)\right) - \Gamma_q \\
&\le
\mms_i^{\mathrm{goods}} - \Gamma_q.
\end{aligned}
\]
Thus, every partition has minimum chore utility at most $\mms_i^{\mathrm{goods}} - \Gamma_q$.
Hence,
$
\mms_i^{\mathrm{chores}}
\le
\mms_i^{\mathrm{goods}} - \Gamma_q$.
Combining the lower and upper bounds, we obtain $\mms_i^{\mathrm{chores}} = \mms_i^{\mathrm{goods}} - \Gamma_q$.
\end{proof}

We now prove the equivalence of MMS allocations.

\begin{claim-inside-lemma}
$I$ admits an MMS allocation if and only if $I^C$ admits an MMS allocation.
\end{claim-inside-lemma}
\begin{proof}
Suppose first that the original goods instance admits an MMS allocation $A = (A_1, \ldots, A_n)$.
Add dummy items to the bundles so that every bundle has size exactly $q$.
Then, for every agent $i$, we have $v_i(A_i) \ge \mms_i^{\mathrm{goods}}$.
Since $|A_i| = q$,
\[
u_i^C(A_i)
=
v_i(A_i) - \Gamma_q
\ge
\mms_i^{\mathrm{goods}} - \Gamma_q
=
\mms_i^{\mathrm{chores}}.
\]
Thus, the padded allocation is an MMS allocation in the constructed chores instance.

Conversely, suppose that the constructed chores instance admits an MMS allocation $A = (A_1, \ldots, A_n)$.
Then, for every agent $i$,
\[
u_i^C(A_i)
\ge
\mms_i^{\mathrm{chores}}
=
\mms_i^{\mathrm{goods}} - \Gamma_q.
\]
No agent can receive more than $q$ items, since such a bundle has chore utility strictly smaller than $\mms_i^{\mathrm{goods}} - \Gamma_q$.
Since there are exactly $nq$ items and $n$ agents, every agent receives exactly $q$ items.
Therefore,
\[
\begin{aligned}
v_i(A_i)
&=
u_i^C(A_i) + \Gamma_q \\
&\ge
\mms_i^{\mathrm{chores}} + \Gamma_q \\
&=
\mms_i^{\mathrm{goods}}.
\end{aligned}
\]
Hence, after discarding dummy items, this gives an MMS allocation for the original goods instance.

Therefore, $I$ admits an MMS allocation if and only if $I^C$ admits an MMS allocation.
\end{proof}

Finally, we note that the transformation has polynomial encoding length. 
We add $nq - m = (n - 1)m$ padding items, and hence the number of items remains polynomial.
The values of $R$ and $\alpha_1$ have polynomial bit length.
Every resulting chore coefficient has polynomial bit length.
Moreover, agents with identical valuations in the goods instance have identical valuations in the chores instance.
Thus, the number of agent types is preserved.
This completes the proof of the lemma.
\end{proof}

We now discuss the consequences of the transfer lemma.
Since the construction in \Cref{lem:goods-to-chores-transfer} is polynomial time and preserves the existence of MMS allocations, every hardness result proved for additive goods transfers to additive chores.



\begin{corollary}
\label{cor:chores-hardness}
Deciding whether an additive chores instance admits an MMS allocation is $\Delta_2^P$-complete, even for five agents.
Moreover, the problem is weakly \coNP-hard even for three agents, and strongly \coNP-hard even for two agent types.
\end{corollary}

\begin{proof}
The hardness statements follow by applying \Cref{lem:goods-to-chores-transfer} to the goods instances constructed in \Cref{thm:delta2p-additive-completeness,thm:weak-coNP-hardness,thm:strong-coNP-hardness}.
The reduction preserves the existence of MMS allocations and the number of agents, and the constructed chore costs are additive and nonnegative.
Moreover, agents with identical valuations in the goods instance have identical valuations in the chores instance, so the restriction on the number of agent types is also preserved.

For the strong \coNP-hardness claim, we also need to check that the transformation only results in polynomially bounded numerical values. 
In the construction for goods, the number of goods is polynomial, the restricted \threepartition numbers are polynomially bounded, and the FST-core entries, $\tau$, and $L = n\tau + 1$ are polynomially bounded.
Hence every original bundle value, and in particular $R = \max_i v_i(H)$, is polynomially bounded.
In the transfer from goods to chores, we set $\alpha_1 = R + 1$, so $\alpha_1$ is polynomially bounded, and every resulting chore coefficient $\alpha_1 - v_i(g)$ is polynomially bounded.
The number of padding items is also polynomial.
Overall, we retain strong \coNP-hardness.

Membership in $\Delta_2^P$ follows from \Cref{prop:circuit-membership-delta2p}.
\end{proof}

We also obtain the corresponding inapproximability consequence for chores.
Recall that, following \citet{DBLP:conf/aaai/AzizRSW17}, the optimal MMS ratio of a chores instance is the minimum $\lambda \in [0,\infty)$ for which there exists an allocation $A = (A_1,\dots,A_n)$ satisfying
\[
    v_i(A_i) \ge \lambda ~\mms_i
    \qquad
    \text{for every agent } i.
\]

Let $\mu_i^I$ denote the negative of the maximin share of agent $i$ in instance $I$, that is, 
\[
\mu_i^I = -\mms_i^I = \min_{(P_1,\dots,P_n)\in \Pi_n(M)} \max_{j \in [n]} c_i(P_j).
\]

Observe that for additive chores, if $\mu_i^I = 0$ for some agent $i$, then every chore has cost $0$ for agent $i$.
Therefore, assigning all chores to agent $i$ and the empty bundle to every other agent gives an MMS allocation.
Thus, when considering instances that do not admit an MMS allocation, we may assume that $\mu_i^I > 0$ for every agent $i$.

Under this assumption, the optimal MMS ratio of a chores instance $I$ is
\[
    \lambda^I
    =
    \min_{A \in \Pi_n(M)}
    \max_{i \in N}
    \frac{c_i(A_i)}{\mu_i^I}.
\]
Thus, $I$ admits an MMS allocation if and only if $\lambda^I \le 1$.

\begin{corollary}
\label{cor:no-FPTAS-optimal-MMS-ratio-chores}
Unless $\P = \NP$, there is no FPTAS for computing the optimal MMS ratio of an additive chores instance.
\end{corollary}

\begin{proof}
We apply \Cref{lem:goods-to-chores-transfer} to the instances used in the proof of \Cref{cor:no-FPTAS-optimal-MMS-ratio}.
Let $I^C$ be the constructed chores instance.
For each agent $i$ and each balanced bundle $S$ of size $q$, we have $c_i^C(S) = \Gamma_q - v_i(S)$, and the corresponding value $\mu_i^C = -\mms_i^C$ is $\mu_i^C = \Gamma_q - \mms_i^{\mathrm{goods}}$.

If the original goods instance admits an exact MMS allocation, then by \Cref{lem:goods-to-chores-transfer}, the constructed chores instance also admits an exact MMS allocation.
Hence, $\lambda^C \le 1$.

Now consider the case where the goods instance constructed in the proof of \Cref{cor:no-FPTAS-optimal-MMS-ratio} admits no MMS allocation.
For every allocation of the goods, some agent $i$ receives value at most $\mms_i^{\mathrm{goods}} - 1$.
Consider any allocation $A = (A_1, \dots, A_n)$ in the constructed chores instance.

If $A$ is balanced, then after discarding the dummy items, it induces an allocation of the original goods.
Thus, for some agent $i$,
\[
\begin{aligned}
    c_i^C(A_i)
    ~=~
    \Gamma_q - v_i(A_i) 
    ~\ge~
    \Gamma_q - \bigl(\mms_i^{\mathrm{goods}} - 1\bigr) 
    ~=~
    \mu_i^C + 1.
\end{aligned}
\]

If $A$ is not balanced, then since there are exactly $nq$ items and $n$ bundles, some agent receives a bundle of size at least $q + 1$.
By the proof of \Cref{lem:goods-to-chores-transfer}, this agent obtains chores of value strictly smaller than their MMS value, or equivalently, of cost strictly larger than $\mu_i^C$.
Since all costs in the additive construction are integral, this cost is at least $\mu_i^C + 1$.

Thus, in the case where MMS allocations do not exist, $\lambda^C \ge 1 + {1}/{(\max_i \mu_i^C)}$.
Moreover, since $\mu_i^C = \Gamma_q - \mms_i^{\mathrm{goods}} \le \Gamma_q$ for every agent $i$, we get $\lambda^C \ge 1 + {1}/({\Gamma_q})$.
The quantity $\Gamma_q$ is polynomially bounded in the construction, and so the gap between the case where an MMS allocation exists and the one where one does not exist is inverse-polynomial.

Suppose there were an FPTAS for computing the optimal MMS ratio of additive chores.
Consider running it with $\varepsilon = {1}/{(2\Gamma_q)}$.
Since $1/\varepsilon$ is polynomially bounded, the running time would be polynomial.

If an MMS allocation exists, then $\lambda^C \le 1$, and hence a $(1 + \varepsilon)$-approximation returns a value that is at most $1 + \varepsilon = 1 + {1}/{(2\Gamma_q)}$.
If no MMS allocation exists, then every allocation has ratio at least $1 + {1}/{(\Gamma_q)} > 1 + {1}/{(2\Gamma_q)}$.
Therefore, we would be able to distinguish the case where an MMS allocation exists from the case where one does not exist, implying $\P = \NP$.
\end{proof}

%% file: 9_backmatter/new_conclusion.tex
\section{Future Directions}
\label{sec:future-directions}

An interesting direction for future work is to further investigate the computational complexity of MMS allocation existence under additional structural restrictions, such as fewer agents or more restricted classes of valuations.
An important direction is to study the computational complexity of finding approximate MMS allocations.
Given the sustained progress in improving approximation guarantees, a natural question is whether there are provable limits to further improvements.
We note that an algorithm for finding approximately optimal MMS allocations, that is, allocations whose achieved MMS ratio approximates the optimal MMS ratio, also yields guarantees for finding approximate MMS allocations.
Hence, it would be interesting to better understand the approximability of the optimal MMS ratio, particularly whether a polynomial-time approximation scheme exists.

%% file: 9_backmatter/appendix.tex
\section*{Appendix}

\section{Other Missing Proofs and Details}
\subsection{Proof of \autoref{lem:lexmaxfourpartitionpromise-delta2p-hardness}}

\lexmaxfourpartitionpromisehardness*

We reduce from the canonical $\Delta_2^P$-hard problem \lexmaxsat~\cite{DBLP:journals/jcss/Krentel88}.
An instance consists of a CNF formula $\varphi(x_1, \ldots, x_n)$, with the variables ordered as $x_1 \succ x_2 \succ \dots \succ x_n$.
The question is whether $x_n = 1$ in the lexicographically largest satisfying assignment of $\varphi$.
If $\varphi$ is unsatisfiable, the instance is treated as a no-instance.

We first reduce \lexmaxsat to the corresponding promise problem.
An instance of \lexmaxsatpromise consists of a CNF formula $\varphi(x_1, \ldots, x_n)$, under the promise that $\varphi$ is satisfiable.
The question is whether $x_n = 1$ in the lexicographically largest satisfying assignment of $\varphi$.

\begin{lemma}
\label{lem:lexmaxsatpromise-delta2p-hardness}
\lexmaxsatpromise is $\Delta_2^P$-hard.
\end{lemma}

\begin{proof}
Let $\varphi(x_1, \ldots, x_n)$ be an instance of \lexmaxsat, and let $C_1, \ldots, C_m$ be the clauses of $\varphi$.
We introduce a new variable $x_0$ and construct the CNF formula
\begin{equation*}
f(x_0, x_1, \ldots, x_n)
=
\left(
\bigwedge_{\ell = 1}^{m}
(\neg x_0 \vee C_\ell)
\right)
\wedge
\left(
\bigwedge_{i = 1}^{n}
(x_0 \vee \neg x_i)
\right).
\end{equation*}
We order the variables as $x_0 \succ x_1 \succ \dots \succ x_n$.
The formula $f$ is always satisfiable, since the all-zero assignment satisfies it.

If $\varphi$ is satisfiable, then the lexicographically largest satisfying assignment of $f$ sets $x_0 = 1$.
This is because, with $x_0 = 1$, the clauses $(\neg x_0 \vee C_\ell)$ force the original clauses $C_\ell$ to be satisfied, while the clauses $(x_0 \vee \neg x_i)$ are automatically satisfied.
Thus, among assignments with $x_0 = 1$, the variables $x_1, \ldots, x_n$ form the lexicographically largest satisfying assignment of $\varphi$.
Therefore, the value of $x_n$ is preserved in such an assignment.

If $\varphi$ is unsatisfiable, then no satisfying assignment of $f$ can have $x_0 = 1$.
Thus, $x_0 = 0$, and the clauses $(x_0 \vee \neg x_i)$ force $x_i = 0$ for every $i \in [n]$.
Consequently, the lexicographically largest satisfying assignment of $f$ is the all-zero assignment, and in particular has $x_n = 0$.

Thus, \lexmaxsat reduces to \lexmaxsatpromise in polynomial time, and the lemma follows.
\end{proof}

Next, we establish hardness when the input formula is a $3$-CNF.
An instance of \lexmaxthreesatpromise is defined as in \lexmaxsatpromise, except that the input formula is a $3$-CNF.

\begin{lemma}
\label{lem:lexmaxthreesatpromise-delta2p-hardness}
\lexmaxthreesatpromise is $\Delta_2^P$-hard.
\end{lemma}

\begin{proof}
Let $\varphi(x_1, \ldots, x_n)$ be an instance of \lexmaxsatpromise.
We first apply the standard reduction from \textsc{SAT} to \textsc{3-SAT}.
For each clause $C = (\ell_1 \vee \dots \vee \ell_k)$ with $k \ge 4$, we introduce new variables $y_{C,1}, \dots, y_{C,k-3}$ and replace $C$ by
\begin{equation*}
f_C
=
(\ell_1 \vee \ell_2 \vee y_{C,1})
\wedge
\left(
\bigwedge_{j=1}^{k-4}
(\neg y_{C,j} \vee \ell_{j+2} \vee y_{C,j+1})
\right)
\wedge
(\neg y_{C,k-3} \vee \ell_{k-1} \vee \ell_k),
\end{equation*}
where the middle conjunction is omitted when $k = 4$.
Clauses containing at most three literals are left unchanged.
For every fixed assignment to the original variables, the clause $C$ is satisfied if and only if this assignment can be extended to the newly introduced variables so as to satisfy $f_C$.
Applying this replacement to every clause gives a $3$-CNF formula $\varphi'$ such that an assignment to $x_1,\ldots,x_n$ satisfies $\varphi$ if and only if it can be extended to a satisfying assignment of $\varphi'$.

We next introduce a new variable $z$ and add the clauses
\begin{equation*}
(x_n \vee \neg z)
\wedge
(\neg x_n \vee z).
\end{equation*}
Thus, every satisfying assignment satisfies $z = x_n$.

Let $\psi$ denote the resulting formula.
We order its variables as
\begin{equation*}
x_1 \succ \dots \succ x_n
\succ
\{\text{all auxiliary variables introduced above}\}
\succ z,
\end{equation*}
where the auxiliary variables may be ordered arbitrarily among themselves.

Every satisfying assignment of $\varphi$ extends to a satisfying assignment of $\psi$, and the restriction of every satisfying assignment of $\psi$ to $x_1,\ldots,x_n$ satisfies $\varphi$.
Since all newly introduced variables have lower precedence than $x_1,\ldots,x_n$, the restriction of the lexicographically largest satisfying assignment of $\psi$ to the original variables is exactly the lexicographically largest satisfying assignment of $\varphi$.
Finally, every satisfying assignment of $\psi$ satisfies $z = x_n$.

Hence, $z = 1$ in the lexicographically largest satisfying assignment of $\psi$ if and only if $x_n = 1$ in the lexicographically largest satisfying assignment of $\varphi$.
Moreover, $\psi$ is satisfiable because $\varphi$ is satisfiable.
Thus, the promise is preserved, and the reduction is polynomial time.
\end{proof}

We are now ready to establish the hardness of \lexmaxfourpartitionpromise.

\begin{proof}[Proof of \Cref{lem:lexmaxfourpartitionpromise-delta2p-hardness}]
We reduce from \lexmaxthreesatpromise.
Let $\psi(v_1,\ldots,v_r,z)$ be an instance, where the variable precedence order is $v_1 \succ v_2 \succ \dots \succ v_r \succ z$.
Recall that $\psi$ is satisfiable by the promise.

We use the reduction from \textsc{3-SAT} to \textsc{3-Dimensional Matching}, followed by the reduction from \textsc{3-Dimensional Matching} to \fourpartition, both as given in~\cite{DBLP:books/fm/GareyJ79}.
In the \textsc{3-Dimensional Matching} problem, we are given a set of triples $\mathcal{T} \subseteq X \times Y \times Z$, where $X$, $Y$, and $Z$ are pairwise-disjoint sets having the same number of elements.
A perfect matching is a subset $\mathcal{T}' \subseteq \mathcal{T}$ in which every element of $X \cup Y \cup Z$ occurs in exactly one triple.
We recall the structural properties of the two reductions that will be important for us.

We apply the first reduction to $\psi$, and let $J = (X,Y,Z,\mathcal T)$ be the resulting instance of \textsc{3-Dimensional Matching}.
For every variable $v$ of $\psi$, the variable gadget corresponding to $v$ contains two distinguished triples, which we denote by $t_v$ and $t_{\bar v}$, and a distinguished element $m_v$ contained in exactly these two triples and in no other triple.
The construction has the following properties:
\begin{enumerate}
\item every perfect matching of $J$ induces a satisfying assignment of $\psi$;
\item every satisfying assignment of $\psi$ extends to a perfect matching of $J$; and
\item in every perfect matching $M$, we have $t_v \in M$ if and only if $v = 1$, and $t_{\bar v} \in M$ if and only if $v = 0$.
\end{enumerate}
In particular, the triple selected to cover $m_v$ records the truth value of $v$.

We next apply the classical reduction from \textsc{3-Dimensional Matching} to \fourpartition.
Let $q = |\mathcal T|$.
For every triple $t = (x,y,z) \in \mathcal T$, the construction creates a distinguished number $u_t$.
We call it the \emph{placeholder} number corresponding to $t$.
For every element $a \in X \cup Y \cup Z$, if $a$ occurs $N(a)$ times in $\mathcal T$, the construction creates $N(a)$ numbers $a[1], a[2], \ldots, a[N(a)]$.
The number $a[1]$ is called the \emph{actual} number, and the remaining numbers are called the \emph{dummy} numbers.

The construction has the property that every subset of total value exactly $T$ has exactly four numbers and is of the form $\{u_t,x[i],y[j],z[k]\}$, where $t = (x,y,z) \in \mathcal T$ and either $i = j = k = 1$, or $i,j,k > 1$.
We refer to the former as an \emph{actual} quadruple and to the latter as a \emph{dummy} quadruple.
Consequently, every feasible \fourpartition solution contains each placeholder number $u_t$ in a different four-element subset.
Moreover, the triples $t$ whose corresponding $u_t$ occur in actual quadruples form a perfect matching of $J$, and every perfect matching of $J$ extends to at least one feasible \fourpartition solution.

For every variable $v$ in the \textsc{3-SAT} instance, whose marker element in $J$ is $m_v$, we define the \emph{marker} number
$g_v := m_v[1]$.
Since $g_v$ is an actual number, in every feasible \fourpartition solution it belongs to an actual quadruple.
The only two triples in $\mathcal T$ containing $m_v$ are $t_v$ and $t_{\bar v}$.
Thus,
\begin{align*}
g_v \text{ belongs to the quadruple containing }u_{t_v}
&\quad\Longleftrightarrow\quad v = 1,\\
g_v \text{ belongs to the quadruple containing }u_{t_{\bar v}}
&\quad\Longleftrightarrow\quad v = 0.
\end{align*}

Next, we specify the ordering of the numbers in the resulting \fourpartition instance.
We first place all the placeholder numbers $u_t$, for $t \in \mathcal T$, before every other number.
We choose their order subject to the following conditions.
For every variable $v \in \{v_1,\ldots,v_r\}$, we place $u_{t_v}$ before $u_{t_{\bar v}}$.
We additionally place $u_{t_z}$ last among all the placeholder numbers.
These conditions are compatible, since all the distinguished variable triples involved are distinct.
All remaining placeholder numbers are ordered arbitrarily.

After the placeholder numbers, we place the marker numbers $g_{v_1},g_{v_2},\ldots,g_{v_r}$ in this order.
We then place all the remaining numbers except $g_z$ in an arbitrary order, and place $g_z$ as the very last number.

Suppose that the placeholder numbers are labelled as $u_1,\ldots,u_q$ according to this order.
Consider any lexicographically maximum solution $\solS = (S_1,\ldots,S_q)$.
Since the $u_i$ occur in pairwise distinct subsets, the first column in $\incmat(\solS)$, corresponding to $u_1$, contains exactly one $1$.
Lexicographic maximality therefore forces this $1$ to occur in the first row, and hence $u_1 \in S_1$.

Inductively, suppose that $u_j \in S_j$ for every $j < i$.
The number $u_i$ cannot belong to any of $S_1,\ldots,S_{i-1}$, because these subsets already contain the distinct placeholder numbers $u_1,\ldots,u_{i-1}$.
Thus, the earliest row in which the unique $1$ in the column of $u_i$ can occur is row $i$, and lexicographic maximality forces $u_i \in S_i$.
Therefore,
\begin{equation*}
u_i \in S_i
\qquad
\text{for every }i \in [q].
\end{equation*}
Consequently, the prefix of the incidence string corresponding to all placeholder numbers is the same for every lexicographically maximum ordering of every feasible partition.
Thus, these columns only canonically determine the ordering of the four-element subsets.

We now show that the marker columns recover the lexicographic ordering of the satisfying assignments of $\psi$.
Fix a variable $v_i$.
The column corresponding to $g_{v_i}$ contains its unique $1$ in the row corresponding to $u_{t_{v_i}}$ when $v_i = 1$, and in the row corresponding to $u_{t_{\bar v_i}}$ when $v_i = 0$.
By construction, the former row occurs earlier than the latter.
Hence, between two canonically ordered feasible solutions whose induced satisfying assignments first differ on $v_i$, the solution with $v_i = 1$ has the lexicographically larger incidence string.
It follows that the columns $g_{v_1},g_{v_2},\ldots,g_{v_r}$ compare feasible solutions exactly according to the lexicographic order $v_1 \succ v_2 \succ \dots \succ v_r$.

Every satisfying assignment of $\psi$ extends to a perfect matching and hence to a feasible \fourpartition solution, while every feasible \fourpartition solution induces a satisfying assignment of $\psi$.
Therefore, the lexicographically maximum \fourpartition solution induces the lexicographically maximum satisfying assignment of $\psi$.

Finally, consider the last bit of the incidence string.
By construction, $g_z$ is the last number.
Moreover, $u_{t_z}$ is the last placeholder number, and hence the block containing $u_{t_z}$ is the last row $S_q$.
Therefore,
\begin{equation*}
\incmat(\solS^*)_{q,4q} = 1
\quad\Longleftrightarrow\quad
g_z \in S_q
\quad\Longleftrightarrow\quad
z = 1.
\end{equation*}
Since the induced assignment is the lexicographically largest satisfying assignment of $\psi$, the resulting instance of \lexmaxfourpartitionpromise is a yes-instance if and only if the original instance of \lexmaxthreesatpromise is a yes-instance.

The source formula $\psi$ is promised to be satisfiable.
Hence, the constructed \textsc{3-Dimensional Matching} instance admits a perfect matching, and consequently the resulting \fourpartition instance admits a solution.
Thus, the promise of \lexmaxfourpartitionpromise is satisfied.

Finally, we verify the claimed bound on the numbers.
Recall that $q = |\mathcal T|$.
The construction creates $q$ placeholder numbers and $N(a)$ numbers for every element $a \in X \cup Y \cup Z$.
Since $\sum_{a \in X \cup Y \cup Z} N(a) = 3q$, there are exactly $q + 3q = 4q$ numbers.
Let $s = |X| = |Y| = |Z|$.
In the classical reduction from \textsc{3-Dimensional Matching} to \fourpartition~\cite{DBLP:books/fm/GareyJ79}, the target is $T = 40r^4 + 15$, where $r = 32s$.
Since the constructed \textsc{3-Dimensional Matching} instance admits a perfect matching, every element of $X \cup Y \cup Z$ occurs in some triple.
Thus, $s \leq |\mathcal T| = q$, and consequently
\[
T
=
40(32s)^4 + 15
\leq
40 \cdot 32^4 q^4 + 15
=
\bigoh(q^4).
\]
We note that every number $a$ in the resulting \fourpartition instance satisfies $T/5 < a < T/3$, and is therefore also $\bigoh(q^4)$.

Moreover, the preceding reduction from \textsc{3-SAT} to \textsc{3-Dimensional Matching} produces a number of triples polynomial in the size of the input formula.
Hence $q = \poly(|\psi|)$, and therefore all numbers in the resulting \fourpartition instance are polynomially bounded in the size of the original \textsc{3-SAT} instance.
The entire reduction is polynomial time, and the lemma follows.
\end{proof}

\subsection{Proofs in Section \ref{sec:prelim}}
\label{sec:missingproofs}
\input{./9_backmatter/missing_proofs}






%% file: 9_backmatter/missing_proofs.tex
\circuitmembership*
\begin{proof}
We describe a polynomial-time algorithm with access to an $\NP$ oracle, following the same oracle-search idea used for additive valuations by \citet{DBLP:journals/jair/TruszczynskiL20}.

Fix an agent $i$, and let $b_i$ be the number of output bits of the circuit representing $v_i$.
Then every bundle value has absolute value less than $2^{b_i}$.
Then, for every bundle $S \subseteq M$, its value is between $-2^{b_i}$ and $2^{b_i}$, and therefore $-2^{b_i} \le \mms_i \le 2^{b_i}$.

For an integer $q$, deciding whether $\mms_i \ge q$ is in $\NP$.
This is because a certificate for it is an $n$-partition $(P_1,\dots,P_n)$ of $M$ such that $v_i(P_j) \ge q$ for every $j \in [n]$.
Given the partition, each value $v_i(P_j)$ can be computed in polynomial time using the circuit.

Hence, by binary search over the integer interval $[-2^{b_i},2^{b_i}]$, we can compute $\mms_i$ using
$\bigoh(b_i)$ adaptive $\NP$ queries.
Since $b_i$ is bounded by the size of the circuit, the total number of oracle queries over all agents is polynomial in the input size.

After computing all values $\mms_1,\dots,\mms_n$, we make one final $\NP$ oracle query asking whether there exists an allocation $A=(A_1,\dots,A_n)$ such that $v_i(A_i)\ge\mms_i$ for every $i\in[n]$. 
This query is in $\NP$: a certificate is the allocation itself, and all inequalities can be verified in polynomial time by evaluating the corresponding circuits.

Thus, the problem is decidable in polynomial time with polynomially many adaptive $\NP$ queries.
Therefore, it belongs to $\P^\NP=\Delta_2^P$.
\end{proof}

\nonpositivepairssubmodular*
\begin{proof}
For the forward direction, it is enough to show that for all $S \subseteq T$ and all $x \notin T$, 
\[
v_i(T \cup \{x\}) - v_i(T) \le v_i(S \cup \{x\}) - v_i(S).
\]

For any set $A$ with $x \notin A$, we have $v_i(A \cup \{x\}) - v_i(A) = \sigma_i(\{x\}) + \sum_{g \in A} \sigma_i(\{g, x\})$.

Therefore, \[v_i(S \cup \{x\}) - v_i(S) = \sigma_i(\{x\}) + \sum_{g \in S} \sigma_i(\{g, x\}),\text{ and }\]
\[v_i(T \cup \{x\}) - v_i(T) = \sigma_i(\{x\}) + \sum_{g \in T} \sigma_i(\{g, x\}).\]
Since $S \subseteq T$, we have $T = S \cup (T \setminus S)$, and consequently,
\[
\begin{aligned}
    v_i(T \cup \{x\}) - v_i(T)
    &=
    \sigma_i(\{x\})
    +
    \sum_{g \in T} \sigma_i(\{g,x\}) \\
    &=
    \sigma_i(\{x\})
    +
    \sum_{g \in S} \sigma_i(\{g,x\})
    +
    \sum_{g \in T \setminus S} \sigma_i(\{g,x\}) \\
    &\le
    \sigma_i(\{x\})
    +
    \sum_{g \in S} \sigma_i(\{g,x\}) \\
    &=
    v_i(S \cup \{x\}) - v_i(S).
\end{aligned}
\]
The inequality holds because all pairwise coefficients are nonpositive.
Therefore, $v_i$ is submodular.

For the reverse direction, we apply the definition of submodularity with $S = \emptyset$, $T = \{g\}$, and $x = h$.
Then,
\[
    v_i(\{g, h\}) - v_i(\{g\})
    \le
    v_i(\{h\}) - v_i(\emptyset).
\]
Expanding the terms, we get $\sigma_i(\{g,h\}) \le 0$.
\end{proof}

\verificationmms*
\begin{proof}
The problem is in \coNP: a certificate that an allocation $A$ is not an MMS allocation is an agent $i$ and an $n$-partition $(P_1,\dots,P_n)$ of the items such that $v_i(P_j) > v_i(A_i)$ for every $j \in [n]$.

For hardness, we reduce from the complement of \partition.
Let $W = \{w_1,\dots,w_m\}$ be an instance of \partition with $\sum_{j = 1}^m w_j = 2T$.

\paragraph{Hardness for goods.}
We create three agents and $m + 2$ goods: number goods $g_1,\dots,g_m$ (one for each number in the input), and two additional goods $a,b$.
Agents $2$ and $3$ have value zero for every good.
Agent $1$ has values $v_1(a) = 2T - 1$, $v_1(b) = 1$, and $v_1(g_j) = 2w_j$ for every $j \in [m]$.

Consider the allocation $A$ in which agent $1$ receives only $a$, and the remaining goods are allocated arbitrarily to agents $2$ and $3$.
Then, $v_1(A_1) = 2T - 1$, and agents $2$ and $3$ are trivially MMS-satisfied.

If the \partition instance is a yes-instance, then the number goods can be split into two bundles of value $2T$ for agent $1$, and the third bundle $\{a,b\}$ also has value $2T$.
Hence, $\mms_1 \ge 2T > v_1(A_1)$, and so $A$ is not an MMS allocation.

Conversely, suppose that $A$ is not an MMS allocation.
Since agents $2$ and $3$ have value zero for every good, agent $1$ must be the violating agent.
Thus, $\mms_1 > 2T - 1$.
Since all values are integral, $\mms_1 \ge 2T$.
Since the total value for agent $1$ is $6T$, any $3$-partition witnessing this must have all three bundles of value exactly $2T$.
The bundle containing $a$ must also contain $b$, because $v_1(a) = 2T - 1$ and all number goods have positive even value.
Since $v_1(a) + v_1(b) = 2T$, this bundle is exactly $\{a,b\}$.
Thus, the remaining two bundles contain only number goods and each has value $2T$, giving a partition of $W$ into two parts of sum $T$.
Therefore, in the goods instance, $A$ is an MMS allocation if and only if the \partition instance is a no-instance.

\paragraph{Hardness for chores.}
We create three agents and $m + 2$ chores: number chores $g_1,\dots,g_m$ (one for each number in the input), and two additional chores $a,b$.
Agents $2$ and $3$ have value zero for every chore.
Agent $1$ has values $v_1(a) = -2T - 1$, $v_1(b) = -1$, and $v_1(g_j) = -2w_j$ for every $j \in [m]$.

Consider the allocation $A$ in which agent $1$ receives both $a$ and $b$, and the remaining chores are allocated arbitrarily to agents $2$ and $3$.
Then, $v_1(A_1) = -2T - 2$, and agents $2$ and $3$ are trivially MMS-satisfied.

If the \partition instance is a yes-instance, then the chores can be split into three bundles of value at least $-2T - 1$ for agent $1$ as follows:
The first bundle contains chores corresponding to a subset of numbers with sum $T$, together with the chore $b$, and thus has value $-2T - 1$.
The second bundle contains the chores corresponding to the remaining subset of numbers with sum $T$, and thus has value $-2T$.
The third bundle is $\{a\}$, which has value $-2T - 1$.
Hence, $\mms_1 \ge -2T - 1 > v_1(A_1)$, and so $A$ is not an MMS allocation.

Conversely, suppose that $A$ is not an MMS allocation.
Since agents $2$ and $3$ have value zero for every chore, agent $1$ must be the violating agent.
Then, $\mms_1 > -2T - 2$.
Since all values are integral, $\mms_1 \ge -2T - 1$.
Thus, there is a $3$-partition in which every bundle has value at least $-2T - 1$ for agent $1$.

Since the total value for agent $1$ is $-6T - 2$, in any such witnessing partition, exactly two bundles have value $-2T - 1$, and one bundle has value $-2T$.
The bundle containing $a$ must be $\{a\}$, because $v_1(a) = -2T - 1$ and all other chores have negative value for agent $1$.

The remaining two bundles contain all number chores and the chore $b$.
Let the bundle containing $b$ contain number chores corresponding to numbers with sum $x$, and let the other bundle contain number chores corresponding to numbers with sum $y$.
Then, $x + y = 2T$.
Since the bundle containing $b$ has value at least $-2T - 1$, we have $-1 - 2x \ge -2T - 1$,
and hence, $x \le T$.
Similarly, since the other bundle has value at least $-2T - 1$, we have $-2y \ge -2T - 1$,
and hence, $y \le T$.
Since $x + y = 2T$, we have $x = y = T$, and so the original \partition instance is a yes-instance.
Therefore, in the chores instance, $A$ is an MMS allocation if and only if the \partition instance is a no-instance.

Thus, in both the goods and the chores setting, the problem is \coNP-hard, and hence \coNP-complete.
\end{proof}

%% file: bibliography.bib
@inproceedings{DBLP:conf/wine/FeigeST21,
  author       = {Uriel Feige and
                  Ariel Sapir and
                  Laliv Tauber},
  editor       = {Michal Feldman and
                  Hu Fu and
                  Inbal Talgam{-}Cohen},
  title        = {A Tight Negative Example for {MMS} Fair Allocations},
  booktitle    = {Web and Internet Economics - 17th International Conference, {WINE}
                  2021},
  series       = {Lecture Notes in Computer Science},
  pages        = {355--372},
  publisher    = {Springer},
  year         = {2021},
  url          = {https://doi.org/10.1007/978-3-030-94676-0_20},
  doi          = {10.1007/978-3-030-94676-0_20},
  bibsource    = {dblp computer science bibliography, https://dblp.org}
}

@book{DBLP:books/fm/GareyJ79,
  author       = {M. R. Garey and
                  David S. Johnson},
  title        = {Computers and Intractability: {A} Guide to the Theory of NP-Completeness},
  publisher    = {W. H. Freeman},
  year         = {1979},
  isbn         = {0-7167-1044-7},
  bibsource    = {dblp computer science bibliography, https://dblp.org}
}

@article{DBLP:journals/jacm/KurokawaPW18,
  author       = {David Kurokawa and
                  Ariel D. Procaccia and
                  Junxing Wang},
  title        = {Fair Enough: Guaranteeing Approximate Maximin Shares},
  journal      = {J. {ACM}},
  volume       = {65},
  number       = {2},
  pages        = {8:1--8:27},
  year         = {2018},
  url          = {https://doi.org/10.1145/3140756},
  doi          = {10.1145/3140756},
  bibsource    = {dblp computer science bibliography, https://dblp.org}
}

@article{DBLP:journals/jcss/PapadimitriouY84,
  author       = {Christos H. Papadimitriou and
                  Mihalis Yannakakis},
  title        = {The Complexity of Facets (and Some Facets of Complexity)},
  journal      = {J. Comput. Syst. Sci.},
  volume       = {28},
  number       = {2},
  pages        = {244--259},
  year         = {1984},
  url          = {https://doi.org/10.1016/0022-0000(84)90068-0},
  doi          = {10.1016/0022-0000(84)90068-0},
  bibsource    = {dblp computer science bibliography, https://dblp.org}
}

@article{DBLP:journals/jcss/Krentel88,
  author       = {Mark W. Krentel},
  title        = {The Complexity of Optimization Problems},
  journal      = {J. Comput. Syst. Sci.},
  volume       = {36},
  number       = {3},
  pages        = {490--509},
  year         = {1988},
  url          = {https://doi.org/10.1016/0022-0000(88)90039-6},
  doi          = {10.1016/0022-0000(88)90039-6},
  bibsource    = {dblp computer science bibliography, https://dblp.org}
}

@article{DBLP:conf/bqgt/Budish10,
title={{The Combinatorial Assignment Problem: Approximate Competitive Equilibrium from Equal Incomes}},
  author={Budish, Eric},
  journal={Journal of Political Economy},
  volume={119},
  number={6},
  pages={1061--1103},
  year={2011}
}

@article{DBLP:journals/aamas/BouveretL16,
  author       = {Sylvain Bouveret and
                  Michel Lemaître},
  title        = {Characterizing conflicts in fair division of indivisible goods using
                  a scale of criteria},
  journal      = {Auton. Agents Multi Agent Syst.},
  volume       = {30},
  number       = {2},
  pages        = {259--290},
  year         = {2016},
  url          = {https://doi.org/10.1007/s10458-015-9287-3},
  doi          = {10.1007/S10458-015-9287-3},
  bibsource    = {dblp computer science bibliography, https://dblp.org}
}

@article{DBLP:journals/jair/TruszczynskiL20,
  author       = {Zbigniew Lonc and
                  Miroslaw Truszczynski},
  title        = {Maximin Share Allocations on Cycles},
  journal      = {J. Artif. Intell. Res.},
  volume       = {69},
  pages        = {613--655},
  year         = {2020},
  url          = {https://doi.org/10.1613/jair.1.11702},
  doi          = {10.1613/JAIR.1.11702},
  bibsource    = {dblp computer science bibliography, https://dblp.org}
}

@inproceedings{DBLP:conf/soda/HeidariKSS26,
  author       = {Ehsan Heidari and
                  Alireza Kaviani and
                  Masoud Seddighin and
                  AmirMohammad Shahrezaei},
  editor       = {Kasper Green Larsen and
                  Barna Saha},
  title        = {Improved Maximin Share Guarantee for Additive Valuations},
  booktitle    = {Proceedings of the 2026 Annual {ACM-SIAM} Symposium on Discrete Algorithms,
                  {SODA} 2026},
  pages        = {2239--2290},
  publisher    = {{SIAM}},
  year         = {2026},
  url          = {https://doi.org/10.1137/1.9781611978971.81},
  doi          = {10.1137/1.9781611978971.81},
  bibsource    = {dblp computer science bibliography, https://dblp.org}
}

@article{DBLP:journals/corr/abs-2511-13056,
  author       = {Xin Huang and
                  Shengwei Zhou},
  title        = {An {FPTAS} for 7/9-Approximation to Maximin Share Allocations},
  journal      = {CoRR},
  volume       = {abs/2511.13056},
  year         = {2025},
  url          = {https://doi.org/10.48550/arXiv.2511.13056},
  doi          = {10.48550/ARXIV.2511.13056},
  eprinttype   = {arXiv},
  eprint       = {2511.13056},
  bibsource    = {dblp computer science bibliography, https://dblp.org}
}

@inproceedings{DBLP:conf/soda/AkramiG24,
  author       = {Hannaneh Akrami and
                  Jugal Garg},
  editor       = {David P. Woodruff},
  title        = {Breaking the 3/4 Barrier for Approximate Maximin Share},
  booktitle    = {Proceedings of the 2024 {ACM-SIAM} Symposium on Discrete Algorithms,
                  {SODA} 2024},
  pages        = {74--91},
  publisher    = {{SIAM}},
  year         = {2024},
  url          = {https://doi.org/10.1137/1.9781611977912.4},
  doi          = {10.1137/1.9781611977912.4},
  bibsource    = {dblp computer science bibliography, https://dblp.org}
}

@inproceedings{DBLP:conf/sigecom/ProcacciaW14,
  author       = {Ariel D. Procaccia and
                  Junxing Wang},
  editor       = {Moshe Babaioff and
                  Vincent Conitzer and
                  David A. Easley},
  title        = {Fair enough: guaranteeing approximate maximin shares},
  booktitle    = {{ACM} Conference on Economics and Computation, {EC} '14},
  pages        = {675--692},
  publisher    = {{ACM}},
  year         = {2014},
  url          = {https://doi.org/10.1145/2600057.2602835},
  doi          = {10.1145/2600057.2602835},
  bibsource    = {dblp computer science bibliography, https://dblp.org}
}

@inproceedings{DBLP:conf/atal/BouveretL14,
  author       = {Sylvain Bouveret and
                  Michel Lemaître},
  editor       = {Ana L. C. Bazzan and
                  Michael N. Huhns and
                  Alessio Lomuscio and
                  Paul Scerri},
  title        = {Characterizing conflicts in fair division of indivisible goods using
                  a scale of criteria},
  booktitle    = {International conference on Autonomous Agents and Multi-Agent Systems,
                  {AAMAS} '14},
  pages        = {1321--1328},
  publisher    = {{IFAAMAS/ACM}},
  year         = {2014},
  url          = {http://dl.acm.org/citation.cfm?id=2617458},
  bibsource    = {dblp computer science bibliography, https://dblp.org}
}

@inproceedings{DBLP:conf/aaai/AzizRSW17,
  author       = {Haris Aziz and
                  Gerhard Rauchecker and
                  Guido Schryen and
                  Toby Walsh},
  editor       = {Satinder Singh and
                  Shaul Markovitch},
  title        = {Algorithms for Max-Min Share Fair Allocation of Indivisible Chores},
  booktitle    = {Proceedings of the Thirty-First {AAAI} Conference on Artificial Intelligence,
                  February 4-9, 2017},
  pages        = {335--341},
  publisher    = {{AAAI} Press},
  year         = {2017},
  url          = {https://doi.org/10.1609/aaai.v31i1.10582},
  doi          = {10.1609/AAAI.V31I1.10582},
  bibsource    = {dblp computer science bibliography, https://dblp.org}
}

@article{DBLP:journals/aamas/HeinenNNR18,
  author       = {Tobias Heinen and
                  Nhan{-}Tam Nguyen and
                  Trung Thanh Nguyen and
                  J{\"{o}}rg Rothe},
  title        = {Approximation and complexity of the optimization and existence problems
                  for maximin share, proportional share, and minimax share allocation
                  of indivisible goods},
  journal      = {Auton. Agents Multi Agent Syst.},
  volume       = {32},
  number       = {6},
  pages        = {741--778},
  year         = {2018},
  url          = {https://doi.org/10.1007/s10458-018-9393-0},
  doi          = {10.1007/S10458-018-9393-0},
  bibsource    = {dblp computer science bibliography, https://dblp.org}
}

@inproceedings{DBLP:conf/sigecom/SeddighinS25,
  author       = {Masoud Seddighin and
                  Saeed Seddighin},
  editor       = {Itai Ashlagi and
                  Aaron Roth},
  title        = {Beating the Logarithmic Barrier for the Subadditive Maximin Share
                  Problem},
  booktitle    = {Proceedings of the 26th {ACM} Conference on Economics and Computation,
                  {EC} 2025},
  pages        = {764--782},
  publisher    = {{ACM}},
  year         = {2025},
  url          = {https://doi.org/10.1145/3736252.3742621},
  doi          = {10.1145/3736252.3742621},
  bibsource    = {dblp computer science bibliography, https://dblp.org}
}

@article{DBLP:journals/corr/abs-2605-08859,
  author       = {Uriel Feige and
                  Vadim Grinberg},
  title        = {On MMS, {APS} and {XOS}},
  journal      = {CoRR},
  volume       = {abs/2605.08859},
  year         = {2026},
  url          = {https://doi.org/10.48550/arXiv.2605.08859},
  doi          = {10.48550/ARXIV.2605.08859},
  eprinttype   = {arXiv},
  eprint       = {2605.08859},
  bibsource    = {dblp computer science bibliography, https://dblp.org}
}

@inproceedings{DBLP:conf/sigecom/HuangS23,
  author       = {Xin Huang and
                  Erel Segal{-}Halevi},
  editor       = {Kevin Leyton{-}Brown and
                  Jason D. Hartline and
                  Larry Samuelson},
  title        = {A Reduction from Chores Allocation to Job Scheduling},
  booktitle    = {Proceedings of the 24th {ACM} Conference on Economics and Computation,
                  {EC} 2023},
  pages        = {908},
  publisher    = {{ACM}},
  year         = {2023},
  url          = {https://doi.org/10.1145/3580507.3597676},
  doi          = {10.1145/3580507.3597676},
  bibsource    = {dblp computer science bibliography, https://dblp.org}
}

@inproceedings{DBLP:conf/sigecom/HuangL21,
  author       = {Xin Huang and
                  Pinyan Lu},
  editor       = {P{\'{e}}ter Bir{\'{o}} and
                  Shuchi Chawla and
                  Federico Echenique},
  title        = {An Algorithmic Framework for Approximating Maximin Share Allocation
                  of Chores},
  booktitle    = {{EC} '21: The 22nd {ACM} Conference on Economics and Computation, 2021},
  pages        = {630--631},
  publisher    = {{ACM}},
  year         = {2021},
  url          = {https://doi.org/10.1145/3465456.3467555},
  doi          = {10.1145/3465456.3467555},
  bibsource    = {dblp computer science bibliography, https://dblp.org}
}

@article{DBLP:journals/ai/AmanatidisABFLMVW23,
  author       = {Georgios Amanatidis and
                  Haris Aziz and
                  Georgios Birmpas and
                  Aris Filos{-}Ratsikas and
                  Bo Li and
                  Herv{\'{e}} Moulin and
                  Alexandros A. Voudouris and
                  Xiaowei Wu},
  title        = {Fair division of indivisible goods: Recent progress and open questions},
  journal      = {Artif. Intell.},
  volume       = {322},
  pages        = {103965},
  year         = {2023},
  url          = {https://doi.org/10.1016/j.artint.2023.103965},
  doi          = {10.1016/J.ARTINT.2023.103965},
  bibsource    = {dblp computer science bibliography, https://dblp.org}
}

@article{DBLP:journals/jair/LiuLSW24,
  author       = {Shengxin Liu and
                  Xinhang Lu and
                  Mashbat Suzuki and
                  Toby Walsh},
  title        = {Mixed Fair Division: {A} Survey},
  journal      = {J. Artif. Intell. Res.},
  volume       = {80},
  pages        = {1373--1406},
  year         = {2024},
  url          = {https://doi.org/10.1613/jair.1.15800},
  doi          = {10.1613/JAIR.1.15800},
  bibsource    = {dblp computer science bibliography, https://dblp.org}
}

@article{DBLP:journals/teco/BarmanK20,
  author       = {Siddharth Barman and
                  Sanath Kumar Krishnamurthy},
  title        = {Approximation Algorithms for Maximin Fair Division},
  journal      = {{ACM} Trans. Economics and Comput.},
  volume       = {8},
  number       = {1},
  pages        = {5:1--5:28},
  year         = {2020},
  url          = {https://doi.org/10.1145/3381525},
  doi          = {10.1145/3381525},
  bibsource    = {dblp computer science bibliography, https://dblp.org}
}

@inproceedings{DBLP:conf/sigecom/GhodsiHSSY18,
  author       = {Mohammad Ghodsi and
                  Mohammad Taghi Hajiaghayi and
                  Masoud Seddighin and
                  Saeed Seddighin and
                  Hadi Yami},
  editor       = {{\'{E}}va Tardos and
                  Edith Elkind and
                  Rakesh Vohra},
  title        = {Fair Allocation of Indivisible Goods: Improvements and Generalizations},
  booktitle    = {Proceedings of the 2018 {ACM} Conference on Economics and Computation, 2018},
  pages        = {539--556},
  publisher    = {{ACM}},
  year         = {2018},
  url          = {https://doi.org/10.1145/3219166.3219238},
  doi          = {10.1145/3219166.3219238},
  bibsource    = {dblp computer science bibliography, https://dblp.org}
}

@article{DBLP:journals/anor/ChevaleyreEEM08,
  author       = {Yann Chevaleyre and
                  Ulle Endriss and
                  Sylvia Estivie and
                  Nicolas Maudet},
  title        = {Multiagent resource allocation in \emph{k} -additive domains: preference
                  representation and complexity},
  journal      = {Ann. Oper. Res.},
  volume       = {163},
  number       = {1},
  pages        = {49--62},
  year         = {2008},
  url          = {https://doi.org/10.1007/s10479-008-0335-0},
  doi          = {10.1007/S10479-008-0335-0},
  bibsource    = {dblp computer science bibliography, https://dblp.org}
}

@inproceedings{DBLP:conf/sigecom/KulkarniMT21,
  author       = {Rucha Kulkarni and
                  Ruta Mehta and
                  Setareh Taki},
  editor       = {P{\'{e}}ter Bir{\'{o}} and
                  Shuchi Chawla and
                  Federico Echenique},
  title        = {Indivisible Mixed Manna: On the Computability of {MMS+PO} Allocations},
  booktitle    = {{EC} '21: The 22nd {ACM} Conference on Economics and Computation, 2021},
  pages        = {683--684},
  publisher    = {{ACM}},
  year         = {2021},
  url          = {https://doi.org/10.1145/3465456.3467553},
  doi          = {10.1145/3465456.3467553},
  bibsource    = {dblp computer science bibliography, https://dblp.org}
}

@inproceedings{DBLP:conf/aaai/ChristodoulouM26,
  author       = {George Christodoulou and
                  Symeon Mastrakoulis},
  editor       = {Sven Koenig and
                  Chad Jenkins and
                  Matthew E. Taylor},
  title        = {Exact and Approximate Maximin Share Allocations in Multi-Graphs},
  booktitle    = {Fortieth {AAAI} Conference on Artificial Intelligence, Thirty-Eighth
                  Conference on Innovative Applications of Artificial Intelligence,
                  Sixteenth Symposium on Educational Advances in Artificial Intelligence,
                  {AAAI} 2026},
  pages        = {16761--16769},
  publisher    = {{AAAI} Press},
  year         = {2026},
  url          = {https://doi.org/10.1609/aaai.v40i20.38719},
  doi          = {10.1609/AAAI.V40I20.38719},
  bibsource    = {dblp computer science bibliography, https://dblp.org}
}

@inproceedings{DBLP:conf/ijcai/BouveretCEIP17,
  author       = {Sylvain Bouveret and
                  Katar{\'{\i}}na Cechl{\'{a}}rov{\'{a}} and
                  Edith Elkind and
                  Ayumi Igarashi and
                  Dominik Peters},
  editor       = {Carles Sierra},
  title        = {Fair Division of a Graph},
  booktitle    = {Proceedings of the Twenty-Sixth International Joint Conference on
                  Artificial Intelligence, {IJCAI} 2017},
  pages        = {135--141},
  publisher    = {ijcai.org},
  year         = {2017},
  url          = {https://doi.org/10.24963/ijcai.2017/20},
  doi          = {10.24963/IJCAI.2017/20},
  bibsource    = {dblp computer science bibliography, https://dblp.org}
}

@inproceedings{DBLP:conf/aaai/IgarashiP19,
  author       = {Ayumi Igarashi and
                  Dominik Peters},
  title        = {Pareto-Optimal Allocation of Indivisible Goods with Connectivity Constraints},
  booktitle    = {The Thirty-Third {AAAI} Conference on Artificial Intelligence, {AAAI}
                  2019, The Thirty-First Innovative Applications of Artificial Intelligence
                  Conference, {IAAI} 2019, The Ninth {AAAI} Symposium on Educational
                  Advances in Artificial Intelligence, {EAAI} 2019},
  pages        = {2045--2052},
  publisher    = {{AAAI} Press},
  year         = {2019},
  url          = {https://doi.org/10.1609/aaai.v33i01.33012045},
  doi          = {10.1609/AAAI.V33I01.33012045},
  bibsource    = {dblp computer science bibliography, https://dblp.org}
}

@article{DBLP:journals/corr/abs-2303-12444,
  author       = {Gilad Ben Uziahu and
                  Uriel Feige},
  title        = {On Fair Allocation of Indivisible Goods to Submodular Agents},
  journal      = {CoRR},
  volume       = {abs/2303.12444},
  year         = {2023},
  url          = {https://doi.org/10.48550/arXiv.2303.12444},
  doi          = {10.48550/ARXIV.2303.12444},
  eprinttype   = {arXiv},
  eprint       = {2303.12444},
  bibsource    = {dblp computer science bibliography, https://dblp.org}
}

@article{DBLP:journals/tcs/GoldbergHH25,
  author       = {Paul W. Goldberg and
                  Kasper H{\o}gh and
                  Alexandros Hollender},
  title        = {The frontier of intractability for {EFX} with two agents},
  journal      = {Theor. Comput. Sci.},
  volume       = {1052},
  pages        = {115367},
  year         = {2025},
  url          = {https://doi.org/10.1016/j.tcs.2025.115367},
  doi          = {10.1016/J.TCS.2025.115367},
  bibsource    = {dblp computer science bibliography, https://dblp.org}
}

@inproceedings{DBLP:conf/aaai/GrecoS20,
  author       = {Gianluigi Greco and
                  Francesco Scarcello},
  title        = {The Complexity of Computing Maximin Share Allocations on Graphs},
  booktitle    = {The Thirty-Fourth {AAAI} Conference on Artificial Intelligence, {AAAI}
                  2020, The Thirty-Second Innovative Applications of Artificial Intelligence
                  Conference, {IAAI} 2020, The Tenth {AAAI} Symposium on Educational
                  Advances in Artificial Intelligence, {EAAI} 2020},
  pages        = {2006--2013},
  publisher    = {{AAAI} Press},
  year         = {2020},
  url          = {https://doi.org/10.1609/aaai.v34i02.5572},
  doi          = {10.1609/AAAI.V34I02.5572},
  bibsource    = {dblp computer science bibliography, https://dblp.org}
}

@article{ezra2026improved,
  author       = {Tomer Ezra and
                  Tamar Garbuz},
  title        = {Improved Impossibility Bounds for Maximin Share Allocations},
  journal      = {CoRR},
  volume       = {abs/2609.15085},
  year         = {2026},
  url          = {https://doi.org/10.48550/arXiv.2609.15085},
  doi          = {10.48550/ARXIV.2609.15085},
  eprinttype   = {arXiv},
  eprint       = {2609.15085}
}

@article{DBLP:journals/corr/abs-2609-31849,
  author       = {Xiaolin Bu and
                  Biaoshuai Tao},
  title        = {Fair and Efficient Allocations: Decision Problems in the Gap of Polynomial
                  Hierarchy},
  journal      = {CoRR},
  volume       = {abs/2609.31849},
  year         = {2026},
  url          = {https://doi.org/10.48550/arXiv.2609.31849},
  doi          = {10.48550/ARXIV.2609.31849},
  eprinttype   = {arXiv},
  eprint       = {2609.31849},
  bibsource    = {dblp computer science bibliography, https://dblp.org}
}

@inproceedings{DBLP:conf/aldt/KeijzerBKZ09,
  author       = {Bart de Keijzer and
                  Sylvain Bouveret and
                  Tomas Klos and
                  Yingqian Zhang},
  editor       = {Francesca Rossi and
                  Alexis Tsouki{\`{a}}s},
  title        = {On the Complexity of Efficiency and Envy-Freeness in Fair Division
                  of Indivisible Goods with Additive Preferences},
  booktitle    = {Algorithmic Decision Theory, First International Conference, {ADT}
                  2009},
  series       = {Lecture Notes in Computer Science},
  volume       = {5783},
  pages        = {98--110},
  publisher    = {Springer},
  year         = {2009},
  url          = {https://doi.org/10.1007/978-3-642-04428-1_9},
  doi          = {10.1007/978-3-642-04428-1_9},
  bibsource    = {dblp computer science bibliography, https://dblp.org}
}

@article{DBLP:journals/jair/BouveretL08,
  author       = {Sylvain Bouveret and
                  J{\'{e}}r{\^{o}}me Lang},
  title        = {Efficiency and Envy-freeness in Fair Division of Indivisible Goods:
                  Logical Representation and Complexity},
  journal      = {J. Artif. Intell. Res.},
  volume       = {32},
  pages        = {525--564},
  year         = {2008},
  url          = {https://doi.org/10.1613/jair.2467},
  doi          = {10.1613/JAIR.2467},
  bibsource    = {dblp computer science bibliography, https://dblp.org}
}

@misc{deHaan2026,
    author       = {Ronald de Haan},
    title        = {The Complexity of Exact MMS Allocations},
    year         = {2026},
    howpublished = {ProofForum, PF:2026.000014},
    note         = {Preliminary manuscript, version 1},
    url          = {https://www.proofforum.org/abs/PF:2026.000014}
}
